\documentclass{asl}
\usepackage{a4}

\title{Defining Logical Systems via Algebraic Constraints on Proofs}


\author{Alexander V. Gheorghiu}
\revauthor{Gheorghiu, Alexander V.}

\address{University College London, \\ 
        London WC1E 6BT, UK}
\email{alexander.gheorghiu.19@ucl.ac.uk}

\author{David J. Pym}
\revauthor{Pym, David J.}

\address{University College London, \\ 
London WC1E 6BT, UK}
\address{
Institute of philosophy, University of London, \\
London WC1H 0AR, UK}
\email{d.pym@ucl.ac.uk}

\newtheorem{Theorem}{Theorem}[section]
\newtheorem{Proposition}[Theorem]{Proposition}

\usepackage{etoolbox}
\theoremstyle{definition}
\newtheorem{Definition}[Theorem]{Definition}
\AtEndEnvironment{Definition}{\hfill $\triangleleft$}%

\newtheorem{Corollary}[Theorem]{Corollary}

\newtheorem{Example}[Theorem]{Example}
\AtEndEnvironment{Example}{\hfill $\Diamond$}%

\usepackage{url}
\usepackage{amssymb}
\usepackage{bussproofs, proof}
\usepackage{cmll}
\usepackage[all]{xy}
\usepackage{txfonts}
\usepackage[utf8]{inputenc}
\usepackage{graphicx}
\usepackage{xcolor}
\usepackage{subcaption}
\usepackage{somedefs}
\usepackage[only,fatsemi,binampersand,bindnasrepma]{stmaryrd}
\usepackage{mathbbol}
\newenvironment{scprooftree}[1]
{\gdef\scalefactor{#1}\begin{center}\proofSkipAmount \leavevmode}
	{\scalebox{\scalefactor}{\DisplayProof}\proofSkipAmount \end{center} }
	\usepackage{mathrsfs}

\usepackage{xcolor}

\newcommand{\lang}[1]{\ensuremath{\mathscr{#1}}}
\newcommand{\alg}[1]{\ensuremath{\mathcal{#1}}}
\newcommand{\logic}[1]{\ensuremath{\mathfrak{#1}}}
\newcommand{\system}[1]{\ensuremath{\mathsf{#1}}}
\newcommand{\set}[1]{\ensuremath{\mathbb{#1}}}
\renewcommand{\ll}{[\![}
\newcommand{\rr}{]\!]}

\newcommand{\ATOMS}{\mathbb{P}}

\DeclareMathSymbol{\fatcomma}{\mathrel}{bbold}{\lq\,}

\newcommand{\proves}[1][]{\simpleproves_{\system{#1}}}
\newcommand{\simpleproves}[1][]{\vdash}

\newcommand{\satisfies}[1][]{\Vdash_{\logic{#1}}}
\newcommand{\uni}{\mathbb{V}}

\newcommand{\wand}{-\mkern-6mu*}
\renewcommand{\epsilon}{\varepsilon}

\renewcommand{\emptyset}{\varnothing}
\newcommand*{\rn}[1]  {\ensuremath{\mathsf{#1}}}

\newcommand*{\weak}{\rn{w}}
\newcommand*{\exch}{\rn{e}}

\newcommand*{\cont}{\rn{c}}
\newcommand*{\lrn}[1]{\rn{#1_L}}
\newcommand*{\rrn}[1]{\rn{#1_R}}

\newcommand{\ax}{\mathsf{ax}}
\newcommand*{\taut}{\rn{taut}}

\newcommand{\at}[1]{{\rm #1}}

\newcommand{\munit}{\varnothing_\times}
\newcommand{\aunit}{\varnothing_+}
\newcommand{\I}{\top^{*}}

\newcommand{\metaseq}{\triangleright}
\newcommand{\seq}{\triangleright}
\newcommand{\metaconsequence}{\mathbin{{\scriptstyle \blacktriangleright}}}
\newcommand{\metasat}{\Vvdash}

\newcommand{\sat}{\Vdash}
\newcommand{\entails}{\vDash}

\usepackage{yfonts}
\newcommand{\sem}[1]{\mathfrak{#1}}

\newcommand{\pow}{\wp}

\begin{document}
		\maketitle    

\begin{abstract}
We present a comprehensive programme analysing the decomposition of proof systems for non-classical logics into proof systems for other logics, especially classical logic, using an algebra of constraints. That is, one recovers a proof system for a target logic by enriching a proof system for another, typically simpler, logic with an algebra of constraints that act as correctness conditions on the latter to capture the former; for example, one may use Boolean algebra to give constraints in a sequent calculus for classical propositional logic to produce a sequent calculus for intuitionistic propositional logic. The idea behind such forms of reduction is to obtain a tool for uniform and modular treatment of proof theory and provide a bridge between semantics logics and their proof theory. The article discusses the theoretical background of the project and provides several illustrations of its work in the field of intuitionistic and modal logics. The results include the following: a uniform treatment of modular and cut-free proof systems for a large class of propositional logics; a general criterion for a novel approach to soundness and completeness of a logic with respect to a model-theoretic semantics; and, a case study deriving a model-theoretic semantics from a proof-theoretic specification of a logic.
\end{abstract}

\section{Introduction} \label{sec:introduction}

The general goal of this paper is to provide a unifying meta-level framework for studying logics. To this end, we introduce a framework in which one can represent the \emph{reasoning} in a logic, as captured by a concept of proof for that logic, in terms of the reasoning within another logic through an algebra of constraints --- as a slogan,
\[
\mbox{
Proof in \system{L} = Proof in \system{L}$'$ + Algebra of Constraints $\mathcal{A}$
}
\]
We shall refer back to this slogan often and will use the following abbreviated form:
\[
\system{L} = \system{L'} \oplus \mathcal{A}
\]
The $\oplus$ is not formal--- that is, $\system{L'}\oplus \mathcal{A}$ may denote any one of several ways of applying constraints $\mathcal{A}$ to $\system{L}'$.

Such decompositions of $\system{L}$ into $\system{L}'$ and $\mathcal{A}$ allow us to study the metatheory of the former by analyzing the latter. This is advantageous when the latter is typically simpler in some desirable way --- for example, it may relax the side conditions on the use of specific rules --- which facilitates the study of the original logic of interest. There are already some examples of such relationships within the literature --- they are discussed below. The framework herein provides a general view of the phenomena and provides an umbrella for these seemingly disparate cases. Importantly, there is no guarantee that the constraints will be solvable algorithmically as it depends on how complex the algebra is, but even relatively simple algebras can have a dramatic impact; for example, we illustrate how constraint systems using Boolean algebra, which admits solvers, are helpful for the study of proof-search and meta-theory in substructural and intuitionistic logics.

The decompositions expressed by the slogan above may be iterated in valuable ways; that is, it is possible to decompose $\system{L}'$ in the slogan above. Each time we do such a decomposition, the combinatorics of the proof system become more tractable as more and more is delegated to the algebra. Eventually, the combinatorics becomes as simple as possible, and one recovers something with all the flexibility of the proof theory for classical logic. Thus, we advance the view that, in general, classical logic forms a combinatorial core of syntactic reasoning since its proof theory is comparatively relaxed --- that is, possibly after iterating decompositions of the kind above, one eventually witnesses the following:
\[
\mbox{
Proof in \system{L} = Classical Proof + Algebra of Constraints $\mathcal{A}$
}
\]
The view of classical logic as the core of logic has, of course, been advanced before --- see, for example, Gabbay~\cite{Gabbay1993}.

Using techniques from universal algebra, we define the algebraic constraints by a theory of first-order classical logic; for example, we may define Boolean algebra by its axiomatization --- see Section~\ref{sec:ex:bi}. We then enrich rules of a system $\system{L}$ with expressions from $\mathcal{A}$ to express rules of another system $\system{L}'$ --- for example, by using Boolean algebra for the constraints, we may express the single-conclusioned $\rrn \lor$-rule from Gentzen's \system{LJ}~\cite{Gentzen1969} with the combinatorics of the multiple-conclusioned $\rrn \lor$-rule from Gentzen's \system{LK}~\cite{Gentzen1969},
\[
\infer{\Gamma \proves \phi \lor \psi}{\Gamma \proves \phi\cdot x, \psi \cdot\bar{x}}
\]
--- one recovers the single-conclusion condition by assigning the variable $x$ to $0$ or $1$ in the Boolean algebra and evaluating the formula accordingly: one keeps formulae that carry a $1$ and deletes formulae which carry $0$. A system of rules enriched in this way is called a \emph{constraint system}. Detailed examples are below.

We consider two kinds of relationships the constraint systems may have with the logic of interest. 
A constraint system is \emph{sound and complete} when the evaluation of construction from the constraints system concludes a sequent iff that sequent is valid in the logic. A stronger relationship is \emph{faithfulness and adequacy}:
\begin{itemize}
    \item[-] (\emph{Faithful}) The evaluation of a construction from the constraint system is a proof in the logical system of interest.
    \item[-] (\emph{Adequate}) Every proof in the logical system of interest is the evaluation of some construction from the constraint system.
\end{itemize}
Both relationships are important, as illustrated by the examples below.

The point of constraint systems is that they allow us to study the metatheory of the logic of interest. There are two principal such activities presented in this paper: first, one may use constraint systems to study questions of proof-search (i.e., how one constructs proofs) in the logic of interest; second, they may be used to bridge the gap between the proof theory and model theory of a logic. On the latter use, constraint systems allow a novel approach to soundness and completeness proofs which bypasses truth-in-a-model and term-model constructors; furthermore, they give a principled way of generating a correct-by-design model-theoretic semantics for the logic of interest by analyzing a proof system for it, making essential use of algebraic constraints and the aforementioned decomposition to classical logic. 

Of course, the idea that one may use labels to internalize the semantics of logics within proof systems has taken several forms and goes back as far as Kanger~\cite{kanger1957provability}. It underpins a systematic development of analytic tableaux (see, for example, Fitting ~\cite{Fitting1983,Fitting1998}, Catach~\cite{catach1991tableaux}, Massacci~\cite{massacci1994strongly}, Baldoni~\cite{baldoni1998normal}, Docherty and Pym~\cite{Docherty2018,docherty2019non,Docherty2019}, and Galmiche and M\'ery~\cite{galmiche2002resource,galmiche2003semantic,Galmiche2005,galmiche2010tableaux}), natural deduction systems (see, for example, Simpson~\cite{simpson1994proof}, and Basin et al.~\cite{basin1998natural}), sequent calculi (see, for example, Mints~\cite{mints1968cut,mints1997indexed}, Vigan\`o~\cite{vigano2013labelled}, Kushida and Okada~\cite{kushida2003proof}). Particularly significant within this stream are the relational calculi studied by  Negri~\cite{negri1998cut,negri2003contraction,Negri2005}.

The notion of constraint system presented in this paper is closely related to Gabbay’s \emph{Labelled Deductive Systems} (LDSs)~\cite{Gabbay1996} --- see also Russo~\cite{Russo1996}. However, the paper deviates from the established theory of LDSs in two fundamental ways: First, one may choose any syntactic structure in the grammar of the object-logic (e.g., data composed of formulae, such as sets, multisets, bunches), not just formulae, to annotate; second, the labels do not only express additional information but have an action on the structure. Note, there are other proof systems in the literature in which one labels data composed of formulae --- see, for example, Marx et al.~\cite{marx2000mosaic}. Consequently, more subtle examples are available, not otherwise captured by LDSs.

In summary, two main ideas are developed in this paper: a meta-logic for studying object-logics and algebraic constraints for studying the metatheory of a logic. Both ideas are present elsewhere in the literature and have been studied for different logics but have yet to be formalized and uniformed. This paper aims to provide a general and uniform framework for analyzing and understanding the proof theory of logics represented in this way.

The paper begins in Section~\ref{sec:ex:bi} with an example of a constraint system already in the literature. It continues in Section~\ref{sec:background} with the background and notation required for the general treatment throughout the rest of the paper. Constraint systems are defined in general in Section~\ref{sec:constraints}, where the correctness properties of soundness, completeness, faithfulness, and adequacy are also discussed. In Section~\ref{sec:ex:relationalcalculi}, we study relational calculi as a general class of constraint systems and study an approach to proving soundness and completeness which works with validity directly. We give a concrete illustration of the approach applied to intuitionistic propositional logic (IPL) in Section~\ref{sec:ex:ipl} --- specifically, by using constraint systems, we derive the model-theoretic semantics given by Kripke~\cite{Kripke1963} from \system{LJ}, which is sound and complete by construction. In Section~\ref{sec:ext:fol}, we consider the treatment of first-order logics with constraints, the rest of the paper being restricted to propositional logics. The paper concludes in Section~\ref{sec:conclusion} with a summary and a discussion of future research.

\section{Example: Resource-distribution via Boolean Constraints} \label{sec:ex:bi}

In this section, we summarize the resource-distribution via Boolean constraints (RDvBC) mechanism, which was introduced by Harland and Pym~\cite{Harland1997,Pym2001dis} as a tool for reasoning about the context-management problem during proof-search in logics with multiplicative connectives, such as Linear Logic (LL) and the logic of Bunched Implications (BI). It is the original example of a decomposition of a proof system in the sense of this paper, as explained at the end of the section. We present RDvBC to motivate the abstract technical work in Section~\ref{sec:constraints} for the general approach. We concentrate on the case of BI to indicate the scope of the approach and the challenges involved in setting it up because the logic operates over quite a subtle data structure --- bunches.

\subsection{The Logic of Bunched Implications}

One may regard BI as the free combination (i.e., the fibration --- see Gabbay~\cite{Gabbay1993}) of intuitionistic propositional logic (IPL), with connectives $\land, \lor, \to, \top, \bot$, and intuitionistic multiplicative linear logic (IMLL), with connectives $*, \wand, \I$. 
Let $\set{A}$ be a set of atomic propositions. The following grammar generates the formulae of BI:
	\[
		\phi ::= \at{p} \in \set{A} \mid \top \mid \bot \mid \I \mid \phi \land \psi \mid \phi \lor \psi \mid \phi \to \phi \mid \phi * \phi \mid \phi \wand \phi
		\]
The set of all formulae is $\set{FORM}$. 
		
A distinguishing feature of BI is that it has two primitive implications, $\to$ and $\wand$, each corresponding to a different context-former, $\fatsemi$ and $\fatcomma$, representing the two conjunctions $\land$ and $*$, respectively. As these context-formers do not commute with each other, though individually they behave as usual, contexts in BI are not one of the familiar structures of lists, multisets, or sets. Instead, its contexts are \emph{bunches}  --- a term that derives from the relevance logic literature (see, for example, Read~\cite{Read1988}). The set of bunches $\set{BUNCH}$ is defined by the following grammar:
		\[
		\Gamma ::= \phi \in \set{FORM} \mid \aunit \mid \munit \mid \Gamma \fatsemi \Gamma \mid \Gamma \fatcomma \Gamma
		\]		

A bunch $\Delta$ is a sub-bunch of a bunch $\Gamma$ iff $\Delta$ is a sub-tree of $\Gamma$. One may 
write $\Gamma(\Delta)$ to mean that $\Delta$ is a sub-bunch of $\Gamma$. The operation $\Gamma[\Delta \mapsto \Delta']$ --- abbreviated to 
$\Gamma(\Delta')$ where no confusion arises --- is the result of replacing the occurrence of $\Delta$ by $\Delta'$. 

Bunches have similar structural behaviour to the more familiar data-structures used for contexts in logic (e.g., lists, multisets, sets, etc.). We define this behaviour explicitly by means of an equivalence relation called \emph{coherent equivalence}. Two bunches $\Gamma,\Gamma' \in \set{BUNCH}$ are coherently equivalent when $\Gamma \equiv \Gamma'$, where $\equiv$ is the smallest relation satisfying: 
	\begin{itemize}
		\item[-] commutative monoid equations for $\fatsemi$ with unit $\aunit$
		\item[-] commutative monoid equations for $\fatcomma$ with unit $\munit$
		\item[-] coherence; that is, if $\Delta\equiv\Delta' $, then $\Gamma(\Delta) \equiv \Gamma(\Delta')$. 
	\end{itemize} 

A sequent in BI is a pair $\Gamma \seq \phi $ in which $\Gamma$ is a bunch, and $\phi$ is a formula. We use $\seq$ as a pairing symbol defining sequents to distinguish it from the judgment $\proves$ that asserts that a sequent is a consequence of BI. We may characterize the consequence judgment $\proves$ for BI by provability in the sequent calculus \system{LBI} in Figure~\ref{fig:lbi} (see Pym~\cite{Pym2003}). That is, $\Gamma \proves \phi$ iff there is an $\system{LBI}$-proof of $\Gamma \seq \phi$.

\begin{figure}
\hrule
\vspace{4mm}
\centering
$
\infer[\rn{taut}]{\phi \metaseq \phi}{}
\qquad
\infer[\lrn\bot]{\Gamma(\bot) \metaseq \phi}{}
\qquad
\infer[\rrn\I]{\munit \metaseq \I}{}
\qquad
\infer[\rrn\top]{\aunit \metaseq \top}{}
$
\\[1.5ex]
$
\infer[\lrn\wand]{\Gamma(\Delta\fatcomma\Delta'\fatcomma\phi \wand \psi) \metaseq \chi}{\Delta' \metaseq \phi & \Gamma(\Delta\fatcomma\psi) \metaseq \chi}
\qquad
\infer[\rrn\wand]{\Delta \metaseq \phi \wand \psi}{\Delta\fatcomma\phi \metaseq \psi}
$
\\[1.5ex]
$
\infer[\lrn\ast]{\Delta(\phi * \psi) \metaseq \chi}{\Delta(\phi\fatcomma\psi) \metaseq \chi}
\qquad
\infer[\rrn\ast]{\Delta\fatcomma\Delta' \metaseq \phi*\psi}{\Delta \metaseq \phi & \Delta' \metaseq \psi}
\qquad
\infer[\lrn\I]{\Delta(\top^*) \metaseq \chi}{\Delta(\munit)\metaseq \chi}
$
\\[1.5ex]
$
\infer[\lrn\land]{\Delta(\phi \land \psi) \metaseq \chi}{\Delta(\phi\fatsemi\psi) \metaseq \chi}
\qquad
\infer[\rrn\land]{\Delta \metaseq \phi\land\psi}{\Delta \metaseq \phi & \Delta \metaseq \psi}
\qquad
\infer[\lrn\top]{\Delta(\top) \metaseq \chi}{\Delta(\aunit)\metaseq \chi}
$
\\[1.5ex]
$
\infer[\lrn\lor]{\Delta(\phi \lor \psi) \metaseq \chi}{\Delta ( \phi  ) \metaseq \chi & \Delta ( \psi ) \metaseq \chi }
\qquad
\infer[\rrn\lor_1]{\Delta \metaseq \phi \lor \psi}{\Delta \metaseq  \phi}
\qquad
\infer[\rrn\lor_2]{\Delta \metaseq \phi \lor \psi}{\Delta \metaseq  \psi}
$
\\[1.5ex]
$
\infer[\lrn\to]{\Gamma(\Delta \fatsemi\phi \to \psi) \metaseq \chi}{\Delta \metaseq \phi & \Gamma(\Delta\fatsemi\psi) \metaseq \chi}		
\qquad
\infer[\rrn\to]{\Delta \metaseq \phi \to \psi}{\Delta\fatsemi\phi \metaseq \psi}			
$
\\[1.5ex]
$
\infer[\weak]{\Delta(\Delta'\fatsemi\Delta'') \metaseq \chi}{\Delta(\Delta') \metaseq \chi}
\qquad
\infer[\exch_{(\Delta \equiv \Delta')}]{\Delta' \metaseq \chi}{\Delta \metaseq \chi}
\qquad
\infer[\cont]{\Delta(\Delta') \metaseq \chi}{\Delta(\Delta'\fatsemi\Delta') \metaseq \chi}
$	
\vspace{4mm}
\hrule
\caption{Sequent Calculus $\system{LBI}$}
			\label{fig:lbi}
		\end{figure}

As bunches are intended to be the syntactic trees of $\set{BUNCH}$ modulo $\equiv$, we may somewhat relax the formal reading of the rules of $\system{LBI}$. The effect of coherent equivalence is, essentially, to render bunches into two-sorted nested multisets --- see Gheorghiu and Marin~\cite{Alex2021}. Therefore, we may suppress brackets for sections of the bunch with the same context-former and apply rules sensitive to context-formers (e.g., $\rrn \ast$) accordingly. For example, any context-former may be used in $\rrn{\ast}$ applied to $\at{p}_1\fatcomma \at{p}_1 \fatcomma \at{p}_3 \seq \at{q}_1 \ast \at{q}_2$; the possibilities are as follows:
\[
\infer[\rrn{\ast}]{\at{p}_1\fatcomma \at{p}_2 \fatcomma \at{p}_3 \seq \at{q}_1 \ast \at{q}_2 }{
   \at{p}_1 \seq \at{q}_1 &  \at{p}_2 \fatcomma \at{p}_3 \seq  \at{q}_2 
}
\qquad
\infer[\rrn{\ast}]{\at{p}_1\fatcomma \at{p}_2 \fatcomma \at{p}_3 \seq \at{q}_1 \ast \at{q}_2 }{
   \at{p}_1 \fatcomma \at{p}_2 \seq \at{q}_1 &  \at{p}_3 \seq  \at{q}_2 
}
\]
This concludes the introduction of BI. In the next section, we apply the RDvBC mechanism to analyze proof-search in \system{LBI}.

\subsection{Resource-distribution via Boolean Constraints}

Proof-search in \system{LBI} is complex because the presence of multiplicative connectives (i.e., $\ast$ and $\wand$) requires deciding how to distribute the formulae (or, rather, sub-bunches) when making reductions.

\begin{Example}\label{ex:rdvbc:lbips}
 The following proof-search attempts differ only in the choice of distribution, but one successfully produces a proof, and the other fails:
\[
    \underbrace{
        \infer[\rrn \ast]{\at{p}\fatcomma \at{q}\fatcomma \at{r} \seq \at{p} \ast (\at{q} \ast \at{r})}{
            \infer[\rn{taut}]{\at{p} \seq \at{p} }{}
            & 
            \infer[\rrn \ast]{\at{q}\fatcomma \at{r} \seq \at{q} \ast \at{ r}}{
                \infer[\rn{taut}]{\at{q} \seq \at{q}}{
                    } 
                & 
                \infer[\rn{taut}]{\at{r} \seq \at{r}}{
                    }
                }
        }        
    }_{\mbox{succeeds}}
    \qquad
    \underbrace{
    \infer[\rrn \ast]{\at{p}\fatcomma \at{q
    }\fatcomma \at{r} \seq \at{p} \ast (\at{ q} \ast \at{r})}{\infer{\at{p}\fatcomma \at{q} \seq \at{ p}}{?}  & \infer{\at{r} \seq \at{q} \ast \at{r}}{?}}
    }_{\text{fails}}
    \]
    How can we analyze the various distribution strategies? This is the question RDvBC addresses. 
\end{Example}

There is substantial literature on intricate rules of inference in multiplicative logics that are used to keep track of the relevant information to enable proof-search. However, they are generally tailored for one particular distribution method --- see, for example, Hodas and Miller~\cite{hodas1994logic,Hodas1994thesis}, Winikoff and Harland~\cite{winikoff1995implementing}, Cervasto~\cite{Cervesato2000}, and Lopez~\cite{Lopez2005}. It is in this context that Harland and Pym~\cite{Harland1997, Pym2001dis} introduced the RDvBC mechanism. The idea is that rather than commit to a particular strategy for managing the distribution, one uses Boolean expressions to express that a resource distribution needs to be made and the conditions it needs to satisfy. 

Before presenting the technical details of RDvBC, we give a heuristic account. Essentially, one assigns a Boolean expression to each formula requiring distribution. Constraints on the possible values of this expression are then generated during the proof-search and propagated up the search tree, resulting in a set of Boolean equations. A successful proof-search in the enriched system will generate a soluble set of equations corresponding to a distribution of formulae across the branches of the structure, and instantiating that distribution results in an actual proof. It remains to give the formal detail. We begin by defining the constraint algebra that delivers RDvBC.

A \emph{Boolean algebra} is a tuple $\mathcal{B} := \langle \set{B}, \{+, \times, \bar{\cdot} \}, \{0,1\} \rangle$ in which $\set{B}$ is a set, $+:\set{B}^2 \to \set{B}$, $\times:\set{B}^2 \to \set{B}$, $\bar{\cdot}:\set{B} \to \set{B}$ be operators on $\set{B}$, and $0,1 \in \set{B}$, satisfying the following conditions for any $a,b,c \in \set{B}$:
\[
\begin{array}{cc}
a+(b+c) = (a+b)+c \quad a\times(b\times c) = (a\times b) \times c & a+b=b+c \quad a \times b = b \times a \\
a+(a\times b) = a \quad a \times (a + b) = a & a+0 = a \quad a \times 1 = a \\
a+(b \times c) = (a +b) \times (a+c) & a \times (b+c) = (a \times b)+(a \times c) \\
a+\bar{a} = 1 & a \times \bar{a} = 0
\end{array}
\]
A \emph{presentation} of the Boolean algebra is a first-order classical logic with equality for which the Boolean algebra is a model. We use the following, in which $X$ is a set of \emph{variables}, $e$ are \emph{Boolean expressions}, and $\phi$ are \emph{Boolean formulae}:
\[
		e  ::=  x \in X \mid e+e \mid e \times e \mid \bar{e} \mid 0 \mid 1 \quad 	\phi ::=  (e = e) \mid \phi \with \phi \mid \phi \parr \phi \mid \neg \phi \mid \forall x \phi \mid \exists x \phi
\]
The symbols $\with$ and $\parr$ are used as classical conjunction and disjunction, respectively.

We are overloading $+$ and $\times$ to be both function-symbols in the term language and their corresponding operators in the Boolean algebra; similarly, we are overloading $0$ and $1$ to be both constants in the term language and the bottom and top element of the Boolean algebra. This is to economize on notation. We may suppress the $\times$ when no confusion arises --- that is, $t_1 \times t_2$ may be expressed $t_1t_2$. For a list of Boolean expressions $V = [e_1,\hdots e_n]$, let $\bar{V}:= [\bar{e}_1,\hdots \bar{e}_n]$; we may write $V=e$ to denote that $V$ is a list containing only $e$. 

Let some presentation of Boolean algebra be fixed. An \emph{annotated BI-formula} is a BI-formula $\phi$ with a Boolean expression $e$, denoted $\phi \cdot e$. The annotation of a bunch $\Gamma$ by a list of Boolean expressions $V$ is defined as follows:
\begin{itemize}
    \item[-] if $\Gamma = \gamma$, where $\gamma \in \set{FORM}\cup\{\aunit,\munit\}$ and $V=[e]$, then $\Gamma \cdot V := \gamma\cdot e$;
    \item[-] if $\Gamma = (\Delta_1 \fatsemi \Delta_2)$, and $V = [e]$, then $\Gamma \cdot V: = (\Delta_1 \fatsemi \Delta_2) \cdot e$;
    \item[-] if $\Gamma = (\Delta_1 \fatcomma \Delta_2)$, and $V$ is the concatenation of $V_1$ and $V_2$, then $\Gamma \cdot V: = (\Delta_1\cdot V_1 \fatsemi \Delta_2 \cdot V_2)$.
\end{itemize}
For example, $p \fatcomma (q\fatsemi r) \cdot [x,y] := p \cdot x \fatcomma (q\fatsemi r) \cdot y$. Notably, the annotation only acts on the top-level of multiplicative connectives and treats everything below (e.g., additive sub-bunches) as formulae. This makes sense as all of the distributions in $\system{LBI}$ take place at this level of the bunch. 


This concludes the technical overhead required to define the RDvBC mechanism for BI. Roughly, Boolean constraints are used to mark the multiplicative distribution of formulae. The mechanism is captured by proof-search in the sequent calculus $\system{LBI}_\mathcal{B}$ comprised of the rules in Figure~\ref{fig:lbib}, in which $V$ is a list of Boolean variables that do not appear in any sequents present in the tree and labels that do not change are suppressed. The same names are used for rules in $\system{LBI}_\mathcal{B}$ and $\system{LBI}$ to economize on notation.

\begin{figure}
\hrule
\vspace{4mm}
\centering
$
\infer[\rn{taut}]{\phi\cdot y \fatcomma \Delta\cdot V \metaseq \phi\cdot 1}{y=1 \land V = 0 }
\qquad
\infer[\lrn\bot]{\bot\cdot y \fatcomma \Delta \cdot V \metaseq \phi}{y=1\land V = 0 }
$
\\[1.5ex]
$
\infer[\rrn\I]{\munit\cdot y \fatcomma \Delta \cdot V \metaseq \I}{y=1 \with V=0 }
\qquad
\infer[\rrn\top]{\aunit \cdot y \fatcomma \Delta \cdot V \metaseq \top}{y=1 \with V =0 }
$
\\[1.5ex]
$
\infer[\lrn\wand]{\Gamma(\Delta\fatcomma\phi \wand \psi\cdot e) \metaseq \chi}{\Delta\cdot V \metaseq \phi\cdot e& \Gamma(\Delta \cdot \bar{V}\fatcomma\psi\cdot e) \metaseq \chi & e=1}
\qquad
\infer[\rrn\wand]{\Delta \metaseq (\phi \wand \psi)\cdot e}{\Delta\fatcomma\phi\cdot e \metaseq \psi\cdot e & e =1}
$
\\[1.5ex]
$
\infer[\lrn\ast]{\Delta((\phi * \psi)\cdot e) \metaseq \chi}{\Delta(\phi\cdot e\fatcomma\psi\cdot e) \metaseq \chi & e=1}
\qquad
\infer[\rrn\ast]{\Delta\fatcomma\Delta' \metaseq \phi*\psi}{\Delta \cdot V \metaseq \phi & \Delta' \cdot \bar{V} \metaseq \psi}
\qquad
\infer[\lrn\I]{\Delta(\top^*\cdot e) \metaseq \chi}{\Delta(\munit \cdot e)\metaseq \chi & e=1}
$
\\[1.5ex]
$
\infer[\lrn\land]{\Delta((\phi \land \psi)\cdot e) \metaseq \chi}{\Delta(\phi\cdot e\fatsemi\psi\cdot e) \metaseq \chi & e=1}
\qquad
\infer[\rrn\land]{\Delta \metaseq \phi\land\psi}{\Delta \metaseq \phi & \Delta \metaseq \psi}
\qquad
\infer[\lrn\top]{\Delta(\top\cdot e) \metaseq \chi}{\Delta(\aunit \cdot e)\metaseq \chi & e=1}
$
\\[1.5ex]
$
\infer[\lrn\lor]{\Delta((\phi \lor \psi)\cdot e) \metaseq \chi}{\Delta ( \phi \cdot e ) \metaseq \chi & \Delta ( \psi \cdot e) \metaseq \chi & e=1 }
\qquad
\infer[\rrn\lor_1]{\Delta \metaseq \phi \lor \psi}{\Delta \metaseq  \phi}
\qquad
\infer[\rrn\lor_2]{\Delta \metaseq \phi \lor \psi}{\Delta \metaseq  \psi}
$
\\[1.5ex]
$
\infer[\lrn\to]{\Gamma(\Delta \fatsemi(\phi \to \psi)\cdot e) \metaseq \chi}{\Delta \metaseq \phi & \Gamma(\Delta\fatsemi\psi\cdot e) \metaseq \chi & e=1}		
\qquad
\infer[\rrn\to]{\Delta \metaseq \phi \to \psi }{\Delta\fatsemi\phi \cdot e \metaseq \psi & e=1}			
$
\\[1.5ex]
$
\infer[\weak]{\Gamma(\Delta\cdot e\fatsemi\Delta \cdot e) \metaseq \chi}{\Gamma(\Delta\cdot e) \metaseq \chi & e=1}
\qquad
\infer[\exch_{(\Delta \equiv \Delta')}]{\Delta' \metaseq \chi}{\Delta \metaseq \chi}
\qquad
\infer[\cont]{\Gamma(\Delta\cdot e) \metaseq \chi}{\Gamma(\Delta\cdot e\fatsemi\Delta\cdot e) \metaseq \chi & e=1}
$	
\vspace{4mm}
\hrule
\caption{Sequent Calculus $\system{LBI}_\alg{B}$}
			\label{fig:lbib}
		\end{figure}

An $\system{LBI}_\mathcal{B}$-reduction is a tree constructed by applying the rules of $\system{LBI}_\mathcal{B}$ reductively, beginning with a sequent in which each formula is annotated by $1$. 

\begin{Example}\label{ex:rdvbc:lbiBreduction}
The following is an $\system{LBI}_\mathcal{B}$-reduction $\mathcal{D}$:
\[
      \infer[\rrn \ast]{(\at{p}\cdot 1)\fatcomma (\at{q}\cdot 1)\fatcomma (\at{r}\cdot 1) \seq \at{p} \ast (\at{q} \ast \at{r})\cdot 1}{
            \infer[\rn{taut}]{(\at{p}\cdot x_1) \fatcomma (\at{q}\cdot x_2) \fatcomma (\at{r}\cdot x_3) \seq \at{p}\cdot 1 }{(x_1=1) \land (x_2=0) \land (x_3=0)}
            & 
          \mathcal{D}'
        }  
\]
--- the sub-tree $\mathcal{D}'$ is the following:
\[
      \infer[\rrn \ast]{(\at{p}\cdot \bar{x}_1) \fatcomma (\at{q}\cdot \bar{x}_2) \fatcomma (\at{r}\cdot \bar{x}_3) \seq \at{q} \ast \at{ r}\cdot 1}{
                \infer[\rn{taut}]{(\at{p}\cdot \bar{x}_1y_1) \fatcomma (\at{q}\cdot \bar{x}_2 y_2) \fatcomma (\at{r}\cdot \bar{x}_3 y_3) \seq \at{q}}{ (\bar{x}_1y_1 =0) \land  (\bar{x}_2 y_2 = 1) \land (\bar{x}_3 y_3 = 0)
                    } 
                & 
                \infer[\rn{taut}]{\at{p}\cdot (\bar{x}_1\bar{y}_1) \fatcomma \at{q}\cdot (\bar{x}_2\bar{y}_2) \fatcomma \at{r}\cdot (\bar{x}_3\bar{y}_3) \seq \at{r}}{
                    (\bar{x}_1\bar{y}_1 = 0) \land
                    (\bar{x}_2\bar{y}_2 = 0) \land
                    (\bar{x}_3\bar{y}_3 =1)
                    }
                }
\]
\end{Example}

Having produced an $\system{LBI}_\mathcal{B}$-reduction, if the constraints are consistent, they determine the variables' interpretations to satisfy the constraints. Such interpretations $I$ induce a valuation $\nu_I$ that acts on formulae by keeping formulae whose label evaluates to $1$ and deleting (i.e., producing the empty-string $\epsilon$) for formulae whose label evaluates to $0$; that is, let $\phi$ be a BI-formula and $e$ a Boolean expression,
\[
\nu_I(\phi \cdot e) := 
\begin{cases}
\phi & \text{ if } I(e) = 1 \\
\epsilon  & \text{ if } I(e) = 0
\end{cases}
\]
A valuation extends to sequents by acting on each formulae occurring in it, and it extends to $\system{LBI}_\mathcal{B}$-reductions by acting on each sequent occurring in it and removing the constraints.

\begin{Example}[Example~\ref{ex:rdvbc:lbiBreduction} cont'd] \label{ex:rdvbc:lbiBinstantiated}
 The constraints on $\mathcal{D}$ are satisfied by any interpretation $I(z)=1$ for $z \in \{x_1,y_2\}$ and $I(z)=0$ for $z \in \{x_2, x_3,y_1,y_3\}$. For any such $I$, the tree $\nu_I(D)$ is as follows:
\[
 \infer[\rrn \ast]{\at{p}\fatcomma \at{q}\fatcomma \at{r} \seq \at{p} \ast (\at{q} \ast \at{r})}{
            \infer[\rn{taut}]{\at{p} \seq \at{p} }{}
            & 
            \infer[\rrn \ast]{\at{q}\fatcomma \at{r} \seq \at{q} \ast \at{ r}}{
                \infer[\rn{taut}]{\at{q} \seq \at{q}}{
                    } 
                & 
                \infer[\rn{taut}]{\at{r} \seq \at{r}}{
                    }
                }
        }        
\]
This is the successful derivation in $\system{LBI}$ in Example~\ref{ex:rdvbc:lbips}. According to the constraints, a distribution strategy results in a successful proof-search just in case it sends only the first formula to the left branch. 
\end{Example}

Harland and Pym~\cite{Harland1997,Pym2001dis} proved that $\system{LBI}_\mathcal{B}$ is \emph{faithful} and \emph{adequate} for $\system{LBI}$ in the following sense:
\begin{itemize}
    \item[-] \emph{Faithfulness}. If $\mathcal{R}$ is an $\system{LBI}_\mathcal{B}$-reduction and $I$ is an interpretation satisfying those constraints, then there is a $\system{LBI}$-proof $\mathcal{D}$ such that $\nu_I(\mathcal{R}) = \mathcal{D}$.
    \item[-] \emph{Adequacy}. If $\mathcal{D}$ is an $\system{LBI}$-proof, then there is a $\system{LBI}_\mathcal{B}$-reduction $\mathcal{R}$ and an interpretation $I$ satisfying its constraints such that $\nu_I(R) = \mathcal{D}$.
\end{itemize}

Recall that we may think of BI as the combination of IPL and IMLL. We may express $\system{LBI}$ as the combination of sequent calculi for these two logics (i.e., $\system{IMLL}$ and $\system{LJ}$, respectively)  --- that is,
\[
\system{LBI} = \system{IMLL} \cup \system{LJ}
\]
The RDvBC outsources the substructurallity in $\system{IMLL}$ to Boolean constraints. Hence, in the form of the slogan of this paper,
\[
\system{IMLL} =\system{LJ}\oplus \mathcal{B}
\]
In this section, we have chosen to study BI (as opposed to just IMLL) to illustrate the modularity of constraint systems. That is, we only have constraints participating actively in part of the sequent calculus for BI, but with the same overall effect since the other part conserves them. Abusing the slogan somewhat, we may express the work of this section as follows: 
\[
\system{LBI} = (\system{LJ}\oplus \mathcal{B})\cup\system{LJ}
\]

This paper aims to study the use of constraints in proof systems in general. It provides conditions under which they exist and illustrates the use of constraints to study metatheory. 

\section{Background} \label{sec:background}

We have two things to set up to give a general presentation of algebraic constraint systems: algebra and propositional logic. The former is captured by first-order classical logic (FOL) (e.g., as Boolean algebra is captured by its axiomatization in Section~\ref{sec:ex:bi}), and the latter is given by a general account of propositional logic as a propositional language together with a consequence relation. There are many presentations of these subjects within the literature; therefore, to avoid confusion, in this section, we define them as they are used in this paper. Importantly, this section introduces much notation used in the rest of the paper. As we wish to reserve traditional symbols such as $\proves$ and $\to$ for the object-logics, we will use $\metaconsequence$ and $\Rightarrow$ for the meta-logic. In both cases, we use the symbol $\seq$ as the sequents symbol, regarding $\proves$ and $\metaconsequence$ as \emph{consequence} relations. 

\subsection{First-order Classical Logic} \label{sec:background:fol}

This section presents first-order classical logic (FOL), which we use to define what we mean by algebra in \emph{algebraic} constraints. We assume familiarity with FOL, so give a terse (but complete) summary to keep the paper self-contained. In particular, we assume familiarity with proof theory for FOL --- as covered in, for example, Troelstra and Schwichtenberg~\cite{troelstra2000basic}, and Negri and von Plato~\cite{negri2001}.

\begin{Definition}[First-order Language] \label{def:FOL:Alphabet}
An alphabet is a tuple $\lang{A} := \langle \set{R}, \set{F}, \set{K}, \set{V} \rangle$ in which $\set{R}$, $\set{F}$, $\set{K}$, and $\set{V}$ are pairwise disjoint countable sets of symbols, and each element of $\set{R}$, $\set{F}$ and $\set{K}$ has a fixed arity. The terms, atoms, and well-formed formulae (wffs) of an alphabet are as follows:
\begin{itemize}
    \item[-] The set $\set{TERM}(\lang{A})$ of \emph{terms} from $\lang{A} $ is the smallest set containing $\set{K}$ and $\set{V}$ such that, for any $F \in \set{F}$, if F has arity $n$ and $T_1,...,T_n \in \set{TERM}(\lang{A} )$, then $F(T_1,...,T_n) \in \set{TERM}(\lang{A})$
    \item[-] The set $\set{ATOMS}(\set{A})$ is set of strings $R(T_1,...,T_n)$ such that $R \in \set{R}$ has arity $n$ and $T_1,...,T_n \in \set{TERM}(\lang{A})$
    \item[-] The set $\set{WFF}(\lang{A})$ of \emph{formulae} from $\lang{A}$ is defined by the following grammar, in which $X \in \set{V}$:
\[
\Phi := A \in  \set{ATOMS}(\set{A}) \mid \Phi \Rightarrow \Phi \mid \Phi \with \Phi \mid \Phi \parr \Phi \mid \bot  \mid \forall X \Phi \mid \exists X \Phi
\]
\end{itemize}

    



\end{Definition}

The symbols $\Rightarrow$, $\with$, $\parr$, and $\bot$ are \emph{implication}, \emph{conjunction},  \emph{disjunction}, and \emph{absurdity}, respectively, in FOL. The more traditional symbols, such as $\to$, $\land$, and $\lor$, are reserved for other logics in the paper. We use the usual convention for suppressing brackets; that is,  conjunction ($\with$) and disjunction ($\parr$) bind more strongly than implication ($\Rightarrow$). Moreover, we may use the usual auxiliary terminology for first-order languages (e.g., \emph{sub-formula}, \emph{closed-formula}, \emph{sentence}, etc.) without further explanation. Let be $X$ a variable, $T$ be a term, and $\Phi$ a wff; we write $\Phi[X \mapsto T]$ to denote the result of replacing every free occurrence of $X$ by the term $T$ so that no variable in $T$ becomes bound in $\Phi$ after the substitution. 

\begin{Definition}[First-order Sequent] \label{def:FOL:sequent}
   A first-order sequent is a pair $\Pi\metaseq \Sigma$ in which $\Pi$ and $\Sigma$ are multisets of first-order formulae. 
\end{Definition}

We think of sequents as unjudged statements, hence we use a sequent constructor $\seq$ instead of the consequence relation ($\metaconsequence$). For example, though $\emptyset \seq \emptyset $ is a well-formed sequent in FOL, it is not a consequence of the logic.

\subsubsection{Proof Theory}

One way to characterize FOL --- that is, the consequence relation $\metaconsequence$ --- is by provability in a sequent calculus.

\begin{Definition}[Sequent Calculus \system{G3c}] \label{def:G3c}
   The sequent calculus $\system{G3c}$ is composed of the rules in Figure~\ref{fig:G3c} in which $T$ is a term free for $X$ in $\Phi$ and $Y$ is an eigenvariable.
\end{Definition}
 	 

\begin{figure}[t]
\hrule
  \vspace{4mm}
	\[
	\begin{array}{c@{\qquad}c}
		\infer[\ax]{\Phi, \Pi \metaseq \Sigma, \Phi}{} 
		&
			\infer[ \bot]{\bot, \Pi \metaseq \Sigma}{}
	 \\[1.5ex]
		\infer[\lrn{\with}]{\Phi \with  \Psi,\Pi\metaseq \Sigma}{\Phi,\Psi,\Pi\metaseq\Sigma}
			& 	\infer[\rrn\with]{\Pi \metaseq \Sigma,\Phi \with \Psi}{\Pi \metaseq\Sigma, \Phi & \Pi\metaseq\Sigma,\Psi}
	 \\[1.5ex]
		\infer[\lrn{\parr}]{\Phi\parr\Psi,\Pi \metaseq \Sigma}{\Phi,\Pi \metaseq\Sigma & \Psi,\Pi \metaseq\Sigma }
	   &
		\infer[\rrn{\parr}]{\Pi \metaseq \Sigma , \Phi \parr  \Psi}{\Pi \metaseq\Sigma,\Phi,\Psi }
	 \\[1.5ex]
		\infer[\lrn \Rightarrow]{\Phi \Rightarrow \Psi,\Pi \metaseq \Sigma}{\Pi \metaseq \Sigma ,\Phi & \Psi,\Pi \metaseq \Sigma}
&
\infer[\rrn \Rightarrow]{\Pi \metaseq \Sigma, \Phi \Rightarrow \Psi}{\Phi,\Pi \metaseq \Sigma,\Psi}
	 \\[1.5ex]
    \infer[\lrn \forall]{\forall X \Phi, \Pi \metaseq \Sigma} {\forall X \Phi, \Phi[X \mapsto T], \Pi \metaseq \Sigma}
    &
    \infer[\rrn \forall]{\Pi \metaseq \Sigma, \forall X \Phi}{\Pi \metaseq \Sigma, \Phi[X \mapsto Y]}
    \\[1.5ex]
     \infer[\lrn \exists]{\exists X \Phi, \Pi \metaseq \Sigma}{\Phi[X \mapsto Y], \Pi \metaseq \Sigma}
    		&
     \infer[\rrn \exists]{\Pi \metaseq \Sigma, \exists X \Phi}{\Pi \metaseq \Sigma, \exists X \Phi, \Phi[X \mapsto T]}
    \end{array}
    \]
    \vspace{4mm}
\hrule
    \caption{Sequent Calculus \system{G3c}}
    \label{fig:G3c}
\end{figure}

We write $\Pi \proves[G3c] \Sigma$ to denote that there is a $\system{G3c}$-proof of $\Pi \seq \Sigma$. Troelstra and Schwichtenberg~\cite{troelstra2000basic} proved that $\system{G3c}$-provability characterizes classical consequence:

\begin{Proposition} \label{lem:snc:G3c}
Let $\Pi$ and $\Sigma$ be multisets of formulae,
\[
\Pi \metaconsequence \Sigma \qquad \mbox{iff}  \qquad \Pi\proves[G3c]\Sigma
\]
\end{Proposition}

We have chosen to use \system{G3c} to characterize FOL, as opposed to other proof systems, because of its desirable proof-theoretic properties --- for example, Troelstra and Schwichtenberg~\cite{troelstra2000basic} have shown that the rules of the calculus are (height-preserving) invertible, and that the following rules are admissible: 
\[
		\infer[\lrn{\weak}]{\Phi,\Pi \metaseq \Sigma}{\Pi\metaseq\Sigma}
		\quad
		\infer[\rrn{\weak}]{\Pi \metaseq \Sigma,\Phi}{\Pi\metaseq\Sigma}
		\quad
		\infer[\lrn{\cont}]{\Phi,\Pi \metaseq \Sigma}{\Phi,\Phi,\Pi\metaseq\Sigma}
		\quad
		\infer[\rrn{\cont}]{\Pi \metaseq \Sigma,\Phi}{\Pi\metaseq\Sigma,\Phi,\Phi}
\]

\subsubsection{Model-theoretic Semantics.}

Another way to characterize FOL is by validity in its model-theoretic semantics. As mentioned above, we assume familiarity with the subject and therefore give a terse but complete account of definitions to keep the paper self-contained.

\begin{Definition}[First-order Structure]
A first-order structure is a tuple $\mathcal{S} = \langle \mathbb{U}, \set{R}, \set{F}, \set{K} \rangle$ in which \set{U} is a countable set of elements, $\set{K} \subseteq \set{U}$, $\set{F}$ is a countable set of operators on \set{U} (i.e., endomorphisms $f:\mathbb{U}^{n} \to \mathbb{U}$, for finite $n$), and $\set{R}$ is a countable set of relations on $\set{U}$.
\end{Definition}

\begin{Definition}[Interpretation]
    Let $\mathcal{S}:= \langle \mathbb{U}, \set{R}, \set{F}, \set{K} \rangle$ be a first-order structure, and let $\lang{A} := \langle \set{R}', \set{F}', \set{K}', \set{V} \rangle$  be an alphabet. An interpretation of $\lang{A}$  in $\mathcal{S}$ is a function $\llbracket - \rrbracket$ satisfying the following:
    \begin{itemize}
        \item[-] if $X \in \set{V}$, then $\llbracket X \rrbracket \in \mathbb{U}$;
        \item[-] if $C \in \set{K}'$, then $\llbracket C \rrbracket \in \mathbb{K}$;
        \item[-] if $F \in \set{F}'$, then $\llbracket F \rrbracket \in \set{F}$, and the arity of $\llbracket F \rrbracket$ is the arity of $F$;
        \item[-] if $R \in \set{R}'$, then $\llbracket R \rrbracket \in \set{R}$, and the arity of $\llbracket R \rrbracket$ is the arity of $R$. 
    \end{itemize}
\end{Definition}

We may write $\llbracket-\rrbracket:\lang{A} \to \mathcal{S}$ to denote that $\llbracket - \rrbracket$ is an interpretation of $\lang{A}$ in $\mathcal{S}$. Interpretations extend to terms as follows:
\[
\llbracket F(T_1,...,T_n) \rrbracket:=\llbracket F \rrbracket(\llbracket T_1 \rrbracket,...,\llbracket T_n \rrbracket)
\]

In this paper, we use the term \emph{abstraction} for what is traditionally referred to as a \emph{model}. This is to avoid confusion as we consider the semantics of various propositional logics in subsequent sections, where the term \emph{model} will be significant.

\begin{Definition}[Abstraction]
    An abstraction of an alphabet $\lang{A}$ is a pair $\mathfrak{A}:=\langle \mathcal{S}, \llbracket - \rrbracket \rangle$ in which $\mathcal{S}$ is a structure and $ \llbracket-\rrbracket:\lang{A} \to \mathcal{S}$.
\end{Definition}

\begin{Definition}[Truth in an Abstraction]
    Let $\lang{A}$ be an alphabet, let $\phi$ be a formula over $\lang{A}$, and let $\mathfrak{A} = \langle \alg{S}, \llbracket - \rrbracket  \rangle$ be an abstraction of $\lang{A}$. The formula $\phi$ is true in $\mathfrak{A}$ iff $\mathfrak{A} \metasat \phi$, which is defined inductively by the clauses in Figure~\ref{fig:model}.
\end{Definition}

\begin{figure}[t]
\hrule
  \vspace{4mm}
    \[
    \begin{array}{lcl}
         \mathfrak{M} \metasat R(T_1,...,T_n) & \qquad \mbox{iff} \qquad & \langle \llbracket T_1 \rrbracket,...,\llbracket T_n \rrbracket \rangle  \in \llbracket P \rrbracket \\[1ex]
         \mathfrak{M} \metasat \Phi \Rightarrow \Psi & \qquad \mbox{iff} \qquad &  \mbox{not $\mathfrak{M} \metasat \Phi$ or $ \mathfrak{M} \metasat \Psi$} \\[1ex]
         \mathfrak{M} \metasat \Phi \with \Psi & \qquad \mbox{iff} \qquad &  \mbox{$\mathfrak{M} \metasat \Phi$ and $\mathfrak{M} \metasat \Psi$} \\[1ex]
         \mathfrak{M} \metasat \Phi \parr \Psi & \qquad \mbox{iff} \qquad &  \mbox{$\mathfrak{M} \metasat \Phi$ or $\mathfrak{M} \metasat \Psi$} \\[1ex]
         \mathfrak{M} \metasat \bot & \qquad \mbox{iff} \qquad &  \mbox{never} \\[1ex]
         \mathfrak{M} \metasat \forall X\Phi & \qquad \mbox{iff} \qquad &  \mbox{$\mathfrak{M} \metasat \Phi[X \mapsto T]$ for any $T \in \set{TERM}(\lang{A})$}\\[1ex]
         \mathfrak{M} \metasat \exists X\Phi & \qquad \mbox{iff} \qquad &  \mbox{$\mathfrak{M} \metasat \Phi[X \mapsto T]$ for some $T \in \set{TERM}(\lang{A})$ }
    \end{array}
    \]
\vspace{4mm}
\hrule
    \caption{Truth in an Abstraction}
    \label{fig:model}
\end{figure}

We may extend the truth of formulae in an abstraction to the truth of (multi-)sets of formulae by requiring that all the elements in the set are true in the abstraction --- that is, if $\mathfrak{A}$ is a model and $\Pi$ is a multiset of formulae,
\[
\mathfrak{A} \metasat \Pi \qquad \mbox{ iff } \qquad \mathfrak{A} \metasat \Phi \mbox{ for every 
} \Phi \in \Pi
\]

G\"odel~\cite{Godel} --- see also van Dalen~\cite{vanDalen} --- proved that abstractions characterize FOL:

\begin{Proposition} \label{lem:snc:classical}
Let $\Pi$ and $\Sigma$ be multisets of formulae,
\[
\begin{array}{lcl}
\Pi \metaconsequence \Sigma & \mbox{iff} & 
\mbox{ for any abstraction $\mathfrak{A}$, if $\mathfrak{A} \metasat \phi$ for any $\phi \in \Pi$,} \\ & & \mbox{then there is $\psi \in \Sigma$ such that $\mathfrak{A} \metasat \psi$}
\end{array}
\]
\end{Proposition}

This concludes the summary of FOL. 

\subsection{Propositional Logic} \label{sec:background:propositional}

There is no consensus in the literature on what propositional logic means. This paper uses a relatively broad definition that captures the most common propositional logics (e.g., classical propositional logic, intuitionistic propositional logic, modal logics, linear logics, bunched logics, etc.). We include context-formers explicitly as a part of the language of propositional logics. More precisely, we include \emph{data}-formers as they may appear on either the left or right of sequents for the propositional logics, and we use `context' to refer to the left of sequents. This is useful for two reasons: first, it enables us to move between the propositional logics without ambiguity; second, it enables us to handle propositional logics that are expressed in terms of more complex data structures of formulae than lists, multisets, or sets, such as the family of relevance logics (see, for example, Read~\cite{Read1988}) and the family of bunched logics (see, for example, work by Docherty, O'Hearn and Pym~\cite{ohearn1999logic,Pym2003,Docherty2019}). Throughout, we give a running example of normal modal logics, which relates the work of this paper to that of Negri~\cite{Negri2005}.

\begin{Definition}[Propositional Alphabet] \label{def:Prop:Alphabet}
    A propositional alphabet is a tuple of three elements $\lang{P}:= \langle \set{A} , \set{O} , \set{C} \rangle$ such that $\set{A}$, $\set{O}$, and $\set{C}$ are pairwise disjoint sets of symbols such that $\set{A}$ is countable and $\set{O}$ and $\set{C}$ are finite. The symbols in $\set{O}$ and $\set{C}$ have a fixed arity.
\end{Definition}

The elements of  $\set{A}$ are \emph{atomic propositions}, the elements of $\set{O}$ are \emph{operators}, and the elements of $\set{C}$ are \emph{data-constructors}. We use the term \emph{operators} to subsume `connectives' and `modalities' in the traditional terminology. Similarly, we use the term `data-constructor' as a neutral term for what is sometimes called a `context-former' as we shall have data both on the left and right of sequents with possibly different constructors and reserve the term `context' for the left-hand side. 

\begin{Definition}[Formula, Data, Sequent] \label{def:Prop:Formula_Data}
    Let $\lang{P} := \langle \set{A}, \set{O}, \set{C} \rangle$ be a propositional alphabet. The set of formulae, data, and sequents from $\lang{P}$ are as follows:
    \begin{itemize}
        \item[-]  The set of propositional formulae $\set{FORM}(\lang{P})$ is the smallest set containing $\set{A}$ such that, for any $\phi_1,...,\phi_k \in \set{FORM}(\lang{P})$ and $\circ \in  \set{O}$, if $\circ$ has arity $n$, then $\circ(\phi_1,...,\phi_n) \in \set{FORM}(\lang{P})$
        \item[-] The set $\set{DATA}({\lang{P}})$ is the smallest set containing $\set{FORM}(\lang{P})$ such that, for any $\delta_1,...,\delta_n \in \set{DATA}(\lang{P})$ and $\bullet \in \set{C}$, if $\bullet$ has arity $n$, then $\bullet(\delta_1,...,\delta_n) \in \set{DATA}(\lang{P})$
        \item[-] A $\lang{P}$-sequent is a pair $\Gamma \seq \Delta$ in which $\Gamma, \Delta \in \set{DATA}(\lang{P})$.
    \end{itemize}  
\end{Definition}

\begin{Example} \label{ex:basic:alphabet}
The basic modal alphabet is $\lang{B} = \langle \set{A}, \{\land,\lor, \neg, \square\}, \{\emptyset,\fatcomma, \fatsemi\} \rangle$. The arities of $\land$, $\lor$, $\fatcomma$, and $\fatsemi$ is $2$; the arities of $\neg$ and $\square$ is $1$; and, the arity of $\emptyset$ is $0$. We may write $\phi \supset \psi$ to abbreviate $\neg \phi \lor \psi$. 

Let $\at{p}_1,\at{p}_2,\at{p}_3 \in \set{A}$. Using infix notation, the following are examples of elements from $\set{FORM}(\lang{B})$:
\[
\at{p_3} \qquad (\at{p}_{1} \land \at{p}_2) \qquad  (\at{p}_3 \supset (\at{p}_{1} \land \at{p}_1))
\]
As well as being elements of $\set{FORM}(\lang{B})$, they are also elements in $\set{DATA}(\lang{B})$. 

Another example of an element from $\set{DATA}(\lang{B})$ is the following:
\[
\at{p_3} \fatcomma (\at{p}_3 \supset (\at{p}_{1} \land \at{p}_2))
\]

The following is an example of a $\lang{B}$-sequent:
\[
\at{p}_3 \fatcomma \at{p}_3 \supset (\at{p}_{1} \land \at{p}_2) \seq \at{p}_{1} \land \at{p}_2
\]
\end{Example}

This completes the definition of the language of a propositional logic generated by an alphabet. What makes language into a logic is a notion of consequence. 

\begin{Definition}[Propositional Logic]
 Let $\lang{A}$ be a propositional alphabet. A propositional logic over $\lang{A}$ is a relation $\proves$ over $\lang{A}$-sequents.
 \end{Definition}

The relation $\proves$ is called the \emph{consequence judgment} of the logic; its elements are \emph{consequences}. We write $\Gamma \proves \Delta$ to denote that the sequent $\Gamma \seq \Delta$ is a consequence. This definition of propositional logic needs more sophistication in many regards --- for example, nothing renders the operators of the alphabet as logical constants --- but the point is not to satisfy the doxastic interpretation of what constitutes a logic. Interesting though that may be, it amounts to refining the current definition. What is given here suffices for present purposes and encompasses the vast array of propositional logics in the literature.

\subsubsection{Proof Theory}
 In this paper, we are concerned about the proof-theoretic characterization of a logic and what it tells us about that logic. Fix a propositional alphabet $\lang{P}$.

\begin{Definition}[Sequent Calculus]
A rule $\rn{r}$ over $\lang{P}$-sequents is a (non-empty) relation on $\lang{P}$-sequents; a rule with arity one is an axiom. A sequent calculus is a set of rules at smallest one of which is an axiom. 
\end{Definition}

Note that $\system{LBI}_\mathcal{B}$ in Section 2 does not contain axioms and is, therefore, not a sequent calculus according to this definition. This is because we regard it as a \emph{constraint system} in which axioms are not necessary --- see Section~\ref{sec:constraints}.

We have not defined rules by rule-figures and do not assume they are necessarily closed under substitution. This allows us to speak of rules with side-conditions --- see, for example, $\exch \in \system{LBI}$ in Section~\ref{sec:ex:bi}. Of course, we will otherwise follow standard conventions --- see, for example, Troelstra and Schwichtenberg~\cite{troelstra2000basic}. 

Let $\rn{r}$ be a rule, the situation $\rn{r}(C,P_1,...,P_n)$ may be denoted as follows:
\[
    \infer[\rn{r}]{C}{P_1 & \hdots & P_n}
\]  
In such instances, the string $C$ is called the conclusion, and the strings $P_1,...,P_n$ are called the premisses. 

\begin{Example}[Example~\ref{ex:basic:alphabet} cont'd] \label{ex:basic:conjunctionintroduction}
    The rule $\rrn \land$ over basic modal sequents is defined by the following figure without any side-conditions: 
    \[
	\infer[\rrn{\land}]{\mathrm{\Gamma} \seq \Delta \fatcomma \phi \land \psi}{\mathrm{\Gamma} \seq \Delta \fatcomma \phi &  \mathrm{\Gamma} \seq \Delta \fatcomma \psi}
	\]
This rule is admissible for the normal modal logic $K$ --- see Blackburn et al.~\cite{Blackburn2001}. That is, let $\proves[K]$ be the consequence relation for $K$ (over the basic modal alphabet $\lang{B}$): if $\mathrm{\Gamma} \proves[K] \Delta \fatcomma \phi$ and  $\mathrm{\Gamma} \proves[K] \Delta \fatcomma \psi$, then $\mathrm{\Gamma} \proves[K] \Delta \fatcomma \phi \land \psi$.
 \end{Example}
 
 \begin{Definition}[Derivation]
 	 Let $\system{L}$ be a sequent calculus. The set of $\system{L}$-derivations is defined inductively as follows:
 	\begin{itemize}
 	 \item[-]\textsc{Base Case.} If $C$ is a $\lang{P}$-sequent, then the tree consisting of just the node $C$ is an $\system{L}$-derivation.
 	 \item[-]\textsc{Inductive Step.} Let $\mathcal{D}_1,...,\mathcal{D}_n$ are $\system{L}$-derivations, with roots $P_1,...,P_n$, respectively, and let $\rn{r} \in \system{L}$ be a rule such that $\rn{r}(C,P_1,...,P_n)$ obtains. The tree with root $C$ and immediate sub-trees $\mathcal{D}_1,...,\mathcal{D}_n$ is a $\system{L}$-derivation.
   \end{itemize}
 \end{Definition}
 \begin{Definition}[Proof]
    Let $\system{L}$ be a sequent calculus. An  $\system{L}$-derivation $\mathcal{D}$ is a proof iff the leaves of $\mathcal{D}$ are instances of axioms of $\system{L}$.
 \end{Definition}
 
We write $\Gamma \proves[L] \Delta$ to denote that there is a $\system{L}$-proof of the sequent $\Gamma \seq \Delta$. A sequent calculus may have the following relationships to a propositional logic ($\proves$):
 \begin{itemize}
     \item[-] Soundness: If $\Gamma \proves[L] \Delta$, then $\Gamma \proves \Delta$.
     \item[-] Completeness: If $\Gamma \proves \Delta$, then $\Gamma \proves[L] \Delta$.
 \end{itemize}
 
In saying that $\system{L}$ is a sequent calculus for a propositional logic (i.e., that it characterizes that logic), we assert that $\system{L}$ is sound and complete for that logic.

\begin{Example}[Example~\ref{ex:basic:conjunctionintroduction} cont'd] \label{ex:basic:consequence}
We may characterize modal logics, including $K$, by axiom systems --- see, for example, Blackburn et al.~\cite{Blackburn2001}. Each such system $\system{A}$ is a sequent calculus in the general sense of this paper. 
\end{Example}

This concludes the proof-theoretic account of propositional logics in this paper. 

\subsubsection{Model-theoretic Semantics}

We now give a generic account of model-theoretic semantics (M-tS) that can define a logic over a propositional language. By M-tS we mean a frame semantics \emph{\`a la} Kripke~\cite{Kripke1963,kripke1965semantical} --- see also Beth~\cite{Beth1955}.  We follow Blackburn et al.~\cite{Blackburn2001} in the approach for a general account of M-tS.  

\begin{Definition}[Type] \label{def:type}
   A type $\tau$ is a list of non-negative integers. 
\end{Definition}
\begin{Definition}[Frame, Assignment, Pre-model] \label{def:frame}
   Let $\tau := \langle t_1,...,t_n \rangle$ be a type. A $\tau$-frame is a tuple $\langle \set{U}, k_1,...,k_n \rangle$ in which $\set{U}$ is a set and $k_i$ is a relation on $\set{U}$ of arity $t_i$.  Let $\lang{P}= \langle \set{A}, \set{O},\set{C}\rangle$ be a propositional alphabet. An assignment of $\lang{P}$ to $\mathcal{F}$ is a mapping from propositional atoms to sets of worlds, $I:\set{A} \to \pow(\set{U})$. A $\tau$-pre-model over $\lang{P}$ is a pair $\mathfrak{M}:=\langle \mathcal{F}, I \rangle$, in which $\mathcal{F}$ is a $\tau$-frame and $I$ is an assignment of $\lang{P}$ to $\mathcal{F}$.
\end{Definition}

The elements of $\set{U}$ are called possible worlds. One possible objection to the definition of a frame is the absence of operators (i.e., endomorphisms $f:\set{U}^n \to \set{U}$). This is to simplify the setup and is without loss of generality as operators may be regarded as particular types of relations; that is, the operator $f:\set{U}^n \to \set{U}$ corresponds to the $(n+1)$-ary relation $R$ satisfying $R(w,u_1,...,u_n)$ iff $w=f(u_1,...,u_n)$. 


Intuitively, a formula $\phi$ is true in a model $\mathfrak{M}$ at a world $w$ if the world $w$ satisfies the formulas. 

\begin{Definition}[Satisfaction for a Type] \label{def:sat}
   Let $\tau$ be a type and $\lang{P}$ a propositional alphabet. A $\tau$-satisfaction relation for $\lang{P}$ is a relation $\sat$ parameterized by $\tau$-pre-models between worlds $w$ in the pre-models $\mathfrak{M}= \langle \mathcal{F}, I \rangle$ and $\lang{P}$-data such that the following holds:
   \[
   \mathfrak{M}, w \sat \at{p} \qquad \mbox{iff} \qquad w \in I(\at{p} )
   \]
\end{Definition}

\begin{Definition}[Semantics] \label{def:semantics}
   Let $\tau$ be a type and $\lang{P}$ a propositional alphabet. A semantics is a pair $\sem{S} := \langle \set{M}, \sat \rangle$ in which $\set{M}$ is a set of $\tau$-pre-models and $\set{\sat}$ is a $\tau$-satisfaction relation for $\lang{P}$.
\end{Definition}

\begin{Definition} \label{def:valid}
    Let $\sem{S}$ be a semantics. A sequent $\Gamma \seq \Delta$ is valid in $\sem{S}$ --- denoted $\Gamma \entails_{\sem{S}} \Delta$ --- iff, for any $\mathfrak{M} \in \set{M}$ and any $w \in \mathfrak{M}$, if $\mathfrak{M}, w \sat \Gamma$, then $\mathfrak{M}, w \sat \Delta$.
\end{Definition}

\begin{Example}[Example~\ref{ex:basic:consequence} cont'd] \label{ex:modal:satisfaction}
Fix the type $\tau := \langle 2 \rangle $. An example of a $\tau$-frame is a pair $\langle \{x,y\}, R \rangle$ in which $R$ is a binary relation on $\{x,y\}$. Partition the atoms $\set{A}$ into two classes $\set{A}_1$ and $\set{A}_2$; an example of an assignment $I:\set{A} \to \pow(\{x,y\})$ is given as follows:
\[
I(\at{p}):= 
\begin{cases}
x & \text{ if } \at{p}\in \set{A}_1 \\ 
y & \text{ if } \at{p} \in \set{A}_2
\end{cases}
\]
The pair $\mathfrak{M} := \langle \mathcal{F}, I \rangle$ is an example of a model over the basic modal alphabet $\lang{B}$. The basic semantics $\sem{K}$ is the pair $\langle \set{K}, \sat \rangle$ in which $\set{K}$ is the set of all $\tau$-pre-models and $\sat$ is the smallest relation satisfying the clauses in Figure~\ref{fig:modal:satisfaction} together with the following:
  \[
\begin{array}{ccc}
\mbox{$\mathfrak{M},w \sat \Delta\fatcomma\Delta'$} & \qquad \mbox{iff} \qquad & \mbox{ $\mathfrak{M},w \sat \Delta$ and $\mathfrak{M},w \sat\Delta'$} \\
\mbox{$\mathfrak{M},w \sat \Delta\fatsemi\Delta'$} & \qquad \mbox{iff} \qquad & \mbox{$\mathfrak{M},w \sat \Delta$  or $\mathfrak{M},w \sat\Delta'$} 
\end{array}
\]
The validity judgment $\entails_{\sem{K}}$ defines the modal logic $K$ --- see, for example, Kripke~\cite{Kripke1963}, Blackburn et al.~\cite{Blackburn2001}, and Fitting and Mendelsohn~\cite{Fitting1998}.

\end{Example}

\begin{figure}[t]
  \hrule
  \vspace{4mm}
     \[
     \begin{array}{lcl}
     \mathfrak{M}, w  \sat \at{p} & \qquad \mbox{iff} \qquad & w \in  I(\at{p}) \\[1ex]
     \mathfrak{M}, w  \sat \phi \land \psi & \qquad \mbox{iff} \qquad &  \mbox{$\mathfrak{M}, w \sat \phi$ and $ \mathfrak{M}, w  \sat \psi $}   \\[1ex]
     \mathfrak{M}, w  \sat \phi \lor \psi & \qquad \mbox{iff} \qquad &  \mbox{$\mathfrak{M}, w \sat \phi$ or $ \mathfrak{M}, w  \sat \psi $}   \\[1ex]
     \mathfrak{M}, w  \sat \neg \phi & \qquad \mbox{iff} \qquad &  \mbox{not $\mathfrak{M}, w \sat \phi$}   \\[1ex]
      \mathfrak{M}, w  \sat \Box \phi & \qquad \mbox{iff} \qquad &  \mbox{for any $u$, if $wRu$, then $\mathfrak{M}, u \sat \phi$}   
     \end{array}
     \]
\vspace{4mm}
\hrule
    \caption{Satisfaction for $K$}
    \label{fig:modal:satisfaction}
\end{figure}

The significance of $\set{M}$ in the definition of a semantics is that one may not want to consider all pre-models but instead require them to satisfy a specific condition --- see, for example, the persistence condition for the semantics of intuitionistic propositional logic in Section~\ref{sec:ex:ipl}.

The notion of semantics in this paper is generous, including many relations that one would not typically accept as semantics. This is to keep the presentation simple and intuitive. In the next section, we restrict attention to satisfaction relations that admit particular presentations that enable us to analyse them, but doing so presently would obscure the setup.

Historically, the \emph{a priori} definition of a  consequence-relation has been by validity in a semantics. In this paper, we only work with logics for which we assume there is a sequent calculus. Therefore, we may use the nomenclature of Section~\ref{sec:background:propositional} to relate entailment to consequence via provability. A sequent calculus $\system{L}$ may have the following relationship to a semantics $\sem{S}$:

  \begin{itemize}
     \item[-] Soundness: If   $\Gamma \proves[L] \Delta$, then $\Gamma \entails_{\sem{S}} \Delta$.
     \item[-] Completeness: If $\Gamma \entails_{\sem{S}} \Delta$, then $\Gamma \proves[L] \Delta$.
 \end{itemize}

 This completes the summary of propositional logics. Moreover, it completes the technical background to this paper. There are various judgments present, whose relationship are important for the rest of the paper. In the beginning of the next section we provide a brief summary of how all this background is used before proceeding with the technical work.

\section{Constraint Systems} \label{sec:constraints}

This section provides a formal definition of constraint systems. Briefly, a \emph{constraint system} is a sequent calculus in which the data may carry labels representing expressions over some algebra. Rules may manipulate those expressions or demand constraints on them. At the end of a construction in a constraint system, one checks to see that the constraints are coherent and admit an interpretation in the intended algebra. This generalizes the setup of RDvBC in Section~\ref{sec:ex:bi} to an arbitrary algebra and an arbitrary propositional logic. In the next section (Section~\ref{sec:ex:relationalcalculi}), we provide a method for producing constraint systems in a modular way, and in the one after (Section~\ref{sec:ex:ipl}), we illustrate their use in studying model theory. 

The work in this section is technical and abstract; therefore, we give a brief overview summarizing the main ideas. We begin with a propositional logic $\proves$ over an alphabet $\lang{P}$ and a sequent calculus $\system{L}$ which has some desirable features but which is not immediately related to the logic  --- for example, in Section~\ref{sec:ex:bi}, we had BI as the propositional logic and $\system{L}$ a version of $\system{LJ}$ in which contexts are maintained during reduction. We then introduce an algebra, which we understand in terms of first-order structures $\mathcal{A}$, and present in terms of a first-order alphabet $\lang{A}$ --- for example, see the presentation of Boolean algebra in Section~\ref{sec:ex:bi}. We fuse $\lang{P}$ and $\lang{A}$ creating a language $\lang{P} \oplus \lang{A}$ in which the algebra is used to label $\lang{P}$-data. We introduce a valuation map $\nu_I$, parameterized by interpretations $I:\lang{A}\to\mathcal{A}$, that maps $\lang{P}\oplus\lang{A}$-data to $\lang{P}$-data. This defines the action the algebra has on the data. A constraint system is a generalized notion of a sequent calculus of $\lang{P}\oplus\lang{A}$-sequents, which may have $\lang{A}$ expressions as local or global constraints on the correctness of inferences --- see, for example, $\system{LBI}_\mathcal{B}$ in Section~\ref{sec:ex:bi}. Finally, we give correctness conditions for constraint systems: first, a constraint system $\system{C}$ is sound and complete, relative to $\nu_I$, for the logic $\proves$ iff its constructions witness all and only the consequence of the logic; second, a constraint system $\system{C}$ is faithful and adequate, relative to $\nu_I$, for a sequent calculus $\system{L}'$ iff its constructions witness all and only $\system{L}'$-proofs. Throughout the rest of the paper, we illustrate how constraint systems with these correctness conditions aid in studying logic.

The section is composed of three parts. Section~\ref{sec:constraint:reductivelogic} explains the paradigmatic shift necessary for constraint systems: one constructs proofs upward rather than downward. In Section~\ref{sec:constraints:system}, we define constraint systems formally as the enrichment of a sequent calculus by an algebra of constraints. Finally, in Section~\ref{sec:constraint:correctness}, we define correctness conditions relating constraint systems to logics and their proof-theoretic formulations with sequent calculi.

\subsection{Reductive Logic} \label{sec:constraint:reductivelogic}

The traditional paradigm of logic proceeds by inferring a conclusion from 
established premisses using an \emph{inference rule}. This is the paradigm known as \emph{deductive logic}: 
\[
\infer[\Downarrow]{\text{Conclusion}}{\text{Established Premiss}_1 & ... & 
\text{Established Premiss}_n }
\]

In contrast, the experience of the use of logic is often dual to deductive logic in the sense that it proceeds from a putative conclusion to a collection of premisses that suffice for the conclusion. This is the paradigm known as \emph{reductive logic}: 
\[
\infer[\Uparrow]{\text{Putative Conclusion}}{\text{Sufficient Premiss}_1 & ... & \text{Sufficent Premiss}_n }
\]
Rules used backward in this way are called \emph{reduction operators}. The objects created using reduction operators are called \emph{reductions}. We believe that this idea of reduction was first explained in these terms by Kleene~\cite{kleene2002mathematical}. There are many ways of studying reduction, and a number of models have been considered, especially in the case of classical and intuitionistic logic --- see, for example, Pym and Ritter~\cite{Pym2004}. 

Historically, the deductive paradigm has dominated since it exactly captures the meaning of truth relative to some set of axioms and inference rules, and therefore is the natural point of view when considering \emph{foundations of mathematics}. 
However, it is the reductive paradigm from which much of computational logic derives, including various instances of automated reasoning ---  see, for example, Kowalski~\cite{Kowalski1971}, Bundy~\cite{Bundy1983}, and Milner~\cite{Milner1984}. 

Constraint systems (e.g., $\system{LBI}_\mathcal{B}$) sit more naturally within the reductive perspective, with the intuition that one generates constraints as one applies rules backwards. Therefore, in constraint systems, when we \emph{use} a rule we mean it in the reductive of sense.  

Having given the overall paradigm on logic in which constraint systems are situated, we are now able to define them formally and uniformly.

\subsection{Expressions, Constraints, and Reductions} \label{sec:constraints:system}

In Section~\ref{sec:background:propositional}, we defined what we mean by propositional logic. Recall that we may say \emph{algebra} to mean a first-order structure (see Section~\ref{sec:background:fol}). In this context, what we mean by expressions and constraints are terms and formulae, respectively, from an alphabet in which that algebra is interpreted. 

Let $\lang{A}$ be a (first-order) alphabet.
\begin{Definition}[Expression]
   An $\lang{A}$-expression is a term over $\lang{A}$ --- that is, an element of $\set{TERM}(\lang{A})$
\end{Definition}
\begin{Definition}[Constraint]
    An $\lang{A}$-constraint is a formula over $\lang{A}$ --- that is, an element of $\set{WFF}(\lang{A})$
\end{Definition}

When it is clear that an alphabet has been fixed, we may elide it alphabet when discussing labelled formulae, labelled data, and enriched sequents. We use the terms `expression' and `constraint' to draw attention to the fact that we have a certain algebra in mind and a certain way that the constants and functions of the alphabet are meant to be interpreted. For example, in Section~\ref{sec:ex:bi}, we always take symbol $+$ to always be interpreted as Boolean addition. What may change is the interpretation of variables. In short, we have some set of intended interpretations that are \emph{coherent}.

\begin{Definition}[Coherent Interpretations]
   Let $\set{I}$ be a set of interpretations of an algebra $\mathcal{A}$ in $\lang{A}$. The set $\set{I}$ is coherent iff, for any $I_1,I_2 \in \set{I}$, they behave the same except possibly for atoms.
\end{Definition}

Typically, the set of intended interpretations is maximal in the sense that any interpretation of the algebra in the alphabet is either in the set or is not a variant of an interpretation in the set.

We use expressions to enrich the language of the propositional logic and thereby express meta-theoretic conditions on formulae and sequents. Let $\lang{P}$ be a propositional alphabet. 

\begin{Definition}[Labelled Data]
The set of labelled $\lang{P}$-data is defined inductively as follows:
    \begin{itemize}
    \item[-] \textsc{Base Case.} If $\phi$ is a formula and $e$ is an $\lang{A}$-expression, then $\phi \cdot e$ is a $\lang{A}$-labelled $\lang{P}$-datum.
    
    \item[-]\textsc{Inductive Step.} If $\delta_1$, ... , $\delta_n$ are labelled $\lang{P}$-data, $\bullet$ is a data-constructor in $\lang{P}$ with arity $n$, and $e$ is an $\lang{A}$-expression, then $\bullet(\delta_1,...,\delta_n)\cdot e$ is a $\lang{A}$-labelled $\lang{P}$-datum.
    \end{itemize}
\end{Definition}

\begin{Definition}[Enriched Sequent]
An $\lang{A}$-enriched $\lang{P}$-sequent is a pair $\Pi \seq \Sigma$, in which $\Pi$ and $\Sigma$ are multisets of $\lang{A}$-labelled $\lang{P}$-data and constraints.
\end{Definition}

We may suppress $\lang{A}$ and $\lang{P}$ when it is clear what alphabet for what algebra is labeling what propositional language. Observe that we have shifted from the object-logic to the meta-logic; that is enriched sequents are a restricted form of meta-logic sequents that encapsulate object-logic sequents with conditions expressed by expressions from the algebra. This setup differs slightly from the presentation of RDvBC in Section~\ref{sec:ex:bi} to simplify presentation of the general case. Recall that data is the general name for contexts in the propositional logic, which may be bunches. Consequently, the presentation of RDvBC in terms of enriched sequents would consist of pairs of multisets each of which contain only one element, the labelled bunch.

\begin{Example} \label{ex:rdvbc:enriched}
The are various enriched sequents in RDvBC --- see Section~\ref{sec:ex:bi}. An additional example is as follows:
\[
\at{p} \cdot x \fatcomma (q \fatsemi r)\cdot y  \seq (\at{p}\land \at{q}) \cdot x
\]
\end{Example}

A \emph{constraint system} is a generalization of sequent calculus that uses enriched sequents and constraints. 
The constructions of a constraint system are generated co-recursively on enriched sequents, producing constraints along the way along.

\begin{Definition}[Constraint System]
    A constraint rule is a relation between an enriched sequents and a list of enriched sequents and constraints. A constraint system is a set of constraint rules.
\end{Definition}

We use the same notation as in Section~\ref{sec:background:propositional} for constraint rules; that is, that $\rn{r}(C,P_1,...,P_n)$ obtains may be expressed as follows:
\[
\infer[\rn{r}]{C}{P_1 & ... & P_n}
\]
In this case, $C$ is an enriched sequent and $P_1,...,P_n$ are either enriched sequents or constraints. The terms \emph{premiss} and \emph{conclusion} are analogous to those employed for sequent calculus rules in Section~\ref{sec:background:propositional}. We assume the convention of putting constraints after enriched sequents in the list of premisses. 

\begin{Example}[Example~\ref{ex:rdvbc:enriched} cont'd] \label{ex:rdvbc:rule}
System $\system{LBI}_\mathcal{B}$ in Section~\ref{sec:ex:bi} is a constraint system.  Therefore, any rule in it is an example of a constraint rule. We shall consider two examples.

The following is a constraint rule:
\[
\infer[\rrn\ast]{\Delta\fatcomma\Delta' \seq \phi*\psi}{\Delta \cdot V \seq \phi & \Delta' \cdot \bar{V} \seq \psi}
\]
If $e$ is a label on $\Delta$, then $\Delta \cdot V$ denotes the result of replacing $e$ by a product of $e$ and $V$. The following inference is an instance of the rule:
\[
    \infer{\at{p} \cdot 1 \fatcomma (q \fatsemi r)\cdot 1 \seq (\at{p} \ast \at{q})\cdot 1}{\at{p} \cdot (1x) \fatcomma (q \fatsemi r)\cdot (1y) \seq \at{p} \cdot 1 & \at{p} \cdot (1\bar{x}) \fatcomma (q \fatsemi r)\cdot (1\bar{y}) \seq \at{q} \cdot 1  }
\]

Another example of a constraint rule is as follows:
\[
\infer[\rn{taut}]{\phi\cdot y \fatcomma \Delta\cdot V \metaseq \phi\cdot 1}{y=1 \land V = 0 }
\]
Here $V$ are all the labels on $\Delta$ and $V =0$ denotes $x_1=0 \with .... \with x_n = 0$ if $V$ is the list $x_1,...,x_n$. An instance of the rule is the following:
\[
\infer{p \cdot(1x)\fatcomma (q \fatsemi r) \cdot 1y \seq p \cdot 1}{1x =1 & 1y = 0}
\]
Intuitively, this corresponds to an \emph{axiom} of a sequent calculus as it reduces to constraints which, when solved, state weather or not the sequent in the conclusion is a consequence of BI. 
\end{Example}

Unlike sequent calculi, a constraint system does not necessarily contain axioms. This is possible because the set of things generated by a constraint system is defined co-inductively, so the restriction is unnecessary. We define reductions co-inductively because constraint systems sit within the paradigm of reductive logic. The reductions in this paper are all finite. 

\begin{Definition}[Reduction in a Constraint System]
    Let $\system{C}$ be a constraint system and let $S$ be an enriched sequent. A tree of enriched sequents $\mathcal{R}$ is a $\system{C}$-reduction of $S$ iff   there is a rule  $\rn{r} \in \system{C}$ such that $\rn{r}(S,P_1,..,P_n)$ obtains and the immediate sub-trees $\rn{R}_i$, with root $P_i$, are as follows: if $P_i$ is an enriched sequent,   it is a $\system{C}$-reduction of $P_i$; the single node $P_i$, otherwise (i.e., if $P_i$ is a constraint).
\end{Definition}

\begin{Example}[Example~\ref{ex:rdvbc:rule} cont'd] \label{ex:rdvbc:reduction}
A reduction in a constraint system (the system $\system{LBI}_\mathcal{B}$) is given in Example~\ref{ex:rdvbc:lbiBreduction}. 
\end{Example}

This concludes the definition of a constraint system and reduction from it. It remains to give conditions defining how it relates to logics.

\subsection{Correctness Conditions} \label{sec:constraint:correctness}

The distinguishing feature of reductions in a constraint system is the constraints. These constraints are understood as correctness conditions in two ways. They are \emph{global}
correctness conditions when they determine that a completed reduction is a valid certificate witnessing that some sequent is the consequence of a logic; this is \emph{soundness} and \emph{completeness}. They are \emph{local} correctness conditions when they determine that each reduction step corresponds to a valid inference (i.e., an admissible inference) in the logic; this is faithfulness and adequacy for some sequent calculus for the logic. Of course, local correctness implies global correctness. In either reading, we regard constraints in using rules as side-conditions on the reduction.


\begin{Definition}[Side-condition]
 Let $\system{C}$ be a constraint system and let $\mathcal{R}$ be a $\system{C}$-reduction. A side-condition of $\mathcal{R}$ is a constraint that is a leaf of $\mathcal{R}$. 
\end{Definition}

The side-conditions are global constraints on the reduction, determining the conditions for which the structure is meaningful. 

\begin{Definition}[Coherent Reduction]
 Let $\system{C}$ be a constraint system, let $\mathcal{R}$ be a $\system{C}$-reduct\-ion, and let $\set{S}$ be the set of side-conditions of $\mathcal{R}$. The set $\set{S}$ is coherent iff there is an interpretation in which all of the side-conditions in $\set{S}$ are  valid; the reduction $\mathcal{R}$ is coherent iff $\set{S}$ is coherent.
\end{Definition}

We may regard coherent reductions as proofs of certain sequents, but this requires a method of reading what sequent of the propositional logic the reduction asserts. 

\begin{Definition}[Ergo] \label{def:ergo}
 An ergo is a map $\nu_I$, parameterized by intended interpretations, from enriched sequents to sequents.
\end{Definition}

Let $\system{C}$ be a constraint system and $\nu$ an ergo.  We write $\Gamma \proves[C]^\nu \Delta$ to denote that there is a coherent $\system{C}$-reduction $\mathcal{R}$ of an enriched sequent $S$ such that $\nu_I(S)=\Gamma \seq \Delta$, where $I$ is an interpretation satisfying all the side-conditions of $\mathcal{R}$. An example of this is the valuation given for RDvBC in Section~\ref{sec:ex:bi} that deletes formulae and bunches labelled by an expression that evaluates to $0$ and keeps those that evaluate to $1$.

\begin{Definition}[Soundness and Completeness of Constraint Systems] \label{def:snc:ergo}
A constraint system $\system{C}$ may have the following relationships to a propositional logic:
 \begin{itemize}
     \item[-] Soundness: If  $\Gamma \proves[C]^\nu \Delta$, then $\Gamma \proves[] \Delta$.
     \item[-] Completeness: If $\Gamma \proves[] \Delta$, then $\Gamma \proves[C]^\nu \Delta$.
 \end{itemize}
\end{Definition}

This defines constraint systems and their relationship to logics. Observe that soundness and completeness is a \emph{global} correctness condition on reductions in the sense that only once the reduction has been completed and one has generated all of the constraints and solved them to find an interpretation does one know whether or not the reduction witnesses the validity of some sequent in the logic. In other words, a partial reduction (i.e., a reduction to which one may still apply rules) does not necessarily contain any proof-theoretic information about the logic.

In contrast, one may consider a \emph{local} correctness conditions in which applying a reduction operator from a constraint system corresponds to using some rules of inference in a sequent calculus for the logic. This is stronger than the global correctness condition as when the reduction is completed, and the constraints are solved, the resulting interpretations that allow one to read a reduction in a constraint system as a proof in a sequent calculus for the logic, and thus as a certificate for the validity of some sequent.

Fix a propositional alphabet $\lang{P}$, an algebra $\mathcal{A}$, an alphabet $\lang{A}$ for that algebra, and a set $\set{I}$ of intended interpretations of $\lang{A}$ in $\mathcal{A}$. Fix a constraint system $\system{C}$ and an ergo $\nu$. The ergo extends to $\system{C}$-reductions by pointwise application to the enriched sequents in the tree and by deleting all the constraints.





Using this extension, constraint systems are computational devices capturing sequent calculi. For this reason, we do not use the terms \emph{soundness} and \emph{completeness}, but rather use the more computational terms of \emph{faithfulness} and \emph{adequacy}. 

\begin{Definition}[Faithful \& Adequate] \label{def:fna}
 Let $\system{C}$ be a constraint system, let  $\system{L}$ be a sequent calculus, and let $\nu$ be a valuation.
\begin{itemize}
    \item[-] System $\system{C}$ is \emph{faithful} to $\system{L}$ if, for any $\system{C}$-reduction $\mathcal{R}$ and interpretation $I$ satisfying the constraints of $\mathcal{R}$, the application $\nu_I(\mathcal{R})$ is an $\system{L}$-proof.
    \item[-] System $\system{C}$ is \emph{adequate} for $\system{L}$ if, for any $\system{L}$-proof $\mathcal{D}$, there is a $\system{C}$-reduction $\mathcal{R}$ and an interpretation $I$ satisfying the constraints of $\system{R}$ such that $\nu_I(\mathcal{R}) = \mathcal{D}$.
\end{itemize}
\end{Definition}

Intuitively, constraint systems for a logic (more precisely, constraint systems that are faithful and adequate with respect to a sequent calculus for a logic) separate combinatorial and idiosyncratic aspects of that logic. The former refers to how rules manipulate the data in sequents, while the latter refers to the constraints generated by the rules. Note that this gives a \emph{local} correctness condition of reductions from a constraint system as each reductive inference in the constraint system corresponds to some reductive inference in a sequent calculus for the logic.

In the next section (Section~\ref{sec:ex:relationalcalculi}), we provide sufficient conditions for a propositional logic to have a constraint system that evaluates to a sequent calculus for that logic. The conditions are quite encompassing and automatically give soundness and completeness for the sequent calculus for a semantics for the logic. Attempting a total characterization (i.e., precisely defining the properties a logic must satisfy in order for it to have a constraint system and a valuation to a sequent calculus for the logic) is unrealistic. The reason is that propositional logics, constraints systems, algebras, and valuations, have so many degrees of freedom that one could not plainly present their dependencies at once. Instead, one should consider such a classification relative to a fixed structure of propositional logics (e.g., substructural logics with data comprising multisets of formulae), for a fixed algebra (e.g., Boolean algebra), and fixed valuations (i.e., keeping formulae whose label evaluate to $1$ and deleting formulae whose label evaluate to $0$). 


Significantly, having constructed a reduction in a constraint system, one must solve the constraints before one knows what `proof' the reduction represents for the target logic. This may be challenging depending on the algebra over which the constraint take place. Nonetheless, with even relatively simple algebras one may have relatively useful constraint systems; for example RDvBC uses Boolean algebra, which admits algorithmic solvers, that means one can outsource the solving of constraints for $\system{LBI}_\mathcal{B}$-reductions. However, the point of constraint systems is not to \emph{do} proof-search but to \emph{study} proof-search. In this way, how the constraints are solved is less important than how they can be interpreted in terms of control during proof-search in the object-logic. 

\section{Example: Relational Calculi} \label{sec:ex:relationalcalculi}

The relational calculi introduced by Negri~\cite{Negri2005} can be viewed as  constraint systems; that is, the \emph{constraint algebra} is provided by a first-order theory capturing an M-tS for a logic, and the labelling action captures satisfaction in that semantics. Traditionally, $x:\phi$ is used in place of $\phi \cdot x$ for relational calculi, and we shall adopted this notation for this section to be consistent with the existing work. The change in notation is a \emph{aide-m\'emoire} that we are working with a particular form of constraint systems.    

This section gives sufficient conditions for a sequent calculus to admit a relational calculus. We further give conditions under which these relational calculi (regarded as constraint systems) are faithful and adequate for a sequent calculus for the logic. We continue the study of the modal logic $K$ in Section~\ref{sec:background:propositional} as a running example.

First, we define what it means for a semantics of a propositional logic to be first-order definable; this is a pre-condition for producing relational calculi that express the semantics. We call the propositional logic we are studying the \emph{object-logic}; and, we call FOL the \emph{meta-logic}. For clarity, we use the convention prefixing \emph{meta}- for structures at the level of the meta-logic where the terminology might otherwise overlap; for example,  \emph{formulae} are syntactic construction at the object level, and \emph{meta-formulae} are syntactic construction in the meta-logic. 

Second, we give a sufficient condition, called \emph{tractability}, for us to take a first-order definition $\Omega$ of a semantics and produce a relational calculus from it. Essentially, the condition amounts to unfolding $\Omega$ within \system{G3c} so that we can suppress all the logical structures from the meta-logic, leaving only a labelled calculus for the propositional logic --- namely, the relational calculus. 

Third, we give a method for transforming tractable definitions into sequent calculi and prove that the result is sound and complete for the semantics. 
 
\subsection{Tractable Propositional Logics} \label{sec:relationalcalculi:tractable}

Relational calculi for a logic work by internalizing a semantics of that logic. In work by Negri~\cite{Negri2005} on relational calculi for normal modal logics, the basic atomic formulae over which the relational calculi operate come in two forms: they are either of the form $(x:\phi)$, in which $x$ is a variable denoting an arbitrary world, $\phi$ is a formula, and $:$ is a pairing symbol intuitively saying that $\phi$ is satisfied at $x$; or, they are of the form $xRy$, in which $x$ and $y$ are variables denoting worlds and $R$ is a relation denoting the accessibility relation of the frame semantics. Therefore, we begin by fusing the language $\lang{P}$ of the object-logic with a first-order language $\lang{F}$, able to express frames for the semantics, into the first-order language we use for the relational calculi. 

\begin{Definition}[Fusion]
    Let $\lang{F}:= \langle \set{R}, \emptyset,\set{K},\set{V} \rangle$ be a first-order alphabet and let $\lang{P} := \langle \set{A}, \set{O}, \set{C} \rangle$. The fusion $\lang{F} \otimes \lang{P}$ is the first-order alphabet $\langle \set{R}\cup\{:\}, \set{O} \cup \set{C}, \set{K}\cup \set{A}, \set{V} \rangle $ 
\end{Definition}

Observe that $\lang{P}$-formulae and $\lang{F}$-terms both becomes terms in $\lang{F}\oplus \lang{P}$, and $:$ is a relation. In particular, the object-logic operators (i.e., connectives, modalities) are function-symbols in the fusion. Further note that $(x:\phi)$ and $(\phi:x)$ are well-formed formulae in the fusion; the former is desirable, and the latter is not. We require a \emph{model theory} $\Omega$ over the fused language such that $:$ is interpreted as satisfaction in the semantics. Relative to such a theory, while well-formed, the meta-formulae $(\phi:x)$ are nonsense. To aid readability, we shall use the convention of writing $\hat{\phi}$ for meta-variables that we intend to be interpreted as object-formulae and $\hat{\Gamma}$ or $\hat{\Delta}$ for meta-variables that we intend to be interpreted as object-data. 

\begin{Definition}[Definition of a Semantics] \label{def:definition}
   Let $\Omega$ be a set of sentences from a fusion $\lang{F} \otimes \lang{P}$ and let $\sem{S}$ be a semantics over $\lang{P}$. The set $\Omega$ defines the semantics $\sem{S}$ iff the following holds: $\Omega, (x : \Gamma) \proves (x : \Delta)$ iff $\Gamma \entails \Delta$.
\end{Definition}

Such theories $\Omega$ may at first appear obscure, but in practice they can be fairly systematically constructed. Intuitively, the abstractions of $\Omega$ are composed of models from the semantics together with an interpretation of the satisfaction relation. Thus, $\Omega$ is typically composed of two theories $\Omega_1$ and $\Omega_2$.  The theory $\Omega_1$ captures frames; for example, in modal logic, if the accessibility relation is transitive, then $\Omega_1$ contains $\forall x,y,z(xRy \with y R z \Rightarrow xRz)$. The theory $\Omega_2$ captures the conditions of the satisfaction relation; for example, if the object-logic contains an additive conjunction $\land$, then $\Omega$ may contain $\forall x, \hat{\phi}, \hat{\psi}((x:\hat{\phi}\land \hat{\psi}) \Rightarrow (x:\hat{\phi})\with (x:\hat{\psi}))$ and $\forall x, \hat{\phi}, \hat{\psi}((x:\hat{\phi})\with (x:\hat{\psi}) \Rightarrow (x:\hat{\phi}\land \hat{\psi}))$. For an illustration of how $\Omega$ can be constructed according to this intuition in even relatively complex settings, see work on the logic of Bunched Implications by Gheorghiu and Pym~\cite{Alex:BI_Semantics}. 

\begin{Example}\label{ex:modal:satisfaction:symbolic}
By the universal closure of $(\Phi \iff \Psi)$ we mean the meta-formulae $\Theta$ and $\Theta'$ in which $\Theta$ is the universal closure of $\Phi \Rightarrow \Psi$ and $\Theta'$ is the universal closure of $\Psi \Rightarrow \Phi$. Consider the semantics $\sem{K}= \langle \set{M}, \sat \rangle$ in Example~\ref{ex:modal:satisfaction}. It is defined by the universal closures of the formulae in Figure~\ref{fig:modal:satisfaction:symbolic}, which merits comparison with Figure~\ref{fig:modal:satisfaction}, together with the universal closure of the following:
\[
\begin{array}{ccc}
(x:\hat{\Gamma}\fatcomma \hat{\Delta}) & \iff & (x:\hat{\Gamma})\with(x:\hat{\Delta}) \\
(x:\hat{\Gamma}\fatsemi \hat{\Delta}) & \iff & (x:\hat{\Gamma})\parr(x:\hat{\Delta})
\end{array}
\]
Every model in $\set{M}$ intuitively gives an abstraction of these formulas since they simply give a formal expression of the clauses defining satisfaction. 

Notably, there is no meta-formula corresponding to atomic satisfaction --- that is, $(w:\at{p})$ --- because it is handled by the structure of meta-sequents. That is, it follows from working with validity directly (i.e., without passing though truth-in-a-model): atomic satisfaction is captured by an atomic tautology, $\Omega, (x : \at{ p}) \seq (x : \at{ p})$. 
\end{Example} 

\begin{figure}[t]
       \hrule
   \vspace{4mm}
     \[
     \begin{array}{lcl}
     (w  : \hat{\phi} \land \hat{\psi}) & \iff &  (w : \hat{\phi})\with (w :\hat{\psi})   \\[1ex]
     (w : \hat{\phi} \lor \hat{\psi}) & \iff &  (w :\hat{\phi})\parr (w:\hat{\psi})  \\[1ex]
     (w : \neg \hat{\phi}) & \iff &  \big((w : \hat{\phi})\Rightarrow \bot\big)   \\[1ex]
    (w : \Box \hat{\phi}) & \iff &  \forall u(wRu \Rightarrow (u : \hat{\phi}))   
     \end{array}
     \]
     \vspace{4mm}
    \hrule
    \caption{Satisfaction for Modal Logic $K$ (Symbolic)}
    \label{fig:modal:satisfaction:symbolic}
\end{figure}

We may use the meta-logic to characterize those propositional logics whose semantics is particularly amenable to analysis; first-order definability is, perhaps, the most general condition we may demand. What are some other properties of $\Omega$ that may be useful? Since we are interested in a \emph{computational} analysis of the semantics, we require that it is finite, among other things. In particular, we restrict the structure of the theory to something amenable to proof-theoretic analysis according to $\system{G3c}$. 

There is literature on  generating proof systems for propositional logics defined axiomatically; see, for example, work by Ciabattoni et al.~\cite{Ciabattoni2008,Ciabattoni2009,CIABATTONI2012}. Within this tradition, Marin et al.~\cite{Marin2022} have used \emph{focusing} in intuitionistic and classical logic, conceived of as a meta-logic, as a general tool to uniformly express an algorithm for turning axioms into rules applicable across different domains. We use a similar method and, therefore, polarize the syntax for the meta-logic. 

Let $\set{MA}$ be the set of meta-atoms. The \emph{positive} meta-formulae $P$ and \emph{negative} meta-formulae $N$ are defined as follows:
\[
\begin{array}{lcl}
P & ::= & A \in \set{MA} \mid \bot \mid P \with P \mid N \Rightarrow P \mid P \parr P \mid  \exists X P  \\
N & ::= & A \in \set{MA} \mid N \with N \mid P \Rightarrow N \mid \forall X N
\end{array}
\]
This taxonomy arises from behaviour; specifically,using this taxonomy we can define a class of formulae that we can systematically transform them into \emph{synthetic rules} using focusing in $\system{G3c}$. While closely related to the taxonomy used by Marin et al.~\cite{Marin2022} it is not the same as they work over a syntax that has positive and negative connectives.


\begin{Definition}[Polarity Alternation]
    The number of \emph{polarity alternations} in a polarised formula $\Phi$ is $\pi(\Phi)$ defined as follows:
    \[
    \pi(\Phi):=
    \begin{cases}
        0 & \text{if } \Phi \in \set{MA} \\
        \max\{\pi(\Phi_1), \pi(\Phi_2)\} & \text{if } \Phi = \Phi_1 \circ \Phi_2 \text{ and } \circ \in \{\with, \parr\} \\
        \pi(\Phi_1) & \text{if } \Phi = Q X \Psi \text{ and } Q \in \{\forall, \exists\} \\
        1+\max\{\pi(\Phi_1), \pi(\Phi_2)\} & \text{if } \Phi = \Phi_1\Rightarrow \Phi_2
    \end{cases}
    \]
\end{Definition}

 \begin{Definition}[Tractable Meta-formula]
    A meta-formula $\Phi$ is tractable iff $\Phi$ is negative and $\pi(\Phi) \leq 2$, or $\Phi$ is positive and $\pi(\Phi) \leq 1$.
 
    \end{Definition}

The class of \emph{geometric implications} studied by Negri~\cite{negri2003contraction} for the systematic generation of sequent calculus rules from axioms defining propositional logics is a subset of the tractable formulae. A meta-formula $\Theta$ is a geometric implication iff $\Theta$ is the universal closure of a meta-formula of the form $(\Phi_1 \with ... \with \Phi_m) \Rightarrow (\exists Y_1\Psi_1 \parr ... \parr \exists Y_n\Psi_n)$ such that $\Psi_i:= \Psi_1^i \with ... \with \Psi^i_{m_i}$, with the $\Psi_j^i$ meta-atoms for $1 \leq j \leq m_i$ and $1 \leq i \leq n$, and $ \Phi_i$ meta-atoms for $1 \leq i \leq m$.  To see this, 
observe that geometric implications are of the form $N \Rightarrow P$ in which $N$ is the conjunction of atoms and $P$ is the disjunction of (positive) formulae of existentially quantified conjunctions --- that is, formulae of the form $\exists P'$, where $P'$ is a conjunction of atoms. Docherty and Pym~\cite{Docherty2018,Docherty2019} have used this notion of meta-formulae to give a uniform account of proof systems internalizing semantics for the family of bunched logics, with application to separation logics.

The motivation for tractability is to make a certain step in the generation of relational calculi possible, as seen in the proof of Proposition~\ref{lem:calculus:tractable}. 

\begin{Definition}[Tractable Theory, Semantics, Logic] \label{def:tractable:theory}
   A set of meta-formulae $\Omega$ is a tract\-able theory iff $\Omega$ is finite and any $\Phi \in \Omega$ is a negative tractable meta-sentence. A semantics $\sem{S}$ is tractable iff it is defined by a tractable theory $\Omega$. A propositional logic is tractable iff it admits a tractable semantics $\sem{S}$.
\end{Definition}

\begin{Example} \label{ex:modal:satisfaction:tractable}
The semantics for modal logic in Example~\ref{ex:modal:satisfaction} is tractable, as witnessed by the tractable definition in Example~\ref{ex:modal:satisfaction:symbolic}. 
\end{Example}

 It remains to give an algorithm that generates a relational calculus given a tractable definition and to prove correctness of that algorithm. Fix a semantics $\sem{S} := \langle \set{M}, \sat \rangle$ with a tractable definition $\Omega$. Recall that $\Gamma \entails \Delta$ obtains iff  $\Omega, (x:\Gamma) \proves (x:\Delta)$ obtains. The relational calculus we generate is a meta-sequent calculus $\system{R}$ for the meta-logic expressive enough to capture all instances $\Omega, (x:\Gamma) \proves (x:\Delta)$ but sufficiently restricted such that all the meta-connectives and quantifiers may be suppressed. 

 \subsection{Generating Relational Calculi} \label{sec:relationalcalculi:generating}

 By \emph{generic hereditary reduction} on a meta-formula $\Phi$ we mean the indefinite use of reduction operators from $\system{G3c}$ on $\Phi$ and the generated sub-formulae, until they are meta-atoms, beginning with a meta-sequent $\Phi,\Pi \metaseq \Sigma$, with generic $\Pi$ and $\Sigma$. For example, the following is a generic hereditary reduction for $(A \with B) \parr (C \with D)$  with $A$, $B$, $C$, and $D$ as meta-atoms:
\[
    \infer[\Uparrow \rrn \parr]{(A \with B) \parr (C \with D), \Pi \metaseq \Sigma}{
        \infer[\Uparrow \rrn \with]{(A \with B), \Pi \metaseq \Sigma}{A, B, \Pi \metaseq \Sigma}
        &
        \infer[\Uparrow \rrn \with]{(C \with D), \Pi \metaseq \Sigma}{C, D, \Pi \metaseq \Sigma}
    }
\]
Such reductions are collapsed into \emph{synthetic} rules, which is the rule-relation taking the putative conclusion to the premisses. The above instance collapses to the following:
\[
\infer{(A \with B) \parr (C \with D), \Pi \metaseq \Sigma}{A, B, \Pi \metaseq \Sigma & C, D, \Pi \metaseq \Sigma}
\]
The quantifier rules have side-conditions in order to be applicable, and we assert these conditions in the synthetic rule. For example, when using $\lrn \forall$ when doing generic hereditary reduction on $\forall X \Phi$, we require that the term $T$ for which the variable $X$ is substituted in $\Phi$ is already present in the meta-sequent; for example let $\Phi:= (A(X) \with B(X)) \parr (C(X) \with D(X))$, we have the following synthetic rule for $\forall X \Phi$ with the side condition that $T$ occurs in either $\Pi$ or $\Sigma$:
\[
\infer{\forall X(A(X) \with B(X) \parr (C(X) \with D(X)), \Pi \metaseq \Sigma}{A(T), B(T), \Pi \metaseq \Sigma & C(T), D(T), \Pi \metaseq \Sigma}
\]

\begin{Definition}[Sequent Calculus for a Tractable Theory] \label{def:calculus:tractable}
    Let $\Omega$ be a tractable theory. The sequent calculus $\system{G3c}(\Omega)$ is composed of $\ax, \bot, \lrn \cont$, $\rrn \cont$, and the synthetic rules for the meta-formulae in $\Omega$.
\end{Definition}

The tractability condition is designed such that the following holds:

\begin{Proposition}\label{lem:calculus:tractable}
    Let $\Omega$ be a tractable definition and let $\Pi$ and $\Sigma$ be multisets of meta-atoms,
    \[
     \Omega, \Pi \proves[G3c] \Sigma \qquad \mbox{iff} \qquad \Omega, \Pi \proves[\system{G3c}(\Omega)] \Sigma
     \]
\end{Proposition}
\begin{proof}
    Assume $\Omega, \Pi \proves[G3c] \Sigma$. Without loss of generality (see, for example, Liang and Miller~\cite{LIANG20094747} and Marin~\cite{Marin2022}), there is a focused $\system{G3c}+\rrn{\cont}+\lrn{\cont}$-proof $\mathcal{D}$ of $\Omega, \Pi \metaseq \Sigma$. We can assume that $\mathcal{D}$ is \emph{focused} up to possibly using instance of $\lrn{c}$ or $\rrn{c}$. That is, $\mathcal{D}$ is structured by sections of alternating phases of the following kind:
    \begin{itemize}
        \item[-] an instance of $\lrn c$ or $\rrn c$
        \item[-] hereditary reduction on positive meta-formulae on the right and negative meta-formulae on the left
        \item[-] eager reduction on negative meta-formulae on the right and positive meta-formulae on the left.
    \end{itemize}
    Since $\Pi$ and $\Sigma$ are composed of meta-atoms and $\Omega$ is composed of negative meta-formulae, $\mathcal{D}$ begins by a contraction and then hereditary reducing on some $\Phi \in \Omega$. Since $\Phi$ is tractable, this section in $\mathcal{D}$ may be replaced by the synthetic rule for $\Phi$. Doing this to all the phases in $\mathcal{D}$ yields a tree of sequents $\mathcal{D}'$ which is a $\Omega, \Pi \proves[\system{G3c}(\Omega)] \Sigma$.

    Assume $\Omega, \Pi \proves[\system{G3c}(\Omega)] \Sigma$. Since all the rules in $\system{G3c}(\Omega)$ are admissible in $\system{G3c}$, we immediately have $\Omega, \Pi \proves[\system{G3c}] \Sigma$. 
\end{proof}

\begin{Example} \label{ex:modal:calculus}
    Consider the tractable theory $\Omega_{\sem{K}}$ in Example~\ref{ex:modal:satisfaction:symbolic}. The sequent calculus $\system{G3c}(\Omega_{\sem{K}})$ contains, among other things, the following rules corresponding to the clause for $\land$ in Figure~\ref{fig:modal:satisfaction:symbolic} in which $w$, $\phi$, and $\psi$ already occur in $\Omega$, $\Pi$, or $\Sigma$:
\[
        \infer{\Omega, \Pi \metaseq \Sigma}{\Omega, \Pi \metaseq \Sigma, (w:\phi \land \psi) & \Omega, (w:\phi), (w:\psi), \Pi \metaseq \Sigma}
\]
\[
    \infer{\Omega, \Pi \metaseq \Sigma}{\Omega, \Pi, (w:\phi \land \psi) \metaseq \Sigma & \Omega, \Pi \metaseq \Sigma, (w:\phi) & \Omega, \Pi \metaseq \Sigma, (w:\psi)}
    \]
In practice, one does not use the rules in this format. Rather, one would only apply the rules if one already knew that the left-branch would terminate; that is, one uses the following:
\[
\infer{\Omega, (w:\phi \land \psi), \Pi \metaseq \Sigma}{\Omega, (w:\phi), (w:\psi), (w:\phi \land \psi),  \Pi \metaseq \Sigma } 
\]
\[
\infer{\Omega, \Pi \metaseq \Sigma,  (w:\phi \land \psi)}{\Omega, \Pi \metaseq \Sigma,  (w:\phi \land \psi), (w:\phi) &  \Omega, \Pi \metaseq \Sigma, (w:\phi \land \psi), (w:\psi) }
\]
This simplification can be made systematically according to the shape of the meta-formula generating the rules; it corresponds to \emph{forward-chaining} and \emph{back-chaining} in the proof-theoretic analysis of the meta-formula --- see, for example, Marin et al.\cite{Marin2022}.

One desires a systematic account of the transformation of rules of arbitrary shape into rules of other (more desirable) shape. This remains to be considered in the context of relational calculi and demands further analysis on the structure of $\Omega$. Some results of such transformations for arbitrary sequent calculi have been provided by Indrzejczak~\cite{indrzejczak2021sequents}.
\end{Example}

The calculus $\system{G3c}(\Omega)$ is a restriction of $\system{G3c}$ precisely encapsulating the proof-theoretic behaviours of the meta-formulae in $\Omega$. It remains to suppress the logical constants of the meta-logic entirely, and thereby yield a relational calculus expressed as a labelled sequent calculus for the propositional logic. 

\begin{Definition}[Relational Calculus for a Tractable Theory] \label{def:relationalcalculus}
    Let $\Omega$ be a tractable theory. The relational calculus for $\Omega$ is the sequent calculus $\system{R}(\Omega)$ that results from $\system{G3c}(\Omega)$ by suppressing $\Omega$.
\end{Definition}

\begin{Theorem}[Soundness \& Completeness] \label{thm:snc:relationalcalculi}
    Let $\sem{S}$ be a tractable semantics and let $\Omega$ be a tractable definition for $\sem{S}$.
    \[
    \Gamma \entails_{\sem{S}} \Delta \qquad \mbox{iff} \qquad  (x:\Gamma) \proves[R(\Omega)] (x:\Delta)
    \]
\end{Theorem}
\begin{proof}
    We have the following:
    \[
    \begin{array}{cclr}
    \Gamma \entails_{\sem{S}} \Delta & \qquad \mbox{iff} \qquad & \Omega, (x:\Gamma) \proves (x:\Delta) & \qquad \mbox{(Definition~\ref{def:tractable:theory})} \\
    & \qquad \mbox{iff} \qquad &  \Omega, (x:\Gamma) \proves[G3c] (x:\Delta) & \qquad \mbox{(Proposition~\ref{lem:snc:G3c})} \\
     & \qquad \mbox{iff} \qquad &  \Omega, (x:\Gamma) \proves[G3c(\Omega)] (x:\Delta) & \qquad \mbox{(Proposition~\ref{lem:calculus:tractable})} 
    \end{array}
    \]
    It remains to show that $\Omega, (x:\Gamma) \proves[G3c(\Omega)] (x:\Delta)$ iff $(x:\Gamma) \proves[R(\Omega)] (x:\Delta)$.

    Let $\mathcal{D}$ be a $\system{G3c}(\Omega)$-proof of  $\Omega, (x:\Gamma) \metaseq (x:\Delta)$, and let $\mathcal{D}'$ be the result of removing $\Omega$ from every meta-sequent in $\mathcal{D}$. By Definition~\ref{def:relationalcalculus}, we have that $\mathcal{D}'$ is a $\system{R}(\Omega)$-proof of $(x:\Gamma) \metaseq (x:\Delta)$. Thus, $\Omega, (x:\Gamma) \proves[G3c(\Omega)] (x:\Delta)$ implies $(x:\Gamma) \proves[R(\Omega)] (x:\Delta)$.

    Let $\mathcal{D}$ be a $\system{R}(\Omega)$-proof of $(x:\Gamma) \metaseq (x:\Delta)$, and let $\mathcal{D}'$ be the result of putting $\Omega$ in every meta-sequent in $\mathcal{D}$. By Definition~\ref{def:relationalcalculus}, we have that $\mathcal{D}'$ is a $\system{G3c}(\Omega)$-proof of $\Omega,(x:\Gamma) \metaseq (x:\Delta)$. Thus, $(x:\Gamma) \proves[R(\Omega)] (x:\Delta)$ implies $\Omega, (x:\Gamma) \proves[G3c(\Omega)] (x:\Delta)$. 
\end{proof}

\begin{Example} \label{ex:RK}
    The sequent calculus in Example~\ref{ex:modal:calculus} becomes a relational calculus $\system{R}(\Omega_\sem{K})$ by suppressing $\Omega$ in the rules; for example,
    \[
        \infer{\Omega, \Pi \metaseq \Sigma}{\Omega, \Pi \metaseq \Sigma, (w:\phi \land \psi) & \Omega, (w:\phi), (w:\psi), \Pi \metaseq \Sigma}
        \]
        becomes
        \[
         \infer{\Pi \metaseq \Sigma}{\Pi \metaseq \Sigma, (w:\phi \land \psi) & (w:\phi), (w:\psi), \Pi \metaseq \Sigma}
    \]
    Abbreviating $\neg \Box \neg \phi$ by $\Diamond \phi$ and doing some proof-theoretic analysis on $\system{R}(\Omega_\sem{K})$, we have the simplified system $\system{RK}$ in Figure~\ref{fig:R} --- $\Phi$ denotes a meta-formula, $\Pi$ and $\Sigma$ denote multiset of meta-formulae, $x$ and $y$ denote world-variables, $\Delta$ denotes object-logic data, $\phi$ and $\psi$ denote object-logic formulae. This is, essentially, the relational calculus for $K$ introduced by Negri~\cite{Negri2005}.
\end{Example}

\begin{figure}[t]
\hrule
   \vspace{4mm}
	\[
	\begin{array}{cc}
		\infer[\ax]{\Phi, \Pi \seq \Sigma, \Phi}{}
		&
		\infer[\rrn \bot]{\bot, \Pi \seq \Sigma}{} \\[1.5ex]
		\infer[\lrn \land]{(x:\phi\land \psi), \Pi \seq \Sigma}{(x:\phi), (x:\psi), \Pi \seq \Sigma} 
		&
			\infer[\rrn \land]{\Pi \seq \Sigma, (x:\phi \land \psi)}{\Pi \seq \Sigma, (x:\phi) & \Pi \seq \Sigma, (x:\psi)}
		\\[1.5ex]
		\infer[\lrn \lor]{(x:\phi\lor \psi), \Pi \seq \Sigma}{(x:\phi), \Pi \seq \Sigma & (x:\psi), \Pi \seq \Sigma}
		&
		\infer[\rrn \lor]{\Pi \seq \Sigma, (x:\phi \lor \psi)}{\Pi \seq \Sigma, (x:\phi), (x:\psi)}\\[1.5ex]
		\infer[\lrn \Box]{(x:\Box \phi), xRy, \Pi \seq \Sigma}{ (y:\phi), (x:\Box \phi), xRy, \Pi \seq \Sigma}
		&
		\infer[\rrn \Box]{\Pi \seq \Sigma, (x:\Box \phi)}{xRy, \Pi \seq \Sigma, (y:\phi)} \\[1.5ex]
		\infer[\lrn \Diamond]{(x:\Diamond \phi), \Pi \seq \Sigma}{xRy, (y:\phi), \Pi \seq \Sigma} 
		&
		\infer[\rrn \diamond]{xRy, \Pi \seq \Sigma, (x:\Diamond \phi)}{xRy, \Pi \seq \Sigma, (x:\Diamond \phi), (y:\phi)} 
		\\[1.5ex]
		\infer[\lrn \bot]{(x:\bot),\Pi \seq \Sigma}{\bot, \Pi \seq \Sigma}
		&
		\infer[\rrn \bot]{\Pi \seq \Sigma, (x:\bot)}{\bot, \Pi \seq \Sigma, \bot}
  \\[1.5ex]
\infer[\lrn \fatcomma]{(x:\Gamma \fatcomma \Gamma'), \Pi \seq \Sigma}{(x:\Gamma), (x:\Gamma'), \Pi \seq \Sigma} 
&
\infer[\rrn \fatsemi]{\Pi \seq \Sigma,(x:\Delta \fatsemi \Delta')}{\Pi \seq \Sigma,(x:\Delta),(x:\Delta')} 
    \end{array}
    \]
    \vspace{4mm}
\hrule
    \caption{Relational Calculus \system{RK}}
    \label{fig:R}
\end{figure}

While we have effectively transformed (tractable) semantics into relational calculi, giving a general, uniform, and systematic proof theory to an ample space of logics, significant analysis remains to be done. In Example~\ref{ex:RK}, we showed that under relatively mild conditions, one expects the relational calculus to have a particularly good shape. This begs for further characterization of the definitions of semantics and what properties one may expect the resulting relational calculus to have; the beginnings of such an analysis are given below Definition~\ref{def:definition} in which we require $\Omega$ to contain a first-order definition of frames together with an inductive definition of the semantics. 

We have presented a general account of relational calculi, but certain specific families of logics ought to be studied in particular. For example, Negri~\cite{Negri2005} demonstrated that relational calculi for modal logics are particularly simple. An adjacent class is the family of hybrid logics --- see, for example, Blackburn et al.~\cite{Blackburn2001,blackburn1995hybrid}, Areces and ten Cate~\cite{Areces}, Br\"auner~\cite{brauner2001hybrid}, and Indrzejczak~\cite{indrzejczak2011natural}. Indeed, one may regard the meta-logics for propositional logics (i.e., the fused language) as hybrid logics --- see, for example, Blackburn~\cite{blackburn2000}.

\subsection{Faithfulness \& Adequacy} \label{sec:relationalcalculi:fna}


In this section, we give sufficient conditions for faithfulness and adequacy of a relational calculus with respect to a sequent calculus. More precisely, we give conditions under which one may transform a relational calculus into a sequent calculus for the object-logic. The result is immediate proof of soundness and completeness for the sequent calculus concerning the semantics; significantly, it bypasses term- or counter-model construction. This idea has already been implemented for the logic of Bunched Implications by Gheorghiu and Pym~\cite{Alex:BI_Semantics}.

While the work of the preceding section generates a relational calculus, one may require some proof theory to yield a relational calculus that meets the conditions in this section faithfulness and adequacy. Likewise, one may require proof theory on the generated sequent calculus to yield a sequent calculus one recognizes as sound and complete concerning a logic of interest. We do not consider these problems here, but they are addressed explicitly for BI in the previous work by the authors.

Our objective is to systematically transform (co-)inferences in the relational calculus into (co-)inferences of the propositional logic. Regarded as constraint systems, relational calculi do not have any side-conditions on inferences; instead, all of the constraints are carried within sequents. Thus we do not need to worry about assignments and aim only to develop a valuation $\nu$. We shall define $\nu$ by its action on sequents and extend it to reductions like in Section~\ref{sec:constraint:correctness}.

Fix a propositional logic $\proves$ and relational calculus $\system{R}$. We assume the propositional logic has data-constructors $\circ$ and $\bullet$ such that
\[
\mbox{$\Gamma \circ \Gamma' \proves \Delta$ \quad iff \quad  $(w:\Gamma)\with (w:\Gamma') \proves[R] (w :\Delta)$}
\]
and
\[
\mbox{$\Gamma \proves \Delta \bullet \Delta'$ \quad iff \quad  $(w:\Gamma) \proves[R] (w :\Delta) \parr (w :\Delta')$}
\]
This means that the weakening, contraction, and exchange structural rules are admissible for $\circ$ and $\bullet$ on the left and right, respectively. In particular, these data-constructors behave like classical conjunction and disjunction, respectively.

\begin{Example}
The logic with relational calculus $\system{RK}$ satisfies the data-constructor condition --- specifically, $\fatcomma$ is conjunctive and $\fatsemi$ is disjunctive. 
\end{Example}

A list of meta-formulae is \emph{monomundic} iff it only contains one world-variable (but possibly many occurrences of that world-variable; we write $\Pi^w$ or $\Sigma^w$ to denote monomundic lists contain the world-variable $w$.  A monomundic list is \emph{basic} iff it only contains meta-atoms of the form $(w:\Gamma)$, which is denoted $\bar{\Pi}^w$ or $\bar{\Sigma}^w$.

\begin{Definition}[Basic Validity Sequent]
   A basic validity sequent (BVS) is a pair of basic monomundic lists, $\bar{\Pi}^w \seq \bar{\Sigma}^w$.
\end{Definition}

\begin{Definition}[Basic Rule]
A rule in a relational calculus is basic iff it is a rule over BVSs --- that is, it has the following form:
\[
\infer[]{\bar{\Pi}^w \metaseq \Sigma^w}{\bar{\Pi}_1^{w_1} \metaseq \Sigma_1^{w_1} & ... & \bar{\Pi}_n^{w_n} \metaseq \Sigma_n^{w_n}}
\]
\end{Definition}
\begin{Definition}[Basic Relational Calculus]
 A relational calculus $\rn{r}$ is basic iff it is composed of basic rules.
\end{Definition}

Using the data-structures $\bullet$ and $\circ$, a BVS intuitively corresponds to a sequent in the propositional logic. Define $\lfloor - \rfloor_\circ$ and $\lfloor - \rfloor_\bullet$ on basic monomundic lists as follows:
\[
\lfloor (w:\Gamma_1),..., (w:\Gamma_m) \rfloor_\circ := \Gamma_1 \circ ... \circ \Gamma_m \qquad  \lfloor (w:\Delta_1),..., (w:\Delta_n)\rfloor_\bullet := \Delta_1 \bullet ... \bullet \Delta_n
\]
We can define $\nu$ on BVSs by this encoding,
\[
\nu(\bar{\Pi}^w \metaseq \bar{\Sigma}^w):= \lfloor \bar{\Pi}^w \rfloor \seq \lfloor \bar{\Sigma}^w \rfloor
\]
The significance is that whatever inference is made in the semantics using BVSs immediately yields an inference it terms of propositional sequents. 

Let $\rn{r}$ be a basic rule, its propositional encoding $\nu(\rn{r})$ is the following:
\[
\infer[]{\nu(\bar{\Pi}^w \metaseq \Sigma^w)}{\nu(\bar{\Pi}_1^{w_1} \metaseq \Sigma_1^{w_1}) & ... & \nu(\bar{\Pi}_n^{w_n} \metaseq \Sigma_n^{w_n})}
\]
This extends to basic relational calculi pointwise,
\[
\nu(\system{R}) := \{ \nu(\rn{r}) \mid \rn{r} \in \system{R}\}
\]

Despite their restrictive shape, basic rules are quite typical. For example, if the body of a clause is composed of only conjunctions and disjunctions of assertions, the rules generated by the algorithm presented above will be basic. Sets of basic rules can sometimes replace more complex rules in relational calculi to yield a basic relational calculus from a non-basic relational calculus --- see Section~\ref{sec:ex:ipl} for an example.

We are thus in a situation where the rules of a reduction system intuitively correspond to the rules of a sequent calculus. The formal statement of this is below.

\begin{Theorem}\label{thm:faq:propositionalencoding}
A basic relational calculus $\system{R}$ is faithful and adequate with respect to its propositional encoding $\nu(\system{R})$. 
\end{Theorem}
\begin{proof}
The result follows by Definition~\ref{def:fna} because a valuation of an instance of a rule in $\system{R}$ corresponds to an instance of a rule of $\system{R}$ on the states of the sequents involved. 

Faithfulness follows by application of $\nu$ on \system{R}-proofs. That is, for any $\system{R}$-reduction $\mathcal{D}$, one produces a corresponding $\nu(\mathcal{R})$-proof by apply $\nu$ to each sequent in $\mathcal{D}$. 

Adequacy follows by introducing arbitrary world-variables into a $\nu(\system{R})$-proof. Let $\mathcal{D}$ be a $\nu(\system{R})$-proof, it concludes by an inference of the following form:
\[
\infer[]{\nu(\bar{\Pi}^w \metaseq \Sigma^w)}{\deduce{\nu(\bar{\Pi}_1^{w_1} \metaseq \Sigma_1^{w_1})}{\mathcal{D}_1} & ... & \deduce{\nu(\bar{\Pi}_n^{w_n} \metaseq \Sigma_n^{w_n})}{\mathcal{D}_n}}
\]
We can co-inductively define a corresponding $\system{R}$-with the following co-recursive step in which $\mathcal{R}_i$ is the reduction corresponding to $\mathcal{D}_i$:
\[
\infer[]{\bar{\Pi}^w \metaseq \Sigma^w}{\deduce{\bar{\Pi}_1^{w_1} \metaseq \Sigma_1^{w_1}}{\mathcal{R}_1} & ... & \deduce{\bar{\Pi}_n^{w_n} \metaseq \Sigma_n^{w_n}}{\mathcal{R}_n}}
\]
Hence, for any $\nu(\system{R})$-proof, there is a $\system{R}$-reduction $\mathcal{R}$ such that $\nu(\mathcal{R})= \mathcal{D}$, as required.
\end{proof}

Of course, despite basic rules being relatively typical, many relational calculi are not comprised of \emph{only} basic rules. Nonetheless, the phenomenon does occur for even quite complex logic. It can be used for the semantical analysis of that logic in those instances  --- see, for example, Gheorghiu and Pym~\cite{Alex:BI_Semantics} for an example of this in the case of the logic of Bunched Implications. Significantly, this approach to soundness and completeness differs from the standard term-model approach and has the advantage of bypassing truth-in-a-model (i.e., satisfaction).

\section{Example: Intuitionistic Propositional Logic} \label{sec:ex:ipl}

In Section~\ref{sec:ex:relationalcalculi}, we gave a general, uniform, and systematic procedure for generating proof systems for logics that have model-theoretic semantics satisfying certain conditions. What about the reverse problem? That is, given a proof-theoretic characterization of a propositional logic, can we \emph{derive} a model-theoretic semantics for it (in the sense of Section~\ref{sec:background:propositional})? This chapter provides an example of this for intuitionistic propositional logic (IPL). 

We begin from a naive position on IPL. Our \emph{definition} of IPL is by \system{LJ} --- see Gentzen~\cite{Gentzen1969}. We choose this over other systems (e.g., $\system{G3i}$ --- see Troelstra and Schwichtenberg~\cite{troelstra2000basic}) because we assume that we do not even known much about its proof theory so that we may explain through the analysis what we require. In the end, we recover the model-theoretic semantics by Kripke~\cite{Kripke1963} using constraint systems as the enabling technology. Of course, for the purposes of this chapter we shall imagine that we do not know about the semantics.

Reflecting on Section~\ref{sec:ex:relationalcalculi}, we expect that the semantics we synthesize for IPL will be tractable. Therefore, we intend  to build a relational calculus to bridge the proof theory and semantics of IPL as in Section~\ref{sec:ex:relationalcalculi}, but this time we build it from the proof theory side. Recall that relational calculi are fragments of proof systems for FOL (i.e., the meta-logic); therefore, we begin in Section~\ref{sec:ex:IPL:bool} by building a constraint system for IPL that is \emph{classical} in shape. The system is derived in a principled way from this desire, but it is only sound and complete for IPL. We require it to be \emph{faithful} and \emph{adequate} for \system{LJ} because we hope to generate  clause governing each connective from its rules. Hence, in Section~\ref{sec:ex:IPL:fna}, we analyse the constraint system to recover a faithful and adequate constraint system for IPL. In Section~\ref{sec:ex:ipl:analysis}, we study the reductive behaviour of connectives of IPL in this constraint systems and write tractable FOL-formulae that capture the same behaviour in \system{G3c}. The resulting theory $\Omega$ determines a model-theoretic semantics for IPL, as shown in Section~\ref{sec:ex:ipl:snc}.

\subsection{Multiple-conclusions via Boolean Constraints} \label{sec:ex:IPL:bool}

We begin by defining IPL naively --- that is, from \system{LJ} (see Gentnzen~\cite{Gentzen1969}). We do it in the style of Section~\ref{sec:background:propositional} to keep it consistent with the treatment of propositional logics in the rest of this paper.

\begin{Definition}[Alphabet $\lang{J}$]
    The alphabet is $\lang{J}:= \langle \set{P}, \{\land, \lor, \to, \neg \}, \{\fatcomma, \fatsemi, \emptyset \} \rangle$, in which symbols $\land, \lor, \to, \fatcomma, \fatsemi$ have arity $2$, the symbol $\neg$ has arity $1$, and $\emptyset$ has arity $0$.
\end{Definition}

Let $\equiv$ be the smallest relation satisfying commutative monoid equations for $\fatcomma$  and $\fatsemi$ with unit $\emptyset$. 

\begin{Definition}[System \system{LJ}]
    Sequent calculus $\system{LJ}$ is comprised of the rules in Figure~\ref{fig:lklj}, in which $\Delta$ is either a $\lang{J}$-formula $\phi$ or $\emptyset$, and $\Gamma \equiv \Gamma'$ and $\Delta \equiv \Delta'$ in $\exch$.
\end{Definition}

\begin{figure}[t]
\hrule
\vspace{4mm}
\[
\begin{array}{c}
		\infer[\lrn{\weak}]{\phi\fatcomma\Gamma \seq \Delta}{\Gamma\seq\Delta}
		\quad
		\infer[\rrn{\weak}]{\Gamma \seq \phi}{\Gamma\seq\emptyset}
		\quad
		\infer[\lrn{\cont}]{\phi\fatcomma\Gamma \seq \Delta}{\phi\fatcomma\phi\fatcomma\Gamma\seq\Delta}
		\quad
		\infer[\exch]{\Gamma \seq \Delta}{\Gamma'\seq \Delta'}
		\\[1.5ex]
		\infer[\rrn\land]{\Gamma \seq \phi \land \psi}{\Gamma \seq \phi & \Gamma\seq\psi}
		\quad
		\infer[\lrn{\land^1}]{\phi \land  \psi\fatcomma\Gamma\seq \Delta}{\phi\fatcomma\Gamma\seq\Delta}
		\quad
		\infer[\lrn{\land^2}]{\phi\land \psi\fatcomma\Gamma\seq \Delta}{\psi\fatcomma\Gamma\seq\Delta}
		\quad
				\infer[\rrn \neg]{\Gamma \seq \neg \phi}{\phi\fatcomma\Gamma \seq \emptyset}
		\\[1.5ex]
		\infer[\lrn{\lor}]{\phi\lor\psi\fatcomma\Gamma \seq \Delta}{\phi\fatcomma\Gamma \seq\Delta & \psi\fatcomma\Gamma \seq\Delta }
	       \quad
		\infer[\rrn{\lor^1}]{\Gamma \seq \phi \lor  \psi}{\Gamma \seq \phi }
		\quad
		\infer[\rrn{\lor^2}]{\Gamma \seq \phi \lor  \psi}{\Gamma \seq\psi }
		\quad
				\infer[\lrn \neg]{\neg \phi \fatcomma \Gamma \seq \emptyset}{\Gamma \seq \phi}
		\\[1.5ex]
		\infer[\rrn \to]{\Gamma \seq \phi \to \psi}{\phi\fatcomma\Gamma \seq \psi}
		\quad
		\infer[\lrn \to]{\phi \to \psi\fatcomma\Gamma_1\fatcomma\Gamma_2 \seq \Delta}{\Gamma_1 \seq \phi & \psi\fatcomma\Gamma_2 \seq \Delta}
        \quad
		\infer[\ax]{\phi \seq \phi}{}
    \end{array}
    \]
    \vspace{4mm}
    \hrule
    \caption{Sequent Calculus \system{LJ}}
    \label{fig:lklj}
\end{figure}

In this section, $\system{LJ}$-provability $\proves[LJ]$ \emph{defines} the judgment relation for IPL. Our task is to derive a model-theoretic characterization of IPL. As we saw in Section~\ref{sec:ex:relationalcalculi}, the logic in which semantics is defined is classical. Therefore, our strategy is to use the constraint to present IPL in a sequent calculus with a classical shape; that is, the constraint informs precisely where the semantics of IPL diverge from those of FOL. This tells us how the semantic clauses of FOL need to be augmented to define the connectives of IPL. The calculus in question is Gentzen's \system{LK}~\cite{Gentzen1969}.

The essential point of distinction between \system{LJ} and \system{LK} is in $\rrn c$, $\lrn \to$, and $\lrn \neg$ as it is these rules that enable multiple-conclusioned sequents to appear in the latter but not the former. To bring these behaviours closer, we introduce a constraints constraint system for IPL whose combinatorial behaviour is like $\system{LK}$, but for which we have constraints to recover $\system{LJ}$.

The algebra of the constraint system is Boolean algebra --- see Section~\ref{sec:ex:bi}.

\begin{Example}\label{ex:Booleaninfusion}
The following is an enriched $\lang{J}$-sequent.
\[
(\Gamma\cdot1) \fatcomma (\phi \cdot x) \seq  (\Delta \cdot \bar{x}) \fatsemi (\psi \cdot x)
\]
\end{Example}

For the rules that are the same across both systems, the expressions are inherited from the justifying sub-formulae of the conclusion; for example, in the case of conjunction, one has the following:
\[
 \infer[\rrn \land]{\Gamma \seq \phi \land \psi}{\Gamma \seq \phi & \Gamma \seq \psi} \quad \raisebox{1ex}{ becomes } \quad \infer[\rrn{\land^\alg{B}}]{\Gamma \seq \Delta \fatsemi (\phi \land \psi\cdot x) }{ \Gamma \seq \Delta \fatsemi (\phi \cdot x) & \Gamma \seq \Delta \fatsemi (\psi \cdot x)}
\]
For readability, we may suppress the Boolean expressions on these rules.

\begin{Definition}[System $\system{LK}\oplus\mathcal{B}$]
    Sequent calculus $\system{LK}\oplus\mathcal{B}$ is comprised of the rules in Figure~\ref{fig:lkplusb}, in which $\Gamma$  and $\Delta$ are enriched $\lang{J}$-datum, and $\Gamma \equiv \Gamma'$ and $\Delta \equiv \Delta'$ in $\exch^\mathcal{B}$.
\end{Definition}

\begin{figure}[t]
\hrule
\vspace{4mm}
\[
\begin{array}{c}
\infer[\lrn{\weak^\alg{B}}]{\phi\fatcomma\Gamma \seq \Delta}{\Gamma\seq\Delta}
		\quad
		\infer[\rrn{\weak^\alg{B}}]{\Gamma \seq \Delta\fatcomma\phi}{\Gamma\seq\Delta}
		\quad
		\infer[\lrn{\cont^\alg{B}}]{\phi\fatcomma\Gamma \seq \Delta}{\phi\fatcomma\phi\fatcomma\Gamma\seq\Delta}
		\quad
		\infer[\rrn{\cont^\alg{B}}]{\Gamma \seq \Delta\fatsemi\phi}{\Gamma\seq\Delta\fatsemi\phi\cdot x \fatsemi\phi \cdot \bar{x}}
		\\[1.5ex]
		\infer[\exch^\alg{B}]{\Gamma\seq \Delta}{\Gamma' \seq \Delta'}
		\qquad
  		\infer[\ax^\alg{B}]{\phi\cdot x \seq \phi \cdot y}{x=y=1}
    \qquad
    \infer[\lrn{\neg^\alg{B}}]{\neg \phi \fatcomma \Gamma \seq \Delta}{\Gamma \seq \Delta \cdot \bar{x} \fatcomma \phi \cdot x}
				\qquad \infer[\rrn{\neg^\alg{B}}]{\Gamma \seq \Delta \fatsemi \neg \phi}{\phi \fatcomma \Gamma \seq \Delta}
    \end{array}
    \]
    \[	
    \begin{array}{c@{\qquad}c}
    	\infer[\rrn{\land^\alg{B}}]{\Gamma \seq \Delta\fatsemi\phi \land \psi}{\Gamma \seq \phi\fatsemi\Delta & \Gamma\seq\psi\fatsemi\Delta}
		&
		\infer[\rn{\land_{L_1}^{\alg{B}}}]{\phi \land  \psi\fatcomma\Gamma\seq \Delta}{\phi\fatcomma\Gamma\seq\Delta}
		\quad
		\infer[\rn{\land_{L_2}^{\alg{B}}}]{\phi\land \psi\fatcomma\Gamma\seq \Delta}{\psi\fatcomma\Gamma\seq\Delta}
		\\[1.5ex]
		\infer[\lrn{\lor^\alg{B}}]{\phi\lor\psi\fatcomma\Gamma \seq \Delta}{\phi\fatcomma\Gamma \seq\Delta & \psi\fatcomma\Gamma \seq\Delta }
	   	&
		\infer[\rn{\lor_{R_1}^\alg{B}}]{\Gamma \seq \Delta \fatsemi \phi \lor  \psi}{\Gamma \seq\Delta\fatsemi\phi }
			\quad
		\infer[\rn{\lor_{R_2}^\alg{B}}]{\Gamma \seq \Delta \fatsemi \phi \lor \psi}{\Gamma \seq\Delta\fatsemi\psi }
		\\[1.5ex]
		\infer[\rrn{\to^\alg{B}}]{\Gamma \seq \Delta\fatsemi \phi \to \psi}{\phi\fatcomma\Gamma \seq \Delta\fatsemi\psi}
		&
		\infer[\lrn{\to^\alg{B}}]{\phi \to \psi\fatcomma\Gamma_1\fatcomma\Gamma_2 \seq \Delta_1\fatsemi\Delta_2}{\Gamma_1 \seq \Delta_1 \cdot \bar{x}\fatcomma\phi \cdot x & \psi\fatcomma\Gamma_2 \seq \Delta_2}
\end{array}
  \]
  \vspace{4mm}
  \hrule
    \caption{Constraint System $\system{LK}\oplus\alg{B}$}
    \label{fig:lkplusb}
\end{figure}

The ergo rendering $\system{LK}\oplus \mathcal{B}$ sound and complete for IPL is precisely the demand for one to choose which of the formulae in the succedent of a sequent to assert as a consequence of the context. This reading is closely related to the semantics of intuitionistic proof-search provided by Pym and Ritter~\cite{Pym2005games}.

\begin{Definition}[Choice Ergo]
    Let $I:X \to \alg{B}$ be an interpretation of the language of the Boolean algebra. The choice ergo is the function $\sigma_I$ which acts on $\lang{J}$-formulae as follows:
	\[
	\sigma_{I}(\phi) \mapsto  
		\begin{cases}
		\phi  & \text{ if } \phi \text{ unlabelled } \\
		\sigma_I(\psi) & \text{ if } I(x) =1 \text{ and } \phi=\psi \cdot x \\
		\emptyset & \text{ if } I(x) = 0 \text{ and } \phi=\psi \cdot x
		\end{cases}
	\]
	The choice ergo acts on enriched $\lang{J}$-data by acting point-wise on the formulae; and it acts on enriched sequents by acting on each component independently --- that is, $\sigma_I(\Gamma \seq \Delta) = \sigma_I(\Gamma) \seq \sigma_I(\Delta)$.  
\end{Definition}

\begin{Example}
    Let $S$ be the  sequent  in Example~\ref{ex:Booleaninfusion}. If $I(x)=1$, then $\sigma_I(S) =\Gamma \fatcomma \phi \seq \psi $
\end{Example}

\begin{Proposition}
 System $\system{LK} \oplus \mathcal{B}$, with the choice ergo $\sigma$, is sound and complete for IPL,
 \[
 \Gamma \simpleproves_{\system{LK}\oplus \mathcal{B}}^\sigma \Delta \qquad \mbox{iff} \qquad  \Gamma \proves[LJ] \Delta 
 \]
\end{Proposition}

\begin{proof}[Proof of Soundness]
    Suppose $\Gamma \simpleproves_{\system{LK}\oplus \mathcal{B}}^\sigma \Delta$, then there is a coherent $\system{LK}\oplus \mathcal{B}$-reduction $\mathcal{R}$ of an enriched sequent $S$ such that $\sigma_I(S) := \Gamma \seq \Delta$, where $I$ is any assignment satisfying $\mathcal{R}$. It follows that $S$ is equivalent (up to exchange) to the sequent $\Gamma' \fatcomma \Pi \seq \Sigma \fatsemi \Delta' $ in which $\Gamma'$ and $\Delta'$ are like $\Gamma$ and $\Delta$ but with labelled data and $I$ applied to the expressions in $\Pi$ and $\Sigma$ evaluates to $0$ but applied to expressions in $\Gamma'$ and $\Delta'$ evaluates to $1$. We proceed by induction $n$ on the height of $\mathcal{R}$ --- that is, the maximal number of reductive inferences in a branch of the tree.
    
    \textsc{Base Case.} If $n=1$, then $\Gamma', \Pi \seq \Sigma, \Delta' $ is an instances of $\rn{ax}^\mathcal{B}$. But then $S= \phi \cdot x \seq \phi \cdot y$, for some formula $\phi$. We have $\phi \proves[LJ] \phi$ by $\rn{ax}$.
    
    \textsc{Inductive Step.} The induction hypothesis (IH) is as follows: if $\Gamma' \simpleproves_{\system{LK}\oplus \mathcal{B}}^\sigma \Delta'$ is witnessed by $\system{LK}\oplus \mathcal{B}$-reductions of $k \leq n$, then $\Gamma \proves[LJ] \Delta$. 
    
    Suppose that the shortest reduction witnessing $\Gamma' \simpleproves_{\system{LK}\oplus \mathcal{B}}^\sigma \Delta'$ is of height $n+1$. Let $\mathcal{R}$ be such a reduction. Without loss of generality, we assume the root of $\mathcal{R}$ is of the form $\Gamma' \fatcomma \Pi \seq \Sigma \fatsemi \Delta'$, as above. It follows by case analysis on the final inferences of $\mathcal{R}$ (i.e., reductive inferences applied to the root) that $\Gamma \proves[LJ] \Delta$. We show two cases, the rest being similar.
    
    \begin{itemize}
        \item[-] Suppose the last inference of $\mathcal{R}$ was by $\rrn{c}^\mathcal{B}$. In this case, $\mathcal{R}$ has an immediate sub-tree $\mathcal{R}'$ that is a coherent $\system{LK}\oplus \mathcal{B}$-reduction of either $\Gamma' \fatcomma \Sigma \seq \Sigma' \fatsemi \Delta'$ or $\Gamma' \fatcomma \Sigma \seq \Sigma' \fatsemi \Delta''$, in which $\Sigma'$ and $\Delta''$ are like $\Sigma$ and $\Delta'$, respectively, but with some formula repeated such that one occurrence carries an additional expression $x$ and the other occurrence with an $\bar{x}$. The coherent assignment of $\mathcal{R}$ are the same as those $\mathcal{R}'$ since the two reductions have the same constraints. We observe that under these coherent assignment $\mathcal{R}'$ witnesses $\Gamma' \simpleproves_{\system{LK}\oplus \mathcal{B}}^\sigma \Delta'$. By the IH, since $\mathcal{R}'$ is of height $n$, it follows that $\Gamma \proves[LJ] \Delta$.
        
        \item[-] Suppose the last inference of $\mathcal{R}$ was by $\rrn{\to}^\mathcal{B}$. In this case, $\mathcal{R}$ has an immediate sub-tree $\mathcal{R}'$. If the principal formula of the inference is not in $\Delta'$, then $\mathcal{R}'$ witnesses $\Gamma' \simpleproves_{\system{LK}\oplus \mathcal{B}}^\sigma \Delta'$. Hence, by the IH, we conclude  $\Gamma \proves[LJ] \Delta$. If the principal formula of the inference is in $\Delta'$, then $\mathcal{R}'$ is a proof of $\phi \fatcomma \Gamma' \fatcomma \Pi \seq \Sigma \fatsemi \Delta'' \fatsemi \psi$, where $\Delta' := \Delta'' \fatsemi \phi \to \psi$. It follows, by the IH, that $\phi, \Gamma' \proves[LJ] \Delta'' \fatsemi \psi$. By the $\rrn \to$-rule in $\system{LJ}$, we have $\Gamma \proves[LJ] \phi \to \psi$ --- that is, $\Gamma \proves[LJ] \Delta$, as required. 
    \end{itemize}
     
Of course, since we are working with $\system{LJ}$, we know that $\Delta$ contains only one formula. This distinction was not important for the proof, so we have left with the more general notation. 
\end{proof}

\begin{proof}[Proof of Completeness]
    This follows immediately from the fact that all the rules of $\system{LJ}$ may be simulated in $\system{LK} \oplus \mathcal{B}$.
\end{proof}

The point of this work is that $\system{LK} \oplus \mathcal{B}$ characterizes IPL in a way that is combinatorially comparable to FOL. This is significant as the semantics of IPL is given classically, hence $\system{LK} \oplus \mathcal{B}$  bridges the proof-theoretic and model-theoretic characterizations of IPL.

\subsection{Faithfulness \& Adequacy} \label{sec:ex:IPL:fna}

Though we may use $\system{LK} \oplus \alg{B}$ to reason about IPL with classical combinatorics, the system does not immediately reveal the meaning of the connectives of IPL in terms of their counterparts in FOL. The problem is that $\system{LK} \oplus \alg{B}$-proofs are only \emph{globally} valid for IPL, with respect to the choice ergo $\sigma$. Therefore, to conduct a semantical analysis of IPL in terms of FOL, we require a constraint system based on FOL whose proofs are \emph{locally} valid --- that is, a system which is not only sound and complete for IPL, but \emph{faithful} and \emph{adequate}. In this section, we analyze the relationship between $\system{LK} \oplus \alg{B}$ and $\system{LJ}$ to produce such a system. 

A significant difference between $\system{LK}$ and $\system{LJ}$ is the use of richer data-structures for the succedent in the former than in the latter (i.e., list or multisets verses formulae). Intuitively, the data-constructor in the succedent acts as a meta-level disjunction, thus we may investigate how $\system{LK} \oplus \mathcal{B}$ captures IPL by considering how $\rrn{\cont^\alg{B}}$ interacts with $\rrn{\lor^\alg{B}}$. We may restrict attention to interactions of the following form: 
\[
		\infer[\rrn{\cont^\mathcal{B}}]{
		    \Gamma \seq \Delta \fatsemi \phi \lor  \psi 
		    }{
		       \infer[\rrn{\lor^{\mathcal{B}_1}}]{
		        \Gamma \seq \Delta \fatsemi (\phi \lor  \psi \cdot x) \fatsemi (\phi \lor \psi \cdot \bar{x})
		    }{
		        \infer[\rrn{\lor^{\mathcal{B}_2}}]{ \Gamma \seq \Delta \fatsemi (\phi \cdot x) \fatsemi (\phi \lor \psi\cdot \bar{x})}{
		         \Gamma \seq \Delta \fatsemi (\phi \cdot x) \fatsemi (\psi \cdot \bar{x})}
		    } 
	}
\]
These may be collapsed into single inference rules, 
\[
\infer{\Gamma \seq \phi \lor \psi \cdot x \fatcomma \Delta}{\Gamma \seq \phi\cdot xy \fatsemi \psi\cdot x\bar{y} \fatsemi \Delta} 
\]
The other connectives either make use of constraints, and therefore have no significant interaction with disjunction, or simply can be permuted without loss of generality --- for example, a typical interaction between $\rrn{\cont^\mathcal{B}}$ and $\rrn{\land^\mathcal{B}}$,
\[
\infer[\rrn {\cont^\mathcal{B}}]{\Gamma \seq \Delta_1 \fatsemi \Delta_2 \fatsemi \Delta_3 \fatsemi \phi \land \psi }{
    \infer[\rrn {\land^\mathcal{B}}]{\Gamma \seq \Delta_1 \fatsemi \Delta_2 \fatsemi \Delta_3 \fatsemi \phi \land \psi \fatsemi \phi\land \psi}{
        \infer[\rrn {\exch^\mathcal{B}}]{\Gamma \seq \Delta_1 \fatsemi \Delta_2 \fatsemi \phi \land  \psi \fatsemi \phi}{
            \infer[\rrn{\land^\mathcal{B}}]{\Gamma \seq \Delta_1 \fatsemi \Delta_2 \fatsemi \phi \fatsemi \phi\land \psi}{
                \Gamma \seq \Delta_1  \fatsemi \phi \fatsemi \phi
                &
                \Gamma \seq \Delta_2 \fatsemi \phi \fatsemi \psi
                }
            }
        &
        \Gamma \seq \Delta_3 \fatsemi \psi
        }
    } 
\]
may be replaced by the derivation, which permutes the inferences,
\[
\infer[\rrn {\land^\mathcal{B}}]{\Gamma \seq \Delta_1 \fatsemi \Delta_2 \fatsemi \Delta_3 \fatsemi \phi \land \psi}{
        \infer[\rrn {\cont^\mathcal{B}}]{\Gamma \seq \Delta_1 \fatsemi \Delta_2 \fatsemi \phi}{
            \infer[\rrn {\exch^\mathcal{B}}]{\Gamma \seq \Delta_1 \fatsemi \Delta_2 \fatsemi \phi \fatsemi \phi}{
                \infer[\rrn{\exch^\mathcal{B}}]{\Gamma \seq \Delta_1 \fatsemi \phi \fatsemi \Delta_2 \fatsemi \phi}{
                    \infer[\rrn{\weak^\mathcal{B}}]{\Gamma \seq \Delta_1 \fatsemi \phi \fatsemi \phi \fatsemi \Delta_2}{\Gamma \seq \Delta_1 \fatsemi \phi \fatsemi \phi}
                    }
                }
            }
        &
        \Gamma \seq \Delta_3 \fatsemi \psi
    }
\]
This analysis allows us to eliminate $\rrn{\cont^\mathcal{B}}$ as it is captured wherever it is needed by the augmented rule for disjunction; similarly, we may eliminate $\lrn{\cont^\mathcal{B}}$ by incorporating it in the other rules. In total, this yields a new constraint system, $\system{LK}^+ \oplus \alg{B}$.

\begin{Definition}[System $\system{LK^+} \oplus \alg{B}$] System $\system{LK^+} \oplus \alg{B}$ is given in Figure~\ref{fig:glkdoubleplusB}, in which $\Gamma$  and $\Delta$ are enriched $\lang{J}$-datum, and $\Gamma \equiv \Gamma$ and $\Delta \equiv \Delta'$ in $\exch$.
\end{Definition}

\begin{figure}[t]
\hrule
\vspace{4mm}
\[
\infer[\ax]{\phi \fatcomma \Gamma \seq \Delta \fatsemi \phi}{} \quad
\infer[\lrn{\weak^\alg{B}}]{\phi\fatcomma\Gamma \seq \Delta}{\Gamma\seq\Delta}
		\quad
		\infer[\rrn{\weak^\alg{B}}]{\Gamma \seq \Delta\fatsemi\phi}{\Gamma\seq\Delta}
		\quad
		\infer[\lrn{\cont^\alg{B}}]{\phi\fatcomma\Gamma \seq \Delta}{\phi\fatcomma\phi\fatcomma\Gamma\seq\Delta}
  \quad
  \infer[\exch^\alg{B}]{\Gamma\seq \Delta}{\Gamma' \seq \Delta'}
  \]
  \[
\begin{array}{c@{\qquad}c}		\infer[\lrn{\neg^\alg{B}}]{\neg \phi \fatsemi \Gamma \seq \Delta}{\Gamma \seq \Delta \fatsemi \phi }
				&
					\infer[\rrn{\neg^\alg{B}}]{\Gamma \seq \Delta \fatsemi \neg \phi \cdot x}{\phi \cdot xy \fatcomma \Gamma \seq \Delta & xy= 1}
		\\[1.5ex]
\infer[\lrn{\land^\alg{B}}]{\phi \land \psi \fatcomma \Gamma \seq \Delta}{\phi \fatcomma \psi \fatcomma \Gamma \proves \Delta}
&
\infer[\rrn{\land^\alg{B}}]{\Gamma \seq \Delta \fatsemi \phi \land \psi}{\Gamma \seq \Delta \fatsemi \phi & \Gamma \seq \Delta \fatsemi \psi}
\\[1.5ex]
\infer[\lrn{\lor^\alg{B}}]{\phi \lor \psi \fatcomma \Gamma \seq \Delta}{\phi \fatcomma \Gamma \seq \Delta & \psi \fatcomma \Gamma \seq \Delta}
&
\infer[\rrn{\lor^\alg{B}}]{\Gamma \seq \phi \lor \psi \cdot x \fatsemi \Delta}{\Gamma \seq \phi\cdot xy \fatsemi \psi\cdot x\bar{y} \fatsemi \Delta} 
		\\[1.5ex]
  \infer[\lrn{\to^\alg{B}}]{\phi \to \psi \fatcomma \Gamma \seq \Delta}{\Gamma \seq \Delta \fatsemi \phi & \psi \fatcomma \Gamma \seq \Delta }
  &
 \infer[\rrn{\to^\alg{B}}]{\Gamma \seq \phi \to \psi \cdot x \fatsemi \Delta}{\Gamma \fatcomma \phi \cdot xy \seq \psi \cdot xy \fatsemi \Delta & xy = 1} 
\end{array}
\]
\vspace{4mm}
\hrule
      \caption{Constraint System $\system{LK^+}\oplus\alg{B}$}
    \label{fig:glkdoubleplusB}
\end{figure}

System $\system{LK^+} \oplus \alg{B}$ characterizes IPL locally --- that is, it is faithful and adequate with respect to \emph{some} sequent calculus for IPL. That sequent calculus, however, is not $\system{LJ}$, but rather a multiple-conclusioned system $\system{LJ}^+$. Essentially, $\system{LJ^+}$ is the multiple-conclusioned sequent calculus introduced by Dummett~\cite{Dummett2000} with certain instances of the structural rules incorporated into the operational rules.

\begin{Definition}[Sequent Calculus $\system{LJ^+}$]
 Sequent calculus $\system{LJ^+}$ is given by the rules in Figure~\ref{fig:ljplus}, in which $\Gamma$  and $\Delta$ are $\lang{J}$-datum, and $\Gamma \equiv \Gamma'$ and $\Delta \equiv \Delta'$ in $\exch$.
\end{Definition}

\begin{figure}[t]
\hrule
 \vspace{4mm}
 \[
 \infer[\ax]{\phi \fatcomma \Gamma \seq \Delta \fatsemi \phi}{} \quad
		\infer[\lrn{\weak}]{\phi\fatcomma\Gamma \seq \Delta}{\Gamma\seq\Delta}
		\quad
		\infer[\rrn{\weak}]{\Gamma \seq \Delta\fatsemi\phi}{\Gamma\seq\Delta}
		\quad
		\infer[\lrn{\cont}]{\phi\fatcomma\Gamma \seq \Delta}{\phi\fatcomma\phi\fatcomma\Gamma\seq\Delta}
  \quad
  \infer[\exch]{\Gamma\seq \Delta}{\Gamma' \seq \Delta'}
  \]
  \[
\begin{array}{c@{\qquad}c}
		\infer[\lrn{\neg}]{\neg \phi \fatcomma \Gamma \seq \Delta}{\Gamma \seq \Delta \fatsemi \phi }
				&
					\infer[\rrn{\neg}]{\Gamma \seq \Delta \fatsemi \neg \phi}{\phi \fatcomma \Gamma \seq \emptyset}
		\\[1.5ex]
\infer[\lrn{\land}]{\phi \land \psi \fatcomma \Gamma \seq \Delta}{\phi \fatcomma \psi \fatcomma \Gamma \seq \Delta}
&
\infer[\rrn{\land}]{\Gamma \seq \Delta \fatsemi \phi \land \psi}{\Gamma \seq \Delta \fatsemi \phi& \Gamma \seq \Delta \fatsemi \psi}
\\[1.5ex]
\infer[\lrn{\lor}]{\phi \lor \psi \fatcomma \Gamma \seq \Delta}{\phi \fatcomma \Gamma \seq \Delta & \psi \fatcomma \Gamma \seq \Delta}
&
\infer[\rrn{\lor}]{\Gamma \seq \phi \lor \psi \fatsemi \Delta}{\Gamma \seq \phi \fatsemi \psi \fatsemi \Delta} 
		\\[1.5ex]
  \infer[\lrn{\to}]{\phi \to \psi \fatcomma \Gamma \seq \Delta}{\Gamma \seq \Delta \fatsemi \phi& \psi \fatcomma \Gamma \seq \Delta }
&
 \infer[\rrn{\to}]{\Gamma \seq \phi \to \psi \fatsemi \Delta}{\Gamma \fatsemi \phi \seq \psi} 
\end{array}
\]
\vspace{4mm}
\hrule
      \caption{Sequent Calculus $\system{LJ^+}$}
    \label{fig:ljplus}
\end{figure}

\begin{Proposition} \label{lem:ljandljplus}
    Sequent calculus $\system{LJ}^+$ is sound and complete for IPL,
    \[
    \Gamma \proves[LJ] \phi  \qquad \text{iff} \qquad \Gamma \proves[LJ^+]  \phi
    \]
\end{Proposition}

\begin{proof}
    Follows from Dummett~\cite{Dummett2000}.
\end{proof}

The choice ergo $\sigma$ extends to a valuation from $\system{LK}^+\oplus \mathcal{B}$ to $\system{LJ}^+$ by pointwise application to every sequent within the reduction.

\begin{Proposition}\label{lem:fna:plussystems}
    System $\system{LK^+} \oplus \alg{B}$, with valuation $\sigma$, is faithful and adequate with respect to $\system{LJ^+}$.
\end{Proposition}

\begin{proof}
 Faithfulness follows from the observation that each rule in $\system{LK^+} \oplus \alg{B}$ produces the corresponding rule in $\system{LJ^+}$ when its constraints are observed. Adequacy follows from the observation that every instance of every rule in $\system{LJ^+}$ is an evaluation of an instance of a rules of $\system{LK^+} \oplus \alg{B}$ respecting its constraints.
\end{proof}

The force of this result is that we may use $\system{LK^+}\oplus\mathcal{B}$ to study intuitionistic connectives in terms of their classical counterparts. This analysis allows us to synthesize a model-theoretic for IPL from the semantics of FOL.

\subsection{Semantical Analysis of IPL} \label{sec:ex:ipl:analysis}

Our approach to deriving a semantics for IPL is to construct a set of meta-formulae $\Omega$ that forms a tractable definition of a semantics for IPL. The idea is that we use $\system{LK^+}\oplus\mathcal{B}$ to determine meta-formulae for each connective that simulate the proof theory of IPL within \system{LK^+}.  

Recall that in Section~\ref{sec:relationalcalculi:fna} we require a data-constructor to represent classical conjunction ($\with$) and one for classical disjunction ($\parr$). These are $\fatcomma$ and $\fatsemi$, respectively. We may now proceed to analyze the connectives of IPL.

We observe in $\system{LK^+} \oplus \alg{B}$ that intuitionistic conjunction has the same inferential behaviour as classical conjunction,
\[
    \infer[\rrn{\land^\mathcal{B}}]{\Gamma \seq \Delta \fatsemi \phi \land \psi}{\Gamma \seq \Delta \fatsemi \phi& \Gamma \seq \Delta \fatsemi \psi} 
    \qquad \mbox{vs.} \qquad
    \infer[\rrn{\with}]{\Pi \seq \Sigma, \Phi \with \Psi}{\Pi \seq \Sigma, \Phi& \Pi \seq \Sigma , \Psi} 
\]
Therefore, it seems that $\land$ in IPL should be defined as $\with$ in FOL. A candidate meta-formula governing the connective is the universal closure of the following --- we use the convention in Section~\ref{sec:relationalcalculi:tractable} in which $\hat{\phi}$ and $\hat{\psi}$ are used as meta-variables for formulae of the object-logic: 
\[
(w:\hat{\phi}\land\hat{\psi}) \iff (w:\hat{\phi})\with(w:\hat{\psi})
\]
This is the appropriate clause for the connective as it enables the following behaviour in the meta-logic in which the double-line suppresses the use of the clause:
\[
\infer={\Omega, \Pi, (w: \Gamma) \seq   (w: \phi \land \psi), \Sigma}{
    \infer[\rrn\with]{\Omega, \Pi, (w: \Gamma) \seq   (w: \phi) \with (w:\psi), \Sigma}{
        \Omega, \Pi, (w: \Gamma) \seq   (w: \phi) ,\Sigma
        &
        \Omega, \Pi, (w: \Gamma) \seq    (w:\psi), \Sigma
        }
    }
\]

Recall, such derivations correspond to the \emph{use} of the clause --- see Section~\ref{sec:relationalcalculi:generating} --- which may be collapsed into rules themselves; in this case, it becomes the following:
\[
\infer[\land\text{-clause}]{\Omega , \Pi, (w : \Gamma) \seq (w : \phi \land \psi) , \Sigma}{\Omega , \Pi, (w : \Gamma) \seq (w : \phi) , \Sigma & \Omega, \Pi, (w: \Gamma)  \seq (w : \psi), \Sigma} 
\]
Intuitively, this rule precisely recovers $\rrn \land \in \system{LJ^+}$,
\[
    \infer[\rrn{\land}]{\Gamma \seq \Delta \fatsemi \phi \land \psi}{\Gamma \seq \Delta \fatsemi \phi& \Gamma \seq \Delta \fatsemi \psi}
\]
Of course, it is important to check that the clause also has the correct behaviour in the left-hand side of sequents; we discuss this at the end of the present section. 

We obtain the universal closure of the following for the clauses governing disjunction $(\lor)$ and $(\bot)$ analogously:
\[
(x: \hat{\phi}\lor \hat{\psi}) \iff \big((x: \hat{\phi}) \parr (x: \hat{\psi})\big) \qquad (x:\bot) \iff \bot
\]

It remains to analyze implication $(\to)$. The above reasoning does not follow \emph{mutatis mutandis} because the constraints in $\system{LK^+}\oplus \mathcal{B}$ becomes germane, so we require something additional to get the appropriate simulation.

How may we express $\to$ in terms of the classical connectives? We begin by considering $\rrn{\to^\mathcal{B}} \in \system{LK^+} \oplus \mathcal{B}$,
\[
\infer[\rrn{\to^\alg{B}}]{\Gamma \seq (\phi \to \psi \cdot x) \fatsemi \Delta}{\Gamma \fatcomma (\phi \cdot xy) \seq (\psi \cdot xy) \fatsemi \Delta & xy = 1} 
\]
Since $\system{LK^+} \oplus \mathcal{B}$ is not only sound and complete for IPL but faithful and adequate, we know that this rule characterizes the connective. The rule admits two assignment classes: $x \mapsto 0$ or $x \mapsto 1$. Super-imposing these valuations can capture the behaviour we desire in the meta-logic on each other by using \emph{possible worlds} to distinguish the possible cases, 
\[
\infer{\Omega , \Pi \seq  (w : \phi \to \psi), \Sigma}{\Omega , \Pi[w \mapsto u], \Pi[w \mapsto v],(u: \phi) \seq (u : \psi), \Sigma[w\mapsto v], \Sigma[w\mapsto u]}
\]
We assume that since $u$ and $v$ are distinct, they do not interact so that the rule captures the following possibilities:
\[
\infer{\Omega, \Pi \seq  (w : \phi \to \psi), \Sigma}{\Omega , \Pi[w \mapsto u],(u: \phi) \seq (u:\psi), \Sigma[w\mapsto u]} 
\qquad
\infer{\Omega, \Pi \seq  (w : \phi \to \psi), \Sigma}{\Omega , \Pi[w \mapsto v] \seq \Sigma[w\mapsto v]}
\]
The assumption is proved valid below --- see Proposition~\ref{lem:worldindependence}. Applying the state function to these rules does indeed recover the possible cases of $\rrn{\to^\mathcal{B}}$, which justifies that this super-imposing behaviour is what we desire of the clause governing implication. It remains only to find that clause.

One of these possibilities amounts to a weakening, a behaviour already present through interpreting the data-structures as classical conjunction and disjunction.
The other possibility we recognize as having the combinatorial behaviour of classical implication concerns creating a meta-formula in the antecedent of the premiss by taking part in a meta-formula in the succedent of the conclusion. Naively, we may consider the following as the clause:
\[
(w:\hat{\phi} \to \hat{\psi}) \iff \big( (w:\hat{\phi}) \Rightarrow (w:\hat{\psi}) \big)
\]
However, this fails to account for the change in world. Thus, we require the clause to have a universal quantifier over worlds and a precondition that enables the $\Pi[w \mapsto u]$  substitution. Analyzing the possible use cases, we observe that $R$ must satisfy reflexivity so that the substitution for $u$ may be trivial (e.g., when validating $(w:\phi \land (\phi \to \psi)) \seq (w: \psi)$). In total, we have the universal closure of the following meta-formulae:
\[
\begin{array}{c}
(x:\hat{\phi} \to \hat{\psi})  \iff  \forall y\big((x R y) \with (y:\hat{\phi}) \Rightarrow (y: \hat{\psi})\big)\\
x R x \qquad 
x R y \Rightarrow \forall \hat{\Gamma} \big((x:\hat{\Gamma}) \Rightarrow (y:\hat{\Gamma})\big)
\end{array}
\]
Observe that we have introduced an ancillary relation $R$ precisely to recover the behaviour determined by the algebraic constraints; curiously, we do not need transitivity, which would render $R$ a pre-order and recover Kripke's semantics for IPL~\cite{Kripke1963} (we discuss this further at the end of Section~\ref{sec:ex:ipl:snc}). Moreover, since the data-constructors behave exactly as conjunction $(\land)$ and disjunction ($\lor)$, we may replace $\hat{\Gamma}$ with $\hat{\phi}$ without loss of generality.

This concludes the analysis. Altogether, the meta-formulae thus generated comprise a tractable definition for a model-theoretic semantics for IPL, called $\Omega_{\text{\rm IPL}}$. Any abstraction of this theory gives the semantics.

\begin{Definition}[Intuitionistic Frame, Satisfaction, and Model]
    An intuitionistic frame is a pair $\alg{F} := \langle \uni, R \rangle$ in which  $R$ is a reflexive relation on $\uni$.

    Let $ \llbracket - \rrbracket$ be an interpretation mapping  \lang{J}-atoms to $\set{U}$. Intuitionistic satisfaction is the relation between elements $w \in \uni$ and $\phi \in \set{F}$ defined by the clauses of Figure~\ref{fig:sat:ipl}. 
    
    A pair $\langle \mathcal{F }, \llbracket - \rrbracket \rangle$ is an intuitionistic model iff it is persistent --- that is, for any $\lang{J}$-formula $\phi$ and worlds $w$ and $v$, 
    \[
    \mbox{if $w R v$ and $w \sat\phi$, then $v \sat \phi$}
    \]
    The class of all intuitionistic models is $\set{K}$.
\end{Definition}

\begin{figure}
 \hrule
  \vspace{4mm}
    \[
    \begin{array}{l@{\qquad}c@{\qquad}l}
        w \satisfies \at{p} & \text{iff}  & w \in \ll \at{p} \rr \\[1ex]
         w \satisfies \phi \land \psi & \text{iff} & w \satisfies \phi \text{ and } w \satisfies \psi \\[1ex]
         w \satisfies  \phi \lor \psi & \text{ iff } & w \satisfies \phi \text{ or } w \satisfies \psi \\[1ex]
         w \satisfies \phi \to \psi & \text{ iff } & \text{ for any $u$, if $w R u$ and $u \satisfies \phi$, then $u \satisfies \psi$ } \\[1ex]
         w \satisfies \bot & \text{  } &  \text{never} 
    \end{array}
    \]
    \vspace{4mm}
   \hrule
\caption{Satisfaction for IPL} \label{fig:sat:ipl}
\end{figure}

This semantics generates the following validity judgment:
\[
\Gamma \entails_{\text{IPL}} \Delta \qquad \mbox{ iff } \qquad \mbox{for any $\mathfrak{M} \in \set{K}$ and any $w \in \mathfrak{M}$, if $w \sat \Gamma$, then $w \sat\Delta$}
\]

It remains to prove soundness and completeness for the semantics, which we do in Section~\ref{sec:ex:ipl:snc}. Of course, we have \emph{designed} the semantics so that it corresponds to $\system{LJ^+}$, rendering the proof a formality. Nonetheless, it is instructive to see how it unfolds. 

As a remark, \emph{tertium non datur} is known not to apply in IPL. How does encoding of IPL within classical logic avoid it? It is instructive to study this question as it explicates the clause for implication, which defines the intuitionistic connective in terms of a (meta-level) classical one.

\begin{Example} \label{ex:lemfails}
The following reduction is a canonical instances of \emph{using} the clause (see Section~\ref{sec:relationalcalculi:generating}):

\begin{prooftree}
\AxiomC{$\Omega_{\text{IPL}}, (w R u), (u : \phi) \seq (w : \phi),  \bot$}
\RightLabel{$\lrn{\with}$}
\UnaryInfC{$\Omega_{\text{IPL}}, (w R u \with u : \phi) \seq (w : \phi),  \bot$ }
\RightLabel{$\rrn{\Rightarrow}$}
\UnaryInfC{$\Omega_{\text{IPL}}   \seq (w : \phi), (w R u \with u  : \phi \Rightarrow \bot)$}
\RightLabel{$\rrn{\forall}$}
\UnaryInfC{$\Omega_{\text{IPL}}   \seq (w : \phi), \forall x (w R x \with x : \phi \Rightarrow \bot)$}
\RightLabel{$\to$-clause}
\UnaryInfC{$\Omega_{\text{IPL}}   \seq (w : \phi), (w : \phi \to \bot) $}
\RightLabel{$\rrn \parr$}
\UnaryInfC{$\Omega_{\text{IPL}}  \seq (w : \phi) \parr (w : \phi \to \bot)$}
\RightLabel{$\lor$-clause}
\UnaryInfC{$\Omega_{\text{IPL}} \seq (w : \phi \lor \phi \to \bot)$}
\end{prooftree}

Since $(u: \phi)$ in the antecedent and $(w: \phi)$ in the succedent are different atoms, since $u$ and $w$ are different world-variables, one has not reached an axiom. In short, despite working in a classical system, the above calculation witnesses that $\phi \lor \neg \phi$ is valid in IPL if and only if one already knows that $\phi$ is valid in IPL or one already knows that $\neg \phi$ is valid in IPL. 
\end{Example}

Curiously, in all instances above, we established the clause for a connective by analyzing the behaviour of the connective in the succedent (i.e., on the right-hand side of sequents), paying no attention to its behaviour in the antecedent (i.e., in the left-hand side of sequents). Nevertheless, despite this bias in all the above cases, the clauses generated behave correctly on both sides of the sequents. This observation is significant because it reaffirms an impactful statement by Gentzen~\cite{Gentzen1969}:
\begin{quote}
    The introductions represent, as it were, the ‘definitions’ of the symbols concerned, and the eliminations are no more, in the final analysis, than the consequences of these definitions. This fact may be expressed as follows: In eliminating a symbol, we may use the formula with whose terminal symbol we are dealing only `in the sense afforded it by the introduction of that symbol.'
    \end{quote}
This statement is one of the warrants for proof-theoretic semantics~\cite{SEP-PtS}, a mathematical instantiation of \emph{inferentialism} --- the semantic paradigm according to which inferences and the rules of inference establish the meaning of expressions. The method for semantics presented here supports the inferentialist view as it precisely determines the meaning of the connectives according to their inferential behaviour. Notably, that the left- and right- behaviours cohere is intuitively a consequence of \emph{harmony}, which is a pre-condition for proof-theoretic semantics --- see, for example, Schroeder-Heister~\cite{SchroedHeister}.

\subsection{Soundness \& Completeness} \label{sec:ex:ipl:snc}

There are two relationships the proposed semantics may have with IPL: soundness and completeness. Since the semantics generated is the same as the one given by Kripke~\cite{Kripke1963}, both of these properties are known to hold. Nonetheless, the method by which the semantics was determined gives a different method for establishing these relationships --- namely, those of Section~\ref{sec:ex:relationalcalculi}. In summary, we need only show that relational calculus generated by the semantics has the same behaviour (i.e., is faithful and adequate) as a sequent calculus for IPL. This is a formalized reading of the traditional approach to soundness in which one demonstrates that every rule of the system can be simulated by the clauses of the semantics --- in the parlance of this paper, this is just showing that the relational calculus generated by the semantics is faithful to the sequent calculus. 

\begin{Theorem}[Soundness]
    If $\Gamma \proves \phi$, then $\Gamma \entails_{\text{\rm IPL}} \phi$.
\end{Theorem}

\begin{proof}
Apply the traditional inductive proof --- see, for example, Van Dalen~\cite{vanDalen}.
\end{proof}

In this paper, we shall prove completeness symmetrically. We show that the relational calculus generated by the semantics is adequate for a sequent calculus characterizing IPL. This suffices because the relational calculus is sound and complete for the semantics, as per Theorem~\ref{thm:snc:relationalcalculi}. To simplify the presentation, we shall have contraction explicit in the relational calculus for $\Omega_{\text{IPL}}$ rather than implicit. 

\begin{Definition}[Relational Calculus \system{RJ}]
  The relational calculus $\system{RJ}$ is comprised of the rules in Figure~\ref{fig:rj} --- $\Phi$ denotes a meta-formula, $\Pi$ and $\Sigma$ denote multiset of meta-formulae, $x$ and $y$ denote world-variables, $\Delta$ denotes object-logic data, $\phi$ and $\psi$ denote object-logic formulae. The rules $\lrn{\fatcomma}$ and $\rrn{\fatsemi}$ are invertible, and the world-variable $y$ does not appear elsewhere in the sequents in $\rrn{\to}$.
\end{Definition}

\begin{figure}[t]
  \hrule
  \vspace{4mm}
	\[
	\begin{array}{c@{\qquad}c}
  \infer[\cont]{\Phi, \Pi \seq \Sigma}{\Phi, \Phi, \Pi \seq \Sigma}
  \qquad
  \infer[\taut]{\Phi, \Pi \seq \Sigma, \Phi}{}
		&
		\infer[\bot]{\bot, \Pi \seq \Sigma}{} \qquad \infer[\rn{ref}]{\Pi \seq \Sigma, x R x}{} \\[1.5ex]
		\infer[\lrn \land]{(x:\phi\land \psi), \Pi \seq \Sigma}{(x:\phi), (x:\psi), \Pi \seq \Sigma} &
			\infer[\rrn \land]{\Pi \seq \Sigma, (x:\phi \land \psi)}{\Pi \seq \Sigma, (x:\phi) & \Pi \seq \Sigma, (x:\psi)}
		\\[1.5ex]
		\infer[\lrn \lor]{(x:\phi\lor \psi), \Pi \seq \Sigma}{(x:\phi), \Pi \seq \Sigma & (x:\psi), \Pi \seq \Sigma}
		&
		\infer[\rrn \lor]{\Pi \seq \Sigma, (x:\phi \lor \psi)}{\Pi \seq \Sigma, (x:\phi), (x:\psi)}
		\end{array}
		\]
		\[
		\infer[\lrn \to]{ (x:\phi \to \psi), \Pi \seq \Sigma}{ \Pi \seq \Sigma, (x R y) & \Pi \seq \Sigma, (y:\phi) & (y:\psi),\Pi \seq \Sigma} 
		\]
		\vspace{-6.5ex}
		\[
		\begin{array}{cc}
		\\[1.5ex]
		\infer[\rrn \to]{ \Pi \seq \Sigma, (x:\phi \to \psi)}{(xRy),(y:\phi),\Pi \seq \Sigma, (y:\psi)}
		&
		\infer[\rn{pers}]{(x R y),(x:\Gamma), \Pi \seq \Sigma}{(x R y),(x:\Gamma),(y:\Gamma), \Pi \seq \Sigma}
		\\[1.5ex]
	\infer[\lrn \bot]{(x:\bot), \Pi \seq \Sigma}{\bot, \Pi \seq \Sigma}
	&
	\infer[\lrn \bot]{\Pi \seq \Sigma, (x:\bot)}{\Pi \seq \Sigma, \bot}
		\\[1.5ex]
			\infer[\lrn \fatcomma]{(x:\Gamma \fatcomma \Gamma'), \Pi \seq \Sigma}{(x:\Gamma), (x:\Gamma'), \Pi \seq \Sigma} 
			&
				\infer[\rrn \fatsemi]{\Pi \seq \Sigma, (x:\Delta \fatsemi \Delta')}{\Pi \seq \Sigma, (x:\Delta), (x:\Delta')}
    \end{array}
    \]
   \vspace{4mm}
   \hrule
    \caption{Relational Calculus \system{RJ}}
    \label{fig:rj}
\end{figure}

\begin{Corollary} \label{cor:IPLandRJ}
$\Gamma \entails_{\text{\rm IPL}} \Delta$ iff $\Gamma \simpleproves_{\system{IPL}}^\sigma \Delta$
\end{Corollary}

\begin{proof}
Instance of Theorem~\ref{thm:snc:relationalcalculi}.
\end{proof}

We desire to transform $\system{RJ}$ into a sequent calculus for which it is adequate, which we may then show is a characterization of IPL. Such transformations are discussed in Section~\ref{sec:relationalcalculi:fna}, but the rules of \system{IPL} are slightly too complex for the procedure of that section to apply immediately. Therefore, we require some additional meta-theory.

The complexity comes from the $\to$-clause as it 
 may result in non-BVSs. However, we immediately use persistence to create a composite behaviour that a basic rule can capture. This is because the combined effect yields BVSs whose contents may be partitioned; by design, persistence uses world-variables that do not, and cannot, interact throughout the rest of the proof.

\begin{Definition}[World-independence]
 Let $\Pi$ and $\Sigma$ be lists of meta-formulae. The lists $\Pi$ and $\Sigma$ are world-independent iff the set of world-variable in $\Pi$ is disjoint from the set of world-variables in $\Sigma$.
 
 Let $\sem{S}$ be a tractable semantics and let $\Omega$ be a tractable definition of it. Let $\Pi_1,\Sigma_1$ and $\Pi_2,\Sigma_2$ be world-independent lists of meta-formulae. The semantics $\sem{S}$ has world-independence iff, if $\Omega, \Pi_1, \Pi_2 \metaconsequence \Sigma_1, \Sigma_2$, then either $\Omega, \Pi_1 \metaconsequence \Sigma_1$ or $\Omega, \Pi_2 \metaconsequence  \Sigma_2$.
\end{Definition}

Intuitively, world-independence says that whatever is true at a world in the semantics does not depend on truth at a world not related to it. 

Let $\Pi_i^1,\Pi_i^2,\Sigma_i^1$, and $\Sigma_i^2$ be lists of meta-formulae, for $1 \leq i \leq n$, and suppose that $\Pi_i^1,\Sigma_i^1$ is world-independent from $\Pi_i^2,\Sigma_i^2$. Consider a rule of the following form:
 \[
\infer{\Pi \metaseq \Sigma}{\Pi_1^1,\Pi_1^2 \metaseq \Sigma_1^1,\Sigma_1^2 & ... & \Pi_n^1,\Pi_n^2 \metaseq \Sigma_n^1,\Sigma_n^2 }
\]
Assuming world-independence of the semantics, this rule can be replaced by the following two rules:
 \[
\infer{\Pi \metaseq \Sigma}{\Pi_1^1 \metaseq \Sigma_1^1 & ... & \Pi_n^1 \metaseq \Sigma_n^1 }
\qquad
\infer{\Pi \metaseq \Sigma}{\Pi_2^1 \metaseq \Sigma_2^1 & ... & \Pi_n^2 \metaseq \Sigma_n^2 }
\]
If all the lists were basic, iterating these replacements may eventually yield a set of basic rules with the same expressive power as the original rule.

\begin{Proposition} \label{lem:worldindependence}
    The semantics of IPL --- that is, the semantics $\langle \set{K}, \satisfies \rangle$ defined by $\Omega_{\text{\rm IPL}}$ --- has world-independence.
\end{Proposition}
\begin{proof}
If $\Omega_{\text{IPL}}, \Pi_1, \Pi_2 \proves \Sigma_1, \Sigma_2$, then there is a $\system{G3}$-proof $\mathcal{D}$ of it. We proceed by induction on the number of resolutions in such a proof.

\textsc{Base Case.} Recall, without loss of generality, an instantiation of any clause from $\Omega_{\text{IPL}}$ is a resolution. Therefore, if $\mathcal{D}$ contains no resolutions, then $\Omega_{\text{IPL}}, \Pi_1, \Pi_2 \proves \Sigma_1, \Sigma_2$ is an instance of $\taut$. In this case, either $\Omega_{\text{IPL}}, \Pi_1 \proves \Sigma_1$ or $\Omega_{\text{IPL}}, \Pi_2 \proves \Sigma_2$ is also an instance of $\taut$, by world-independence.

\textsc{Induction Step.} After a resolution of a sequent of the form $\Omega_{\text{IPL}}, \Pi_1, \Pi_2 \proves \Sigma_1, \Sigma_2$, one returns a meta-sequent of the same form --- that is, a meta-sequent in which we may partition the meta-formulae in the antecedent and succedent into world-independent multisets. This being the case, the result follows immediately from the induction hypothesis. 

The only non-obvious case is in the case of a closed resolution using the $\to$-clause in the antecedent because they have universal quantifiers that would allow one to produce a meta-atom that contains both a world from $\Sigma_1, \Pi_1$ and $\Sigma_2, \Pi_2$ simultaneously, thereby breaking world-independence.

Let $\Pi_1 = \Pi_1', (w \satisfies \phi \to \psi)$ and suppose $u$ is a world-variable appearing in $\Sigma_2, \Pi_2$. Consider the following computation --- for readability, we suppress $\Omega_{\text{IPL}}$:

\begin{scprooftree}{0.95}
\AxiomC{$ \Pi_1', \Pi_2 \seq \Sigma_1, \Sigma_2, w R u$}
\AxiomC{$ \Pi_1', \Pi_2 \seq \Sigma_1, \Sigma_2, (u : \phi) $}
\RightLabel{$\lrn \with$}
\BinaryInfC{$\Pi_1', \Pi_2 \seq \Sigma_1, \Sigma_2, (w R u) \with (u : \phi)$}
\AxiomC{$ \Pi_1', \Pi_2, (u :\psi) \seq \Sigma_1, \Sigma_2 $}
\doubleLine
\RightLabel{$\lrn{\Rightarrow }$}
\BinaryInfC{$ \Pi_1', (w R u \with (u : \phi) \Rightarrow u : \psi), \Pi_2 \seq \Sigma_1, \Sigma_2$}
\RightLabel{$\lrn \forall$}
\UnaryInfC{$\Pi_1', \forall x(w R x \with x : \phi \Rightarrow x : \psi), \Pi_2 \seq \Sigma_1, \Sigma_2$}
\RightLabel{$\to$-clause}
\UnaryInfC{$\Pi_1', (w : \phi \to \psi), \Pi_2 \seq \Sigma_1, \Sigma_2$}
\end{scprooftree}

The $w R u$ may be deleted (by $\lrn \weak$) from the leftmost premiss because the only way for the meta-atom to be used in the remainder of the proof is if $w R u$ appears in the context, but this is impossible (by world-independence). Hence, without loss of generality, this branch reduces to $\Pi_1', \Pi_2 \proves \Sigma_1, \Sigma_2$. Each premiss now has the desired form.
\end{proof}

Using world-independence, we may give a relational calculus $\system{RJ^+}$ characterizing the semantics comprised of basic rules. It arises from analyzing the r\^ole of the atom $x R y$ in \system{RJ} in an effort to get rid of it. Essentially, we incorporate it in $\rrn \to$, which was always its purpose --- see Section~\ref{sec:ex:ipl:analysis}.

\begin{Definition}[System \system{RJ^+}] System $\system{RJ^+}$ is comprised of the rules in Figure~\ref{fig:rjplus}, in which $\lrn{\fatcomma}$ and $\rrn{\fatsemi}$ are invertible.
\end{Definition}

\begin{figure}[t]
\hrule
  \vspace{4mm}
	\[
	\begin{array}{c@{\qquad}c}
\infer[\cont]{\Phi, \Pi \seq \Sigma}{\Phi, \Phi, \Pi \seq \Sigma}
  \qquad
  \infer[\taut]{\Phi, \Pi \seq \Sigma, \Phi}{}
		&
		\infer[\bot]{\bot, \Pi \seq \Sigma}{} \qquad \infer[\rn{ref}]{\Pi \seq \Sigma, x R x}{} \\[1.5ex]
		\infer[\lrn \land]{(x:\phi\land \psi), \Pi \seq \Sigma}{(x:\phi), (x:\psi), \Pi \seq \Sigma} &
			\infer[\rrn \land]{\Pi \seq \Sigma, (x:\phi \land \psi)}{\Pi \seq \Sigma, (x:\phi) & \Pi \seq \Sigma, (x:\psi)}
		\\[1.5ex]
		\infer[\lrn \lor]{(x:\phi\lor \psi), \Pi \seq \Sigma}{(x:\phi), \Pi \seq \Sigma & (x:\psi), \Pi \seq \Sigma}
		&
		\infer[\rrn \lor]{\Pi \seq \Sigma, (x:\phi \lor \psi)}{\Pi \seq \Sigma, (x:\phi), (x:\psi)}
		\\[1.5ex]
		\infer[\lrn \to]{ (x:\phi \to \psi), \Pi \seq \Sigma}{ \Pi \seq \Sigma, (x:\phi) & (x:\psi),\Pi \seq \Sigma} 
		&
		\infer[\rrn \to]{ \Pi \seq \Sigma, (x:\phi \to \psi)}{(y:\phi),\Pi[x \mapsto y] \seq (x:\psi)}
		\\[1.5ex]
	\infer[\lrn \bot]{(x:\bot), \Pi \seq \Sigma}{\bot, \Pi \seq \Sigma}
	&
	\infer[\rrn \bot]{\Pi \seq \Sigma, (x:\bot)}{\Pi \seq \Sigma, \bot}
		\\[1.5ex]
			\infer[\lrn \fatcomma]{(x:\Gamma \fatcomma \Gamma'), \Pi \seq \Sigma}{(x:\Gamma), (x:\Gamma'), \Pi \seq \Sigma} 
			&
				\infer[\rrn \fatsemi]{\Pi \seq \Sigma, (x:\Delta \fatsemi \Delta')}{\Pi \seq \Sigma, (x:\Delta), (x:\Delta')}
    \end{array}
    \]
  \vspace{4mm}
  \hrule
    \caption{Relational Calculus \system{RJ^+}}
    \label{fig:rjplus}
\end{figure}

\begin{Proposition} \label{lem:rjpluscomp}
    $\Gamma \proves[RJ] \Delta$ iff $\Gamma \proves[RJ^+] \Delta$
\end{Proposition}
\begin{proof}
     Every $\system{RJ}$-proof can be simulated in $\system{RJ}^+$ by using  (i.e., reduce with) $\rn{pers}$ eagerly after using $\lrn \to$, and by Thus, $\Gamma \proves[RJ] \Delta$ implies $\Gamma \proves[RJ^+] \Delta$. It remains to show that $\Gamma \proves[RJ^+] \Delta$ implies $\Gamma \proves[RJ] \Delta$. 
     
    Without loss of generality, in $\system{RJ^+}$ one may always use $\rn{pers}$ immediately after $\rrn \to$, as otherwise the use of $\lrn \to$ could be postponed. Similarly, without loss of generality, $\lrn \to$ always instantiates with $y = x$ --- this follows as we require the leftmost branch of the following to close, which it does by $\rn{ref}$:
    \begin{prooftree}
    \AxiomC{$\Pi \seq \Sigma, (xR x)$}
    \AxiomC{$\Pi \seq \Sigma, (x:\phi)$}
    \AxiomC{$\Pi, (x:\psi) \seq \Sigma$}
    \TrinaryInfC{$\Pi, (x:\phi \to \psi) \seq \Sigma$}
    \end{prooftree}
     A $\system{RJ^+}$-proof following these principles maps to a $\system{RJ}$-proof simply by collapsing the instances of $\lrn\to$ and $\rrn \to$ in the former to capture $\lrn\to$ and $\rrn \to$ in the latter.
\end{proof}

Observe that the propositional encoding of $\system{RJ^+}$ is precisely $\system{LJ^+}$. The connexion to IPL follows immediately:

\begin{Corollary} \label{cor:rjplusandljplus}
    System $\system{RJ^+}$ is faithful and adequate with respect to $\system{LJ^+}$.
\end{Corollary}
\begin{proof}
Instance of Theorem~\ref{thm:faq:propositionalencoding}.
\end{proof}

\begin{Theorem}[Completeness]
    If $\Gamma \entails_{\text{\rm IPL}}\Delta$, then $\Gamma \proves \Delta$.
\end{Theorem}
\begin{proof}
We have the following:
\[
\begin{array}{lllr}
     \hspace{2cm} \Gamma \entails_{\text{\rm IPL}}\Delta & \text{implies} & (w:\Gamma) \proves[RJ] (w:\Delta) & \hspace{1cm} \text{(Corollary~\ref{cor:IPLandRJ})} \\
     & \text{implies} & (w:\Gamma) \proves[RJ^+] (w:\Delta) & \text{(Proposition~\ref{lem:rjpluscomp})} \\
     & \text{implies} & \Gamma \proves[LJ^+] \Delta & \text{(Corollary~\ref{cor:rjplusandljplus})} \\
     & \text{implies} & \Gamma \proves[LJ] \Delta & \text{(Proposition~\ref{lem:ljandljplus})}
\end{array}
\]
Since $\system{LJ}$ characterizes IPL, this completeness the proof.
\end{proof}

Thus, we have derived a semantics of IPL for $\system{LJ}$ and proved its soundness and completeness using the constraint systems. This semantics is not quite Kripke's one~\cite{Kripke1963}, which insists that $R$ be transitive, thus rendering it a pre-order. This requirement is naturally seen from the connexion to Heyting algebra and the modal logic S4. In the analysis of Section~\ref{sec:ex:ipl:analysis}, from which the semantics in this paper comes, there was no need for transitivity; the proofs of soundness and completeness go through without it. 

One may add transitivity --- that is, the meta-formula $\forall x,y,z(xRy \with yRz \implies xRz)$ --- to $\Omega_{\text{IPL}}$ and proceed as above, adding the following rule to $\system{RJ}$:
\[
\infer{xRy,yRz,\Pi \seq \Sigma}{xRy,yRz,xRz,\Pi \seq \Sigma}
\]
The proof of completeness passes again through $\system{RJ^+}$ by observing that eagerly using persistence does all the work required of transitivity; that is, according to the eager use of $\rn{pers}$, in a sequent $xRy,yRz,\Pi \seq \Sigma$, the set $\Pi$ is of the form $\Pi'[x \mapsto y] \cup \Pi'[y \mapsto z]$, so that whatever information was essential about (the world denoted by) $x$ is already known about (the world denoted by) $z$ by passing through (the world denoted by) $y$. 

This concludes the case-study of IPL: we have used constraint systems to decompose it into classical logic from which we derived a semantics and proved soundness and completeness. Adding other axioms to $\Omega_{\text{IPL}}$, ones that are not redundant in the above sense recovers various intermediate logics. In this way, constraint systems offer a uniform and modular approach to studying them, which we leave as future work.

\section{Extension: First-order Logic} \label{sec:ext:fol} 

So far, we have concentrated on propositional logics to enable a uniform account of constraint systems across a large class of logics. Nevertheless, there is nothing within the paradigm that is inherently propositional. In this section, we extend the phenomenon of decomposing a logical system according to its combinatorial aspect and algebraic aspects to the setting of first-order logic, as in the case of the original example of a constraint system (i.e., RDvBC in Section~\ref{sec:ex:bi}), constraint systems have computational advantages. Therefore, we illustrate the extension to first-order logic by application to \emph{logic programming}.

Logic Programming (LP) is the programming language paradigm whose operational semantics is based on proof-search (see, for example, Miller et al.~\cite{Miller91}). It is a core discipline in (symbolic) artificial intelligence as proof-search is used to characterize reasoning. The central part of LP is a step known as \emph{resolution}, and the most challenging aspect of resolution is a process called \emph{unification}; the output of the execution of a configuration in LP is a \emph{unifier}. A resolution in LP is a reductive application of a quantifier left-rule combined with an implication left-rule in a sequent calculus; unification is the choice of substitution in the quantifier rule. 

In this section, we show how algebraic constraints may handle unification. The authors have discussed this idea in earlier work~\cite{Samsonschrift}, but in a limited way, as the underlying framework of algebraic constraints had yet to be developed. 

\subsection{A Basic Logic Programming Language} \label{sec:ext:fol:blp}

The Basic Logic Programming language (BLP) is based on uniform proof-search in the hereditary Harrop fragment of intuitionistic logic (see Miller et al.~\cite{Miller91,Miller1989}), which we think of as the \emph{basic logic} (B).
	
	Fix a first-order alphabet  $\lang{A} = \langle \set{R}, \set{F}, \set{K}, \set{V} \rangle$. Denote the set of atoms by $\ATOMS$.
	
	\begin{Definition}[Goal formula, Definite clauses, Programs, Query]
    	The set of goal formulae $\set{G}$ and definite clauses $\set{D}$ are the sets of formula $G$ and $D$ defined as follows, respectively:
    	\[
		\begin{array}{rcl}
		G & ::= &   A \in \ATOMS \mid  G \mid D \to G \mid G \wedge G \mid G \vee G \\
		D & ::= &  A \in \ATOMS \mid  D \mid G \to A \mid D \wedge D 
		\end{array}
		\]
		 A set $P$ of definite formula is a \emph{program}. A query is a pair $P \seq G$ in which $P$ is a program and $G$ is goal-formula.
	\end{Definition}
	
There is an apparent lack of quantifiers in the language. However, restricting attention to goal formulae and definite clauses means the quantifiers can be suppressed without loss of information: definite clauses containing a variable as universally quantified, and goal formulae containing a variable as existentially quantified.

\begin{Definition}[Substitution]
    A substitution is a mapping $\theta : \set{V} \to \set{T}$. If $\phi$ is a formula (i.e., a definite clause or a goal), then $\phi\theta$ is the formula that results from replacing every occurrence of a variable in $\phi$ by its image under $\theta$.
\end{Definition}

Having fixed a program $P$, one gives the system a goal $G$ with the desire of finding a substitution $\theta$ such that $P \proves G\theta$ obtains in IL. The substitution $\theta$ is the unifer. It is the object that the language $\system{BLP}$ computes upon an input query. The operational semantics for BLP is given by proof-search in \system{LB}.
	
	\begin{Definition}[Sequent Calculus $\system{LB}$]
	   Sequent calculus \system{LB} comprises the rules of Figure~\ref{fig:blp}.
	\end{Definition}
	
		\begin{figure}[t] 
		\hrule
  \vspace{4mm}
  \[
  \begin{array}{c}
				\infer[\rrn \to]{P \seq D \to G}{P\fatcomma D \seq G} \qquad 
				\infer[\rrn \land]{P \seq G_0 \land G_1}{P \seq G_0 & P \seq G_1} \qquad 
				\infer[\rrn \lor]{P \seq G_0 \lor G_1}{P \seq G_i}
				\qquad
			 \infer[\ax]{P\fatcomma A \seq A}{} 
			
				\\[1.5ex]
			
				\infer[\lrn \to ]{ P\fatcomma G \to A \seq A}{P\fatcomma G \to A \seq G} \qquad 
				\infer[\lrn \land ]{ P\fatcomma D_0 \land D_1 \seq A}{P\fatcomma D_0\fatcomma D_1 \seq A} \qquad 
				\infer[\lrn \forall ]{ P\fatcomma D \seq A}{P\fatcomma D\fatcomma D\theta \seq A}
				\qquad
					\infer[\rrn \exists]{P \seq G}{P \seq G\theta} 

    \end{array}
			\]
\vspace{4mm}
\hrule
		\caption{Sequent Calculus \system{LB}}
		\label{fig:blp}
	\end{figure}

The operational semantics of BLP is as follows: Beginning with the query $P \seq G$, one gives a candidate $\theta$ and checks by reducing in \system{LB} (in the sense of Section~\ref{sec:constraint:reductivelogic}) whether or not $P \proves G\theta$ obtains or not. The first step (i.e., introducing $\theta$) is also captured by a reduction operator in \system{LB} --- namely, the $\rrn \exists$-rule --- hence, the operational semantics is entirely based on reduction. 


    The following example, also appearing in Gheorghiu et al.~\cite{Samsonschrift}, illustrates how BLP works:
	
	\begin{figure}[t]
		\centering
		\begin{tabular}{||c c c||} 
			\hline
			Red & Green & Blue \\ [0.5ex] 
			\hline\hline
			Al(gebra) & Lo(gic) & Da(tabases) \\ [0.5ex] 
			\, Pr(obability) \, & \, Ca(tegories) \, & \,  Co(mpilers) \, \\ [0.5ex] 
			Gr(aphs) & Au(tomata) & AI \\
			\hline
		\end{tabular}
		\caption{The Extensional Database}
		\label{fig:uucs}
	\end{figure} 

\begin{Example}\label{ex:student} 

To complete the informatics course at Unseen University, students must 
pick one module for each of the three terms of the year, which are called R(ed), G(reen), and B(lue), respectively. The available choices are shown in Figure~\ref{fig:uucs}. More formally, we have a relation $S$ that obtains for valid selections of courses; for example, $S(Al,Ca,Co)$ obtains, but $S(Pr,Gr,Ca)$ does not.

We may use BLP to capture this situation. The setup is captured by a program $P$ composed of two parts: the \emph{extensional} database $ED$ and the \emph{intensional} database $ID$. The extensional database contains the information about the modules:
\[
ED  := R(Al),R(Pr),R(Gr), G(Lo), G(Ca), G(Au), B(Da) , B(Co), B(AI)
\]
Meanwhile, the intensional database comprises the selection logic:
\[
ID :=   R(x)\land G(y) \land B(x)  \to  S(x,y,z) 
\] 
To find the possible combinations of modules, one queries the system for different choices of $M_1$ and 	$M_2$ and $M_3$; that is, one considers the validity of the following sequent:
	\[
	P \seq S(M_1,M_2,M_3) 
	\]
	One possible execution is the following, in which $\phi := S(Al,Lo,AI) \leftarrow (R(Al)\land G(Lo)) \land B(AI)$:
	
	\begin{scprooftree}{0.9}
		\AxiomC{}
		\RightLabel{$\ax$}
		\UnaryInfC{ $P \fatcomma\phi \seq R(Al) $ }
    		\AxiomC{}
    		\RightLabel{$\ax$}
    		\UnaryInfC{ $P \fatcomma \phi \seq G(Lo) $ }
        		\AxiomC{}
        		\RightLabel{$\ax$}
        		\UnaryInfC{ $P \fatcomma \phi \seq B(AI) $ }
		\RightLabel{$\rrn \land$}
		\doubleLine
		\TrinaryInfC{ $P \fatcomma \phi \seq (R(Al)\land G(Lo)) \land B(AI)$  }
		\RightLabel{$\lrn \to$}
		\UnaryInfC{ $P \fatcomma \phi  \seq S(Al,Lo,AI)$ }
		\RightLabel{$\lrn \forall$}
		\UnaryInfC{ $P \seq S(Al,Lo,AI)$ }
		\doubleLine 
		\RightLabel{$\rrn \exists$}
		\UnaryInfC{ $P \seq S(x,y,z)$ }
	\end{scprooftree}
	
	There are $9^3 = 729$ possible ways of selecting three courses out of the nine, and there are only $27$ acceptable choices out of these possibilities. In terms of proof-search spaces, there are $729$ branches, out of which only $27$ may lead to successful reductions. Random selection is, therefore, intractable. Moreover, to the extent that logic programming concerns problem solving (see, for example, Kowalski~\cite{Kowalski1986}), this process does not reflect how one hopes a student would reason.
\end{Example}

Example~\ref{ex:student} is naive in that, rather than simply guessing a possibility, a student would likely observe that one must choose three modules one from each $R$, $G$, and $B$, respectively. How may we capture such reasoning? One possibility is to keep the major logical steps in Example~\ref{ex:student}, but without committing to a particular substitution. Instead, the proof structures thus constructed inform us what kinds of substitutions make sense, precisely those that render the proof structure a \system{LB}-proof. We may capture this approach with algebraic constraints. 
	
	\subsection{Unification via Nominal Constraints} \label{sec:ext:lp:alp}
	
	A unification algebra allows the substitution in \system{LB} to be managed intelligently. It does this by managing unification, as opposed to merely managing substitution, to be expressed logically and enforced at the end of the computation. We only need a way to track what substitution took place and what needs to be unified, which labels governed by equality can handle.
	
		\begin{Definition}[Unification Algebra]
		The unification algebra $\mathcal{U}$ consists of the universe $\set{TERM}{\lang{A}}$ and nothing else.
		\end{Definition}
		
The notion of \emph{enriched sequents} in this setting extends that for propositional logics in that terms within formulae may also carry labels. For example, $P \seq R(x \cdot n_1, t_1 \cdot n_2)$, in which $R \in \set{R}$ is a relation-symbol, $x \in \set{V}$ and $t_1 \in \set{TERM}(\lang{A})$ are terms, and $n_1$ and $n_2$ are expressions for the algebra (i.e., variables), is an enriched sequent.  

	We have no operators in the algebra, but we do have an equality predicate for the constraints. For example, we may write $n=k$, in which $n$ is a variable and $k$ is (interpreted) as a term, to denote that the value of $n$ must be whatever is denoted by $k$. Let $A$ and $B$ be enriched atoms. We write $A \equiv B$ to denote the meta-disjunction of all equations $n = m$ such that $n$ is the label of a term in $A$ occurring in $P$ and $m$ is a label of a term in an atom $B$ such that $A$ and $B$ have the same relation-symbol and $n$ and $m$ occur on corresponding terms. For example, $R(x_1\cdot n_1, y_1 \cdot n_2) \equiv R(x_2\cdot m_1, y_2 \cdot m_2)$ denotes $(n_1 = m_1) \parr (n_2 = m_2)$. The notation denotes falsum when the disjunction is empty (i.e., when there are no correspondences). Let $P$ be a program. We write $\parr_{A \in P \text{ st.} A \sim B} (A \equiv B)$ to denote the disjunction $A \equiv B$ for all $A$ in $P$. We may also write $n \in \{k_1,...,k_n\}$ to denote the disjunction $(n=k_1) \parr ... \parr (n = k_n)$.

	\begin{Definition}[Constraint System $\system{PLB} \oplus \alg{U}$]
	Constraint system $\system{LB}\oplus \alg{U}$ comprises the rules in Figure~\ref{fig:alp}, in which $\theta$ denotes a substitution that always introduce fresh labels. 
	\end{Definition}
		\begin{figure}[t]
   \hrule
 \vspace{4mm}
		\[
  \begin{array}{c}
				\infer[\rrn \to]{P \seq D \to G}{P\fatcomma D \seq G} \quad 
				\infer[\rrn \land]{P \seq G_0 \land G_1}{P \seq G_0 & P \seq G_1}
				\quad	
				\infer[\rrn \lor]{P \seq G_0 \lor G_1}{P \seq G_i}
				\quad
					\infer[\ax]{P \seq B}{\parr_{A \in P \text{ st.} A \sim B} (A \equiv B)}
			
				\\[1.5ex]
			
				\infer[\lrn \to ]{ P\fatcomma G \to A \seq B}{P\fatcomma G \to A \seq G & A\equiv B}
				\quad
				\infer[\lrn \land ]{ P\fatcomma D_0 \land D_1 \seq A}{P\fatcomma D_0\fatcomma D_1 \seq A} 
				\quad 
				\infer[\lrn \forall ]{ P\fatcomma \forall x D \seq A}{P\fatcomma D\fatcomma D\theta \seq A}
				\quad
				\infer[\rrn \exists]{P \seq G}{P \seq G\theta}
    \end{array}
    \]
    \vspace{4mm}
	\hrule
 \caption{Constraint System $\system{PLB} \oplus \alg{U}$}
		 \label{fig:alp}
	\end{figure}

	\begin{Proposition}
	   Constraint system $\system{PLB} \oplus \alg{U}$, with valuation $\sigma$, is faithful and adequate with respect to $\system{LB}$.
	\end{Proposition}

We use $\system{LB}\oplus \alg{U}$ to simulate the actual reasoning one intends BLP to represent. To see this, we return to Example~\ref{ex:student} and show that the constraint system captures the reasoning process one hopes a student would use more realistically.  

\begin{Example}[Example~\ref{ex:student} cont'd]
To simplify notation, let $ID \cdot [n_1,n_2,n_3]$ denote $(R(x.n_1)\land G(y.n_2)) \land B(x.n_3) \to  CS(x.n_1,y.n_2,z.n_3)$. The computation-trace in BLP using the algebraic constraint system $\system{LB}\oplus \alg{U}$ is as follows:
	
	\begin{scprooftree}{0.9}
		\AxiomC{$n_1  \in \{Al, Pr,Gr\}$}
		\UnaryInfC{ $P \seq R(x\cdot n_1) $ }
		\AxiomC{$n_2 \in \{ Lo,Ca,Au \} $}
		\UnaryInfC{ $P \seq G(y\cdot n_2) $ }
		\BinaryInfC{ $P \seq R(x\cdot n_1) \land G(y\cdot n_2) $ }
		\AxiomC{$ n_3 \in \{ Da,Co,AI \} $}
		\UnaryInfC{ $P \seq B(z\cdot n_3) $ }
		\BinaryInfC{ $P \seq (R(x\cdot n_1)\land G(y\cdot n_2)) \land B(z\cdot n_3)$  }
		\AxiomC{$n_1=m_1$}
		\noLine
		\UnaryInfC{$n_2=m_2$}
		\noLine
		\UnaryInfC{$n_3=m_3$}
		\BinaryInfC{ $P, ID\cdot[n_1,n_2,n_3] \seq S(x.m_1,y.m_2,z.m_3)$  }
		\UnaryInfC{ $P \seq CS(x \cdot m_1,y.m_2,z.m_3)$ }
		\UnaryInfC{ $P \seq S(x,y,z)$ }
	\end{scprooftree} 
	
    \noindent Every valid execution in \system{LB} of the initial query is a coherent instantiation of this proof; for example, to recover the 
	one presented in Example~\ref{ex:student}, we use the following interpretation:
	\[
	I(n_1):=I(m_1) := Al \quad I(n_2):=I(m_2) := Lo \quad I(n_3):=I(m_3): = AI
	\]
\end{Example}

This study of logic programming illustrates that algebraic constraint systems extend quite naturally to the predicate logic setting and have practical applications, as showcased by how it addresses unification in logic programming. Of course, this section is cursory compared to the extensive study of propositional logic in the rest of the paper, with the open question of a general theory.  

\section{Conclusion} \label{sec:conclusion}

This paper has introduced the concept of a constraint system, which serves as a uniform tool for studying the metatheory of one logic in terms of the metatheory of another. The advantage is that the latter may be simpler or more well-understood in some practical sense. In short, a constraint system is a labelled sequent calculus in which the labels carry an algebraic structure to determine correctness conditions on proof structures. A motivating example of a class of constraint system already present in the literature are those captured by the \emph{resource-distribution via Boolean constraints} (RDvBC) mechanism by Harland and Pym~\cite{Harland1997,Pym2001dis}, which serve as a meta-theoretic tool to analyze the possible context-management strategies during proof-search in logics with multiplicative connectives --- see Section~\ref{sec:ex:bi}.

Constraint systems have two possible relationships with a logic of interest: soundness and completeness, and faithfulness and adequacy. The former is a \emph{global} correctness condition that says that completed reductions in the constraint system (i.e., constructions to which one cannot apply further reduction operators) whose constraints are coherent (i.e., admit a solution) characterize the consequence of a logic. Meanwhile, the latter is a \emph{local} correctness condition in that each reduction step in the constraint system corresponds to a valid inference for the logic when its constraints are satisfied; consequently, a completed reduction corresponds to a proof in a sequent calculus for the logic. Both correctness criteria are valuable in applications of constraint systems for studying meta-theory.

We illustrated the framework's usefulness in studying metatheory through various examples beyond RDvBC. First, in Section~\ref{sec:ex:relationalcalculi}, we show that they yield a general, uniform, and systematic process for generating relational calculi after Negri~\cite{negri2001,negri2003contraction,Negri2005}; moreover, that relative to this theory of relational calculi, one has an approach to proving soundness and completeness for model-theoretic semantics (M-tS) in the sense of Kripke~\cite{Kripke1963,kripke1965semantical} (see also Beth~\cite{Beth1955}) that entirely bypasses term- and counter-model constructions (see, for example, van Dalen~\cite{vanDalen}). An example of this method in practice has already been given in the case of the logic of Bunched Implications by Gheorghiu and Pym~\cite{Alex:BI_Semantics}. Second, in Section~\ref{sec:ex:ipl}, we show by a case-study for IPL that constraint systems enable one to synthesize an M-tS for a logic from an analysis of its proof theory. Overall, these examples witness that constraint systems help bridge the gap between semantics and proof theory for logics. It needs to be clarified how this paper's approach to soundness and completeness relates to the more traditional approaches. However, such investigation may aid in understanding the underlying principles and is left for future work.

While this paper is only concerned with propositional logics, we showed in Section~\ref{sec:ext:fol} how it could intuitively be extended to the first-order setting. Moreover, there are good computational reasons for doing such extensions in the context of logic as a reasoning technology.

Of course, this paper concerns only the initial framework for algebraic constraint systems. Future work includes giving more examples of constraint systems and developing the applications presented in this paper. For example, the case analysis of deriving semantics for IPL should be repeated for other adjacent logics, especially intermediate, hybrid, and substructural logics. Moreover, in this paper, we have only considered three different notions of algebra ---  Boolean algebra, world algebra (i.e., the algebra corresponding to a frame), and unification algebra --- what are some other valuable algebras and constraint systems, and what can they tell us about the logics being studied? 

Overall, constraint systems provide a general framework for defining and studying logics and have the potential to bridge the gap between model theory and proof theory, as well as give insights about proof-search.

\subsection*{Acknowledgments}
We are grateful to Tao Gu and the referees on an earlier version for their thorough and thoughtful comments on this work.

\bibliographystyle{asl}
\bibliography{bib}

@book{vanDalen,
 author = {Dirk van Dalen},
 title = {{Logic and Structure}},
 publisher = {Springer},
 series = {Universitext},
 year = {2012}
}

@article{Beth1955,
author = {Beth, Evert W.},
title = {{Semantic Construction of Intuitionistic Logic}},
journal = {Indagationes Mathematicae},
volume = {17},
number = {4},
pages = {pp. 327--338},
year = {1955},
doi = {10.1016/S1385-7258(55)50014-6},
url = {https://doi.org/10.1016/S1385-7258(55)50014-6},
issn = {1385-7258}
}

@incollection{kripke1965semantical,
  title={{Semantical Analysis of Intuitionistic Logic I}},
  author={Kripke, Saul A},
  booktitle={{Studies in Logic and the Foundations of Mathematics}},
  volume={40},
  pages={92--130},
  year={1965},
  publisher={Elsevier}
}

@article{CIABATTONI2012,
title = {Algebraic proof theory for substructural logics: Cut-elimination and completions},
journal = {Annals of Pure and Applied Logic},
volume = {163},
number = {3},
pages = {266-290},
year = {2012},
issn = {0168-0072},
doi = {https://doi.org/10.1016/j.apal.2011.09.003},
author = {Agata Ciabattoni and Nikolaos Galatos and Kazushige Terui},
}

@inproceedings{Ciabattoni2008,
author = {Agata Ciabattoni and Nikolaos Galatos and Kazushige Terui},
title = {From Axioms to Analytic Rules in Nonclassical Logics},
year = {2008},
isbn = {9780769531830},
publisher = {IEEE Computer Society},
booktitle = {{Logic in Computer Science --- LICS}},
url = {https://doi.org/10.1109/LICS.2008.39},
doi = {10.1109/LICS.2008.39},
}

@InProceedings{Ciabattoni2009,
author="Ciabattoni, Agata
and Stra{\ss}burger, Lutz
and Terui, Kazushige",
editor="Gr{\"a}del, Erich
and Kahle, Reinhard",
title="Expanding the Realm of Systematic Proof Theory",
booktitle={{Computer Science Logic --- CSL}},
year="2009",
publisher="Springer",
pages="163--178",
isbn="978-3-642-04027-6"
}

@article{Marin2022,
title = {From Axioms to Synthetic Inference Rules via Focusing},
journal = {Annals of Pure and Applied Logic},
volume = {173},
number = {5},
pages = {103091},
year = {2022},
issn = {0168-0072},
doi = {https://doi.org/10.1016/j.apal.2022.103091},
author = {Sonia Marin and Dale Miller and Elaine Pimentel and Marco Volpe},
}

@article{LIANG20094747,
title = {Focusing and Polarization in {L}inear, {I}ntuitionistic, and {C}lassical {L}ogic},
journal = {Theoretical Computer Science},
volume = {410},
number = {46},
pages = {4747-4768},
year = {2009},
issn = {0304-3975},
doi = {https://doi.org/10.1016/j.tcs.2009.07.041},
author = {Chuck Liang and Dale Miller}
}

@article{ohearn1999logic,
  title={{The Logic of Bunched Implications}},
  author={O'Hearn, Peter W and Pym, David J},
  journal={Bulletin of Symbolic Logic},
  volume={5},
  number={2},
  pages={215--244},
  year={1999},
  publisher={Cambridge University Press}
}

@book{kleene2002mathematical,
  title={{Mathematical Logic}},
  author={Kleene, Stephen Cole},
  year={2002},
  publisher={Dover Publications}
}

@phdthesis{baldoni1998normal,
  title={{Normal Multimodal Logics: Automatic Deduction and Logic Programming Extension}},
  author={Baldoni, Matteo },
  year={1998},
  school={Universit{\`a} degli Studi di Torino}
}

@inproceedings{massacci1994strongly,
  title={{Strongly Analytic Tableaux for Normal Modal Logics}},
  author={Massacci, Fabio},
  booktitle={Conference on Automated Deduction --- CADE},
  volume={12},
  pages={723--737},
  year={1994}
}

@article{Kripke1963,
	title={{Semantical Analysis of Modal Logic I: Normal Modal Propositional Calculi}},
	author={Kripke, Saul A},
	journal={FoSSaCSLogic Quarterly},
	volume={9},
	number={5-6},
	pages={67--96},
	year={1963},
	publisher={WILEY-VCH Verlag Berlin GmbH Berlin}
}

@book{Blackburn2001, 
	series={Cambridge Tracts in Theoretical Computer Science}, 
	title={{Modal Logic}}, 
	DOI={10.1017/CBO9781107050884}, 
	publisher={Cambridge University Press}, 
	author={Blackburn, Patrick and Rijke, Maarten de and Venema, Yde}, 
	year={2001}, 
	collection={Cambridge Tracts in Theoretical Computer Science}
}

@misc{Gabbay1993,
	author={Dov M. Gabbay},
	title={{Classical vs non-Classical Logics: The Universality of Classical Logic (MPI-I-93-230)}},
	type = {Technical Report},
	institution = {Saarbrücken: Max-Planck-Institut f{\"u}r Informatik},
	date = {1993},
}

@phdthesis{Russo1996,
	title={Modal logics as labelled deductive systems},
	author={Russo, Alessandra Maria},
	year={1996},
	school={Department of Computing, Imperial College London}
}

@book{Gabbay1996,
	year = {1996},
	publisher = {Oxford University Press},
	title = {{Labelled Deductive Systems}},
	author = {Dov M. Gabbay}
}

@book{Read1988,
	year = {1988},
	author = {Stephen Read},
	title = {{Relevant Logic}},
	publisher = {Basil Blackwell}
}

@InProceedings{Alex2021,
author={Gheorghiu, Alexander
and Marin, Sonia},
editor={Kiefer, Stefan
and Tasson, Christine},
title={{Focused Proof-search in the Logic of Bunched Implications}},
booktitle={{Foundations of Software Science and Computation Structures --- FoSSaCS}},
year={2021},
publisher={Springer},
pages={247---267},
}

@inproceedings{marx2000mosaic,
  title={{The Mosaic Method for Temporal Logics}},
  author={Marx, Maarten and Mikul{\'a}s, Szabolcs and Reynolds, Mark},
  booktitle={{Automated Reasoning with Analytic Tableaux and Related Methods --- TABLEAUX}},
  volume = 9,
  pages={324--340},
  year={2000},
  organization={Springer}
}

@InCollection{SchroedHeister,
	author       =	{Schroeder-Heister, Peter},
	title        =	{{Proof-Theoretic Semantics}},
	booktitle    =	{The {Stanford} {E}ncyclopedia of {P}hilosophy},
	editor       =	{Edward N. Zalta},
	howpublished =	{\url{https://plato.stanford.edu/archives/spr2018/entries/proof-theoretic-semantics/}},
	year         =	{2018},
	edition      =	{{S}pring 2018},
	publisher    =	{Metaphysics Research Lab, Stanford University}
}

@article{Milner1984,
  title={{The Use of Machines to Assist in Rigorous Proof}},
  author={Milner, Robin},
  journal={Philosophical Transactions of the Royal Society of London. Series A, Mathematical and Physical Sciences},
  volume={312},
  number={1522},
  pages={411--422},
  year={1984},
  publisher={The Royal Society London}
}

@book{troelstra2000basic,
  title={{Basic Proof Theory}},
  author={Troelstra, Anne S. and Schwichtenberg, Helmut},
  volume={43},
  series = {Cambridge Tracts in Theoretical Computer Science},
  year={2000},
  publisher={Cambridge University Press}
}

@book{Dummett2000,
	title={{Elements of Intuitionism}},
	author={Dummett, Michael},
	volume={39},
	year={2000},
	series = {Oxford Logic Guides},
	publisher={Oxford University Press}
}

@InProceedings{Docherty2018,
	author={Docherty, Simon
	and Pym, David J.},
	editor={Baier, Christel
	and Dal Lago, Ugo},
	title={{Modular Tableaux Calculi for Separation Theories}},
	booktitle={{Foundations of Software Science and Computation Structures --- FoSSaCS}},
	year={2018},
	pages={441--458}
}

@book{Pym2004,
	title={{Reductive Logic and Proof-search: Proof Theory, Semantics, and Control}},
	author={David J. Pym and Ritter, Eike},
	isbn={9780198526339},
	lccn={2004049208},
	series={Oxford Logic Guides},
	year={2004},
	publisher={Clarendon Press}
}

@PhdThesis{Hodas1994thesis,
	author = {Joshua S. Hodas},
	title = {{Logic Programming in Intutionistic Linear Logic: Theory, Design and Implementation}},
	school = {University of Pennsylvania},
	year = {1994}
}

@article{Cervesato2000,
	title={{Efficient Resource Management for Linear Logic Proof Search}},
	author={Cervesato, Iliano and Hodas, Joshua S and Pfenning, Frank},
	journal={Theoretical Computer Science},
	volume={232},
	number={1-2},
	pages={133--163},
	year={2000},
	publisher={Elsevier}
}

@inproceedings{Lopez2005,
	title={{Implementing Efficient Resource Management for Linear Logic Programming}},
	author={L{\'o}pez, Pablo and Polakow, Jeff},
	booktitle={{International Conference on Logic for Programming, Artificial Intelligence, and Reasoning --- LPAR}},
	pages={528--543},
	year={2005},
	organization={Springer}
}

@inproceedings{Harland1997,
	title={{Resource-distribution via Boolean Constraints}},
	author={James Harland and David J. Pym},
	booktitle={{Automated Deduction --- CADE}},
	pages={222--236},
	year={1997},
	publisher={Springer},
	editor={McCune, William}
}

@incollection{Gentzen1969,
	publisher = {North-Holland Publishing Company},
	year = {1969},
	author = {Gerhard Gentzen},
	editor = {Manfred E. Szabo},
	booktitle = {{The Collected Papers of Gerhard Gentzen}},
	title = {{Investigations into Logical Deduction}}
}

@incollection{Godel,
editor = {Feferman, Solomon and Dawson, John W. and Kleene, Stephen C. and Moore, Gregory H. and Solovay, Robert M. and van Heijenoort, Jean},
author = {Kurt G\"odel},
booktitle = {{Kurt G\"odel: Collected Works: Publications 1929-1936}},
title = {{On the Completeness of the Calculus of Logic}},
volume = {1},
year = {1986},
publisher = {Oxford University Press},
}

@inproceedings{docherty2019non,
  title={{A Non-wellfounded, Labelled Proof System for Propositional Dynamic Logic}},
  author={Docherty, Simon and Rowe, Reuben NS},
  booktitle={{Automated Reasoning with Analytic Tableaux and Related Methods --- TABLEAUX}},
  pages={335--352},
  year={2019},
  organization={Springer}
}

@article{Miller91,
	title = {{Uniform Proofs as a Foundation for Logic Programming}},
	journal = {Annals of Pure and Applied Logic},
	volume = {51},
	number = {1},
	pages = {125 - 157},
	year = {1991},
	author = {Dale Miller and Gopalan Nadathur and Frank Pfenning and Andre Scedrov},
}

@article{negri2003contraction,
  title={{Contraction-free Sequent Calculi for Geometric Theories with an Application to Barr's theorem}},
  author={Negri, Sara},
  journal={Archive for Mathematical Logic},
  volume={42},
  number={4},
  pages={389--401},
  year={2003},
  publisher={Springer}
}

@article{Pym2003,
	title={{Resource-distribution via Boolean Constraints}},
	author={James Harland and David J. Pym},
	year={2003},
	language={English},
	volume={4},
	pages={56--90},
	journal={ACM Transactions on Computational Logic},
	issn={1529-3785},
	publisher={Association for Computing Machinery},
	number={1}
}

@inproceedings{Pym2005games,
  author    = {David J. Pym and Ritter, Eike},
  editor    = {Dan R. Ghica and
               Guy McCusker},
  title     = {{A Games Semantics for Reductive Logic and Proof-search}},
  booktitle = {{Proceedings of the 1st Workshop on Games for Logic and Programming Languages --- GALOP@ETAPS}},
  pages     = {107--123},
  year      = {2005},
}

@article{Kowalski1971,
	author = {Kowalski, Robert and Kuehner, Donald},
	year = {1971},
	month = {12},
	pages = {227-260},
	title = {{Linear Resolution with Selection Function}},
	volume = {2},
	isbn = {978-3-642-81957-5},
	journal = {Artificial Intelligence},
	doi = {10.1016/0004-3702(71)90012-9}
}

@book{Kowalski1986,
	author = {Kowalski, Robert},
	title = {{Logic for Problem-Solving}},
	year = {1986},
	isbn = {04440036597},
	publisher = {North-Holland Publishing Co.},
}

@PhdThesis{Docherty2019,
	author = {Simon Docherty},
	title = {{Bunched Logics: A Uniform Approach}},
	school = {University College London},
	year = {2019},
}

@book{Fitting1998,
	title={{First-order Modal Logic}},
	author={Fitting, Melvin and Mendelsohn, Richard L.},
	year={1998},
	publisher={Springer, Dordrecht},
	doi = {https://doi.org/10.1007/978-94-011-5292-1},
	isbn = {978-0-7923-5335-5}
}

@book{Fitting1983,
	title = {{Proof Methods for Modal and Intuitionistic Logics}},
	author = {Fitting, Melvin},
	doi = {10.1007/978-94-017-2794-5},
	isbn = {978-90-277-1573-9},
	publisher ={Springer},
	year = {1983}
}

@article{Negri2005,
	title={{Proof Analysis in Modal Logic}},
	author={Negri, Sara},
	journal={Journal of Philosophical Logic},
	volume={34},
	number={5},
	pages={507--544},
	year={2005},
	publisher={Springer}
}

@incollection{Samsonschrift,
	author = {Alexander V. Gheorghiu and Simon Docherty and David J. Pym},
	title = {{Reductive Logic, Coalgebra, and Proof-search: A Perspective from Resource Semantics}},
	booktitle = {{Samson Abramsky on Logic and Structure in Computer Science and Beyond}},
	series = {Springer Outstanding Contributions to Logic Series},
	editor = {A. Palmigiano and M. Sadrzadeh},
	note ={to appear},
	year = {2021},
	publisher = {Springer}
}

@article{Miller1989,
	title = {{A Logical Analysis of Modules in Logic Programming}},
	journal = {The Journal of Logic Programming},
	volume = {6},
	number = {1},
	pages = {79-108},
	year = {1989},
	issn = {0743-1066},
	doi = {https://doi.org/10.1016/0743-1066(89)90031-9},
	url = {https://www.sciencedirect.com/science/article/pii/0743106689900319},
	author = {Dale Miller},
}

@article{Pym2001dis,
author = {David J. Pym and Ritter, Eike},
year = {2001},
pages = {315-338},
title = {{On the Semantics of Classical Disjunction}},
volume = {159},
journal = {{Journal of Pure and Applied Algebra}}
}

@article{Galmiche2005,
author = {Galmiche, Didier and M\'{e}ry, Daniel and Pym, David},
title = {{The Semantics of BI and Resource Tableaux}},
year = {2005},
publisher = {Cambridge University Press},
address = {USA},
volume = {15},
number = {6},
journal = {Mathematical Structures in Computer Science},
month = dec,
pages = {1033–1088},
numpages = {56}
}

@article{hodas1994logic,
  title={{Logic Programming in a Fragment of Intuitionistic Linear Logic}},
  author={Hodas, Joshua S and Miller, Dale},
  journal={{Information and Computation}},
  volume={110},
  number={2},
  pages={327--365},
  year={1994},
  publisher={Elsevier}
}

@inproceedings{winikoff1995implementing,
  title={{Implementing the Linear Logic Programming Language Lygon}},
  author={Winikoff, Michael and Harland, James},
  booktitle={{Proceedings of
the International Logic Programming Symposium --- ILPS}},
  pages={66--80},
  year={1995},
  publisher = {MIT Press}
}

@book{negri2001,
place={Cambridge},
title={{Structural Proof Theory}}, 
publisher={Cambridge University Press}, 
author={Negri, Sara and von Plato, Jan}, 
year={2001}
}

@article{Alex:BI_Semantics,
author = {Gheorghiu, Alexander V. and Pym, David J.},
title = {{Semantical Analysis of the Logic of Bunched Implications}},
journal = {Studia Logica},
year = {2023},
volume = {},
number = {},
pages = {},
month = {},
doi = {10.1007/s11225-022-10028-z},
issn = {1572-8730}
}

@book{Bundy1983,
author = {Bundy, Alan},
title = {{The Computer Modelling of Mathematical Reasoning}},
year = {1985},
isbn = {0121412504},
publisher = {Academic Press Professional, Inc.},
}

@InCollection{SEP-PtS,
	author       =	{Schroeder-Heister, Peter},
	title        =	{{Proof-Theoretic Semantics}},
	booktitle    =	{{The Stanford Encyclopedia of Philosophy}},
	editor       =	{Edward N. Zalta},
	howpublished =	{\url{https://plato.stanford.edu/archives/spr2018/entries/proof-theoretic-semantics/}},
	year         =	{2018},
	edition      =	{{S}pring 2018},
	publisher    =	{Metaphysics Research Lab, Stanford University}
}

@book{kanger1957provability,
  title={{Provability in Logic}},
  author={Kanger, Stig},
  year={1957},
  publisher = {Almqvist \& Wiksell}
}

@article{kushida2003proof,
  title={A proof-theoretic study of the correspondence of classical logic and modal logic},
  author={Kushida, Hirohiko and Okada, Mitsu},
  journal={The Journal of Symbolic Logic},
  volume={68},
  number={4},
  pages={1403--1414},
  year={2003},
  publisher={Cambridge University Press}
}

@article{mints1968cut,
  title={Cut-free calculi of the {S5}-type},
  author={Mints, Grigori},
  journal={Zapiski Nauchnykh Seminarov --- LOMI},
  volume={8},
  pages={166--174},
  year={1968},
  language={Russian}
}

@book{indrzejczak2021sequents,
  title={{Sequents and Trees}},
  author={Indrzejczak, Andrzej},
  series={Studies in Universal Logic},
  year={2021},
  publisher={Springer}
}

@article{mints1997indexed,
  title={Indexed systems of sequents and cut-elimination},
  author={Mints, Grigori},
  journal={Journal of Philosophical Logic},
  volume={26},
  pages={671--696},
  year={1997},
  publisher={Springer}
}

@phdthesis{simpson1994proof,
  title={{The Proof Theory and Semantics of Intuitionistic Modal Logic}},
  author={Simpson, Alex K},
  year={1994},
  school={University of Edinburgh}
}

@book{vigano2013labelled,
  title={Labelled Non-classical Logics},
  author={Vigan\`o, Luca},
  year={2013},
  publisher={Springer}
}

@article{basin1998natural,
  title={{Natural Deduction for Non-classical Logics}},
  author={Basin, David and Matthews, Se{\'a}n and Vigan{\`o}, Luca},
  journal={Studia Logica},
  volume={60},
  pages={119--160},
  year={1998},
  publisher={Springer}
}

@article{galmiche2010tableaux,
  title={Tableaux and resource graphs for separation logic},
  author={Galmiche, Didier and M{\'e}ry, Daniel},
  journal={Journal of Logic and Computation},
  volume={20},
  number={1},
  pages={189--231},
  year={2010},
  publisher={Oxford University Press}
}

@article{galmiche2003semantic,
  title={{Semantic Labelled Tableaux for Propositional {BI}$^\bot$}},
  author={Galmiche, Didier and M{\'e}ry, Daniel},
  journal={Journal of Logic and Computation},
  volume={13},
  number={5},
  pages={707--753},
  year={2003},
  publisher={Oxford University Press}
}

@inproceedings{galmiche2002resource,
  title={{Resource Tableaux}},
  author={Galmiche, Didier and M{\'e}ry, Daniel and Pym, David},
  booktitle={{Computer Science Logic --- CSL}},
  pages={183--199},
  year={2002},
  organization={Springer}
}

@article{catach1991tableaux,
  title={{TABLEAUX: A General Theorem Prover for Modal Logics}},
  author={Catach, Laurent},
  journal={{Journal of Automated Reasoning}},
  volume={7},
  pages={489--510},
  year={1991},
  publisher={Springer}
}

@article{negri1998cut,
  title={{Cut Elimination in the Presence of Axioms}},
  author={Negri, Sara and Von Plato, Jan},
  journal={{Bulletin of Symbolic Logic}},
  volume={4},
  pages={418--435},
  year={1998},
  publisher={Cambridge University Press}
}

@ARTICLE{blackburn2000,
  author={Blackburn, P},
  journal={{Journal of Logic and Computation}}, 
  title={{Internalizing Labelled Deduction}}, 
  year={2000},
  volume={10},
  pages={137-168},
  doi={10.1093/logcom/10.1.137}}

@incollection{Areces,
  title={{Hybrid Logic}},
  author = {Carlos Areces, Balder ten Cate},
  editor ={Blackburn, Patrick and Marx, Maarten},
  booktitle={{Handbook of Modal Logic}},
  pages={893--939},
  year={2007},
  publisher={Elsevier}
}

@article{blackburn1995hybrid,
  title={{Hybrid Languages}},
  author={Blackburn, Patrick and Seligman, Jerry},
  journal={{Journal of Logic, Language and Information}},
  volume={4},
  pages={251--272},
  year={1995},
  publisher={Springer}
}

@book{brauner2001hybrid,
  title={{Hybrid Logic and its Proof-Theory}},
  author={Bra{\"u}ner, Torben},
  year={2011},
  publisher={Springer}
}

@book{indrzejczak2011natural,
  title={{Natural Deduction, Hybrid Systems and Modal Logics}},
  author={Indrzejczak, Andrzej},
  year={2010},
  publisher={Springer}
}

\end{document}